\documentclass{amsart}

\usepackage{amssymb}
\usepackage{amsfonts}
\usepackage{amsmath}
\usepackage{amsthm}
\usepackage{amsxtra}
\usepackage{hyperref}
\usepackage{subfigure}
\usepackage{cite}
\usepackage{graphicx}
\usepackage{paralist}
\usepackage{fullpage}

\numberwithin{equation}{section}

\newcommand{\cc}{\mathrm{c.c.}}
\newcommand{\bigo}{\mathcal{O}}

\newcommand{\abs}[1]{\left\vert#1\right\vert}

\newcommand{\paren}[1]{\left(#1\right)}
\newcommand{\bracket}[1]{\left[#1\right]}
\newcommand{\set}[1]{\left\{#1\right\}}
\newcommand{\mean}[1]{\left\langle#1\right\rangle}
\newcommand{\eps}{\epsilon}

\newcommand{\dz}{\partial_{z} }
\newcommand{\dZ}{\partial_{Z} }

\newcommand{\dt}{\partial_{t}}
\newcommand{\dT}{\partial_{T}}

\newcommand{\f}{\phi}
\newcommand{\dphi}{{\partial_\f}}

\newcommand{\nn}{\nonumber}

\newcommand{\R}{\mathbb{R}}

\newcommand{\Azero}{A^{(0)}}

\newcommand{\Ep}{E^+}
\newcommand{\Em}{E^-}
\newcommand{\sech}{\mathrm{sech}}
\newcommand{\sign}{\mathrm{sign}}

\newtheorem{thm}{Theorem}[section]

\newtheorem{prop}[thm]{Proposition}

\begin{document}


\author[Simpson] {Gideon Simpson} 
\email{simpson@math.toronto.edu}
\address{Mathematics Department, University of Toronto \\
  Toronto, Ontario, Canada}

\author[Weinstein]{M.I. Weinstein}
\email{miw2103@columbia.edu}
\address{Applied Mathematics Department, Columbia University \\
  New York City, NY 10027, USA}

\title[Nonlinear Periodic Maxwell]{Coherent Structures and Carrier Shocks in the Nonlinear
  Periodic Maxwell Equations}

\begin{abstract}
We consider the one-dimensional propagation of
  electromagnetic waves in a weakly nonlinear and low-contrast
  spatially inhomogeneous medium with no energy dissipation.
   We focus on the case of a periodic medium, in which dispersion 
   enters only through the (Floquet-Bloch) spectral band dispersion associated
    with the periodic structure; chromatic dispersion ( time-nonlocality of the 
    polarization) 
     is neglected. Numerical
  simulations show that for initial conditions of wave-packet type (a
  plane wave of fixed carrier frequency multiplied by a slow varying,
  spatially localized function) very long-lived spatially localized
  coherent soliton-like structures emerge, whose character is that of
  a slowly varying envelope of a train of shocks. We call this structure an {\it envelope 
  carrier-shock train}.
  
  The structure of the solution violates the oft-assumed nearly
  monochromatic wave packet structure, whose envelope is governed by
  the nonlinear coupled mode equations (NLCME). The inconsistency and
  inaccuracy of NLCME lies in the neglect of all (infinitely many)
  resonances except for the principle resonance induced by the initial
  carrier frequency.  We derive, via a nonlinear geometrical optics
  expansion, a system of nonlocal integro-differential equations
  governing the coupled evolution of backward and forward propagating
  waves. These equations incorporate effects of {\it all} resonances.
  In a periodic medium, these equations may be expressed as a system
  of infinitely many coupled mode equations, which we call the
  extended nonlinear coupled mode system (xNLCME). Truncating xNLCME
  to include only the principle resonances leads to the classical
  NLCME.

  Numerical simulations of xNLCME demonstrate that it captures both
  large scale features, related to third harmonic generation, and fine
  scale {\it carrier shocks} features of the nonlinear periodic
  Maxwell equations.
\end{abstract}

\maketitle

\tableofcontents

\section{Overview}

Realized and potential applications of microstructured dielectric
media motivate a thorough mathematical study of wave-propagation
governed by nonlinear hyperbolic equations, {\it e.g.} Maxwell's
equations with periodic and nonlinear constituitive laws.
This paper explores a class of nonlinear hyperbolic equations with a
spatially periodic flux function:

\begin{subequations}\label{nonlin-periodic}
  \begin{align}
    \partial_t{\bf v} + \partial_x{\bf f}(x, {\bf v}) &= {\bf 0} \\
    {\bf f}(x, {\bf 0}) &= {\bf 0},\\
    {\bf f}(x+2\pi,{\bf v}) &= {\bf f}(x,{\bf v}).
  \end{align}
\end{subequations}
In particular, we shall assume that periodic variations are weak (a
low contrast structure) and study solutions, whose amplitude
is small and such that the effects of periodicity-induced dispersion and
nonlinearity are in balance.

Indeed, a non-trivial spatially periodic structure is 
dispersive.  This can be seen by linearizing \eqref{nonlin-periodic}
about the state ${\bold v}={\bold 0}$, giving the linear system:
\begin{align}
  &\partial_t{\bf V} + \partial_x \left( D_{\bf v}{\bf f}(x, {\bf 0})
    {\bf V} \right) ={\bf 0},
  \label{lin-periodic}\end{align}
which retains periodicity. Floquet-Bloch theory
\cite{Eastham,Reed-Simon-IV} implies that associated to the PDE
\eqref{lin-periodic} is a family of {\it band dispersion functions}
$k\mapsto\omega_j(k),\
k\in\left(-\frac{1}{2},\frac{1}{2}\right]$. Wave propagation is
dispersive since the group velocities, $\omega_j'(k)$, are typically
non-zero. Thus, waves of different wavelengths travel with different
speeds.

Dispersive properties, encoded in the functions $\omega_j(\cdot)$ and
the associated Floquet-Bloch states, can be manipulated by {\it tuning
  the periodic structure} through, for example, modification of the
periodic lattice, the maximum and
minimum variations of $ D_{\bold v}{\bold f}(x, {\bold 0})$ (material
contrast), {\it etc}.  

It is well-known that for general initial conditions, solutions of
hyperbolic systems of conservation laws with spatially homogeneous
nonlinear fluxes:
\begin{equation} \label{nl-hyp-cons-law}
  \partial_t{\bf v}\ + \partial_x{\bf f}({\bf v}) ={\bf 0}
\end{equation}
develop singularities (shocks) in finite time, \cite{lax1964dss ,
  klainerman1980fsw}.  \medskip

\noindent {\bf Question 1:}\ {\it Is {\it spectral band dispersion}, due to
a periodic structure, 
  sufficient to arrest shock formation?}  \footnote{For example,
  though typical smooth initial data for the inviscid Burgers equation
  $\partial_tu+u\partial_xu=0$ develop shocks in finite time, the
  corresponding solutions of the Korteweg - de Vries (KdV) equation,
  $\partial_tu+u\partial_xu+\partial_x^3u=0$, a dispersive
  perturbation, remain smooth for all time. }
\\

The ability to control or inhibit the formation of singularities in
nonlinear wave propagation could have significant impact in, for
example, electromagnetics and elasticity. Strictly speaking, the
answer to Question 1 is no. Indeed, for a system of the form
\eqref{nonlin-periodic}, let us suppose that the flux function was
periodically piecewise constant.  Finite propagation speed
considerations imply that for appropriate initial data, which are
sufficiently localized within a uniform region, a shock will form.
The dispersive character of the periodic structure is manifested only
on sufficiently large spatial and temporal scales.  Thus, the problem
of controlling shock formation should be posed relative to some class
of initial conditions.
\bigskip
 
A second motivation is the study and design of media which support the
propagation of stable soliton-like pulses. These have applications to
optical devices which transfer store or, in general, process  information
 which is encoded as light pulses. Associated
with dispersive wave-propagation at wavenumber $k_\star$ is a
dispersion length $\sim (\omega''(k_\star))^{-1}$.  Soliton formation
is possible on length scales where the dispersion length and the
characteristic length on which nonlinear effects act are
comparable. Technological advances have made it possible to fabricate
 microstructured media 
with specified dispersion lengths at specified wavelengths. For a
given dispersion length, a balance between dispersion and nonlinearity
is achieved by tuning the strength of the nonlinear effects through
 adjusting  the field intensity (by an amount which is material
dependent). An example of this balance at work is {\it gap soliton}
formation in periodic structures. These are experiments in optical
fiber periodic structures (gratings) involving highly intense
(nonlinear) light with carrier wave-length satisfying the Bragg
(resonance) condition.  The length-scale of such solitons is $10^{-2}$
meters \cite{Eggleton-deSterke-Slusher}.
 
Theory predicts the existence of gap solitons traveling at any speed,
$v$, between zero and the speed of light, $c$\
\cite{christodoulides-joseph:89,aceves1989sit}. Experiments
\cite{Eggleton-deSterke-Slusher} demonstrate speeds as low as $.3c$ to
$.5c$.  Potential applications of gap solitons, based on the design of
appropriate localized defects in a periodic structure, are all-optical
storage devices \cite{Goodman-Slusher-Weinstein:02}. The term gap
soliton is used due to frequency of the gap soliton envelope lying in
the spectral gap of the linearized system.  \bigskip

Physical predictions of gap solitons are based on explicit solutions
of {\it nonlinear coupled mode equations} (NLCME), given below in
\eqref{eq:nlcme}. NLCME has been formally derived in, for example,
\cite{desterke1994gs} from \eqref{eq:maxwell1}; see also the
discussion in Section \ref{sec:nlcme}.  Rigorous derivations of NLCME,
from models with appropriate dispersion have been presented for the
anharmonic Maxwell-Lorentz equations \cite{goodman01npl} and other
nonlinear dispersive equations; see {\it e.g.} \cite{groves2001mps,
  schneider2001ncm,schneider2003eas, pelinovsky2007jcm,
  pelinovsky2008mgs}.
  
Within the approximation of a small amplitude wave field as a
wave-packet with slowly varying envelope and {\it single} carrier
frequency, propagating through a low contrast periodic structure near
the Bragg resonance (\ see scaling in Figure \ref{f:wavepack})\ ),
NLCME is argued to govern the principle forward and backward slowly
varying envelopes of carrier waves; see \cite{desterke1994gs} and
references therein.

{\it As discussed in \cite{goodman01npl} and in section
  \ref{sec:nlcme}, if the only source of dispersion is the spatial
  dispersion of the periodic medium ({\it e.g. negligible chromatic
    dispersion}) for weakly nonlinear waves in low contrast media all
  nonlinearity-generated harmonics are resonant and therefore all mode
  amplitudes are coupled at leading order. The correct mathematical
  description would appear to require infinitely many interacting
  modes.} Thus, the classical NLCME are {\bf not} a mathematically
consistent approximation.  NLCME may however be satisfactory physical
description, for some purposes.t Indeed, the soliton wave form prediction based on NLCME
appears to describe some features of experiment.
\footnote{Physicists argue in two ways that the coupling to higher
  harmonics is argued to be negligible : (i)\ The material systems
  considered are dissipative at higher wave numbers. Higher wave
  numbers are damped and therefore these mode amplitudes can be
  ignored, and (ii) {\it Chromatic dispersion} (arising due to the
  finite time response of the medium to the field) causes nonlinearly
  generated harmonics to be off-resonance. Therefore, an initial
  condition exciting the principle modes will not appreciable excite
  higher harmonics.
  These rationales are somewhat ad hoc since the precise damping
  mechanisms are not well-understood and chromatic dispersion is a
  much weaker effect than photonic band dispersion for weak periodic
  structures.}

\begin{figure}
  \centering
  \includegraphics[width=4in]{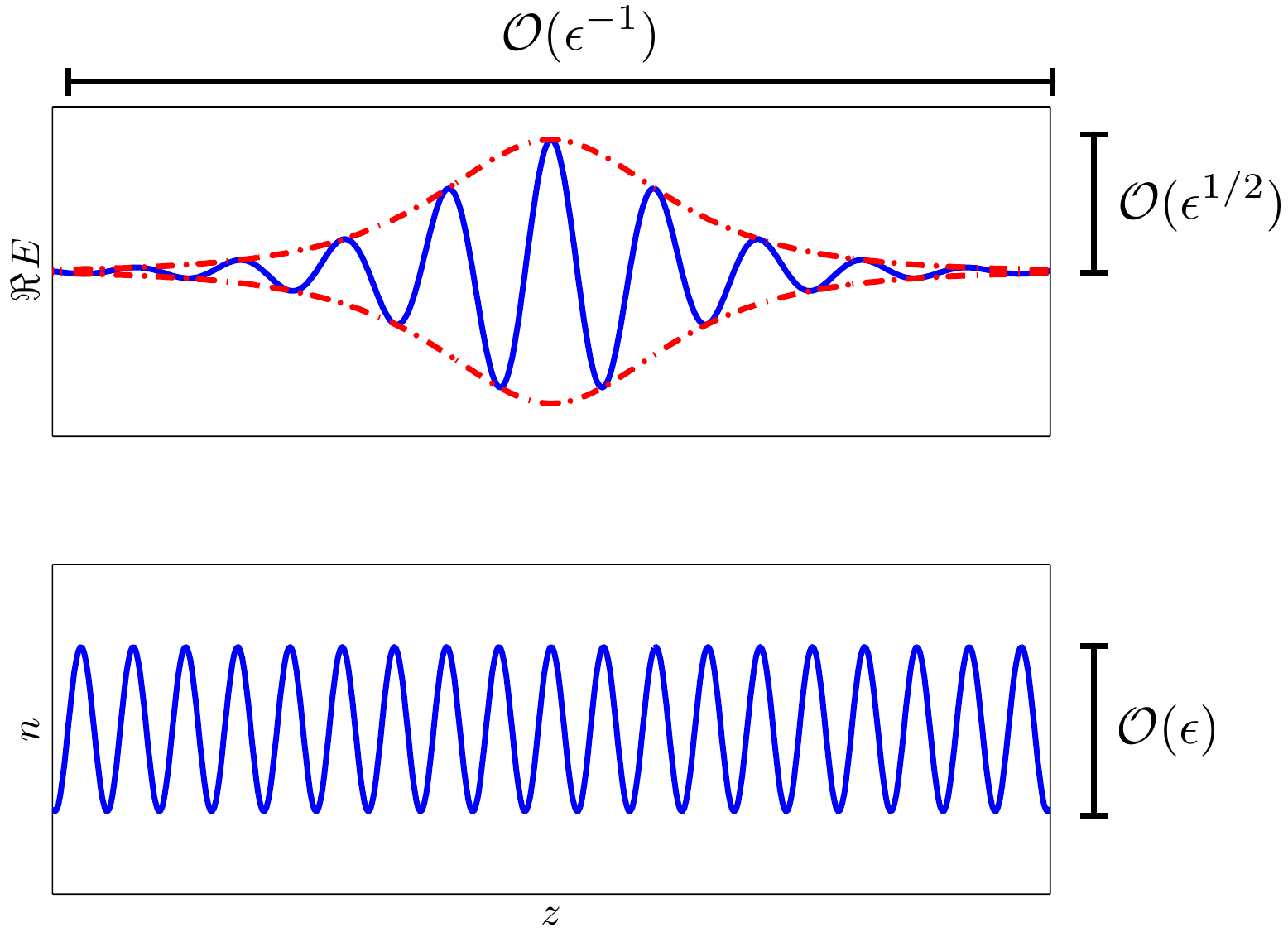}
  \caption{A wavepacket (real part, $\Re E_0(z)$, of complex field)
    with carrier wave length equal twice that of the waveguide
    refractive index ($n(z)$). }
  \label{f:wavepack}
\end{figure}
\bigskip

\noindent {\bf Question 2:}\ {\it Do nonlinear periodic
  hyperbolic systems have stable coherent structures, and can one
  develop a mathematical theory? How 
  are the classical NLCME related to this theory?
  See the discussion in section
  \ref{sec:discussion}.}
\medskip

In this article we report on progress on Questions 1 and 2 in the
context of the one-dimensional, nonlinear Maxwell equations governing
the electric ($E$) and magnetic ($B$) fields:
\begin{subequations}\label{eq:maxwell1}
  \begin{gather}
    \dt D =  \dz B,\\
    \dt B = \dz E.
  \end{gather}
\end{subequations}
with constitutive law
\begin{align}
  \label{DE}
  D\ &=\ \epsilon(z,E)\ E\ \equiv\ \left(\ n^2(z) + \chi E^2\ \right)E\\
  \label{Ndef}
  n(z)\ &=\ n_0\ +\ \epsilon N(z)
\end{align}
$n(z)$ is a {\it linear} refractive index, consisting of a nonzero
background average part, $n_0$, and a fluctuating ({\it
  e.g. periodic}) part $\epsilon N(z)$. The nonlinear term $\chi E^2$
is the nonlinear refractive index, arising from the Kerr effect; in
regions of high intensity the refractive index is higher.  The consituitive law
 \eqref{DE}, prescribes $D$ as  a
 a local function of $E$. Thus chromatic dispersion, which arises due a 
 time-nonlocal relation between $D$ and $E$ has been 
neglected.  {\it For simplicity, we assume $n_0 = 1$, which can be
  arranged by a simple scaling.}

\subsection{Summary of results}\label{sec:summary}

\begin{enumerate}
\item In section \ref{sec:observations} we present numerical
  simulations of the nonlinear periodic Maxwell equations, \eqref{eq:maxwell1},
   for initial data obtained from the
  explicit NLCME soliton. Under this time-evolution there is robust spatially localized structure
   {\it on the scale of the NLCME soliton
    envelope}.  The persistence of a localized structure and speed
  of propagation are consistent with that of the NLCME soliton.  There
  is, however, a deviation from the NLCME soliton related to {\it
    third harmonic} generation; these are the two accessory pulses
  around the principle wave in Figure \ref{f:intro_third_harmonic}
  (a).

\begin{figure}
  \centering
  \subfigure[]{\includegraphics[width=2.3in]{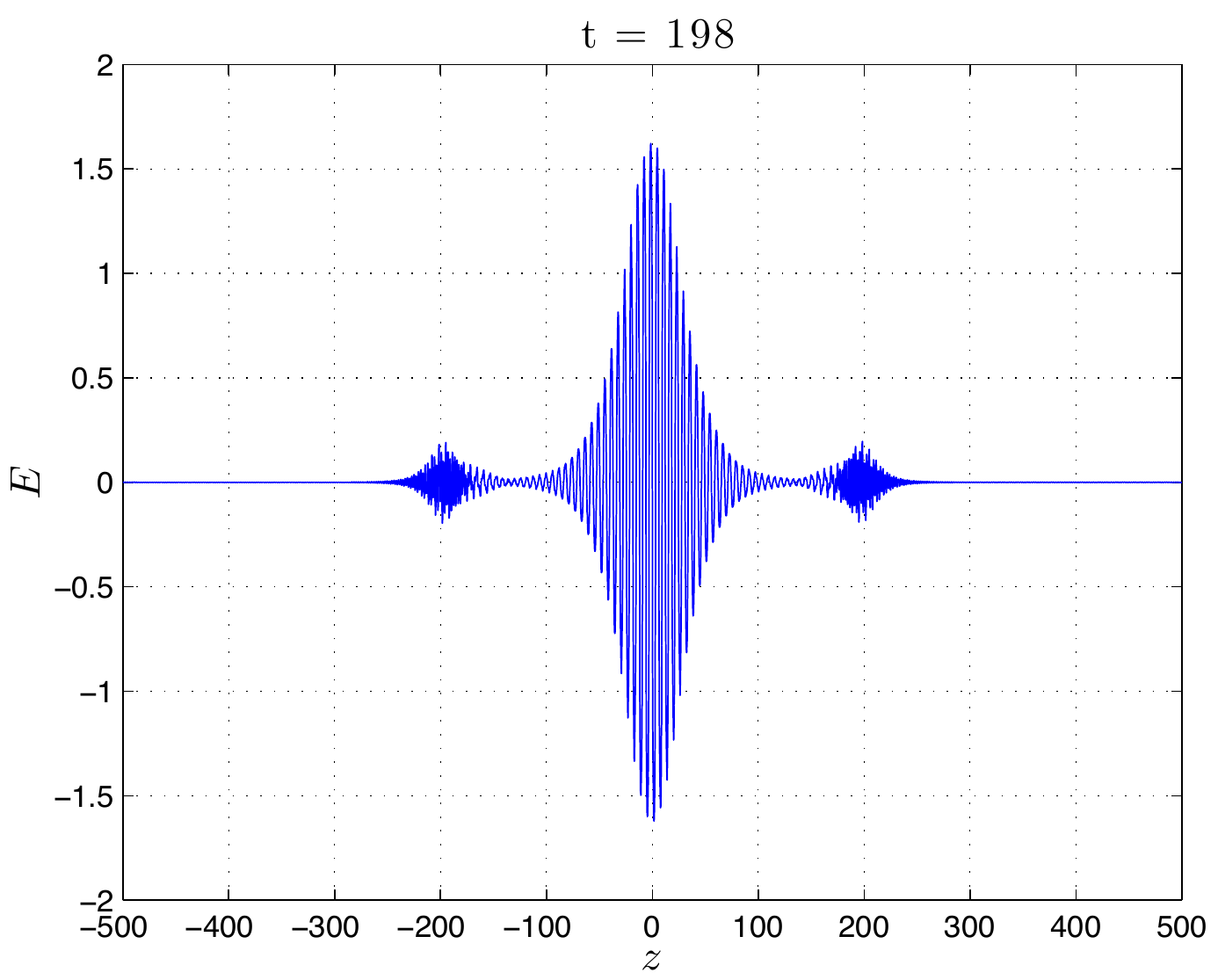}}
  \subfigure[]{\includegraphics[width=2.3in]{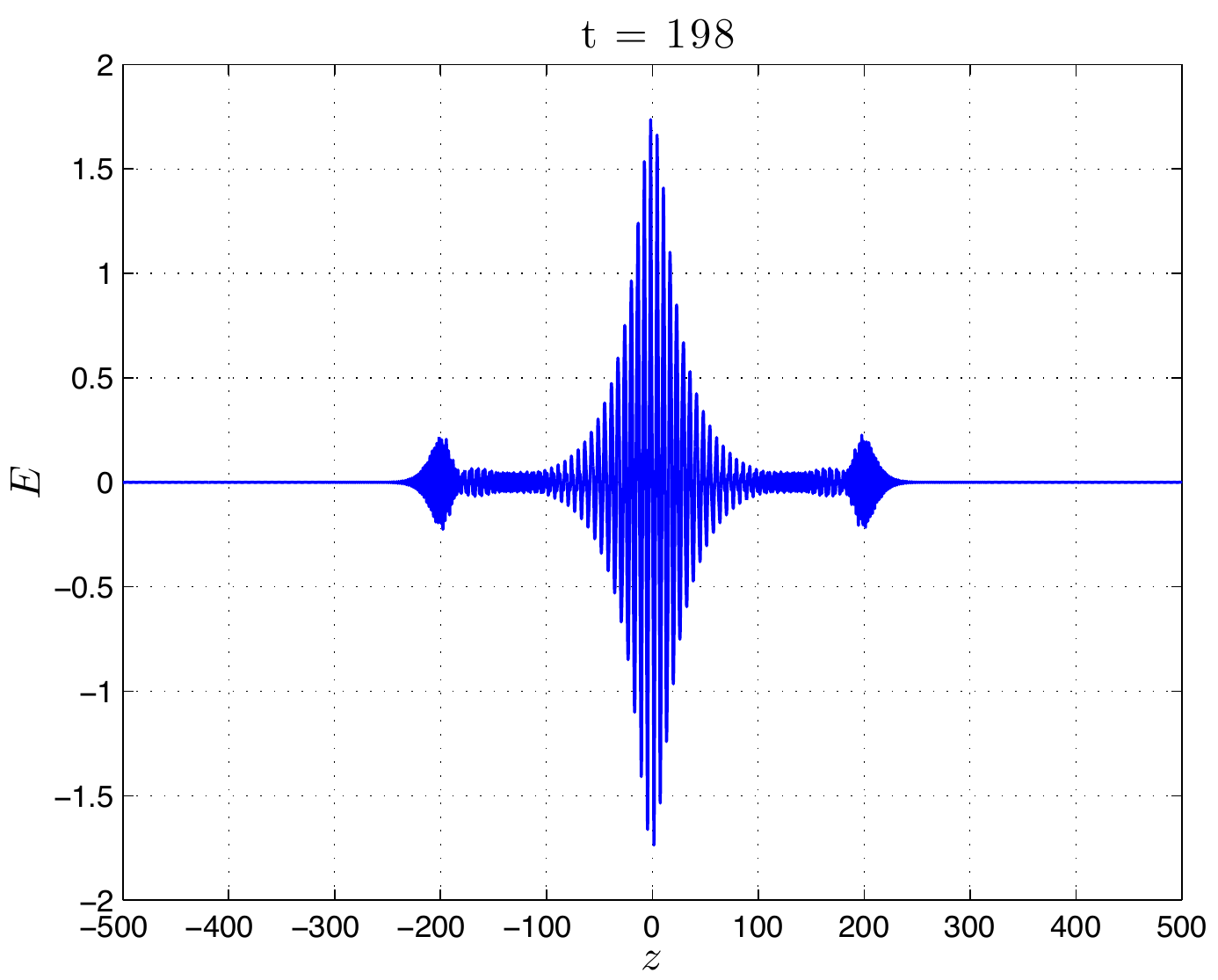}}
  \caption{On the left is a simulation of the Maxwell equations.  On
    the right is the simulation of a truncated asymptotic system,
    resolving the first and third harmonics.  Both simulations were
    initiated with the same initial conditions.  The two side pulses
    about the main wave appear to be the result of third harmonic
    generation.}
  \label{f:intro_third_harmonic}
\end{figure}

\item {\it On the microscopic scale of the carrier} there is nonlinear steepening
  and shock formation. Therefore, the solution does not evolve as a
  slowly varying envelope of a single frequency carrier wave.  The
  long-lived and spatially localized coherent structure
   which emerges has the character of a slowly
  varying envelope of a train of shocks. We call this an {\it envelope carrier-shock train}.
  Figure \ref{f:intro_carrier_shock} illustrates the shock-like small
  spatial scale behavior under slowly varying envelope.  

\begin{figure}
  \centering
  \subfigure[]{\includegraphics[width=2.4in]{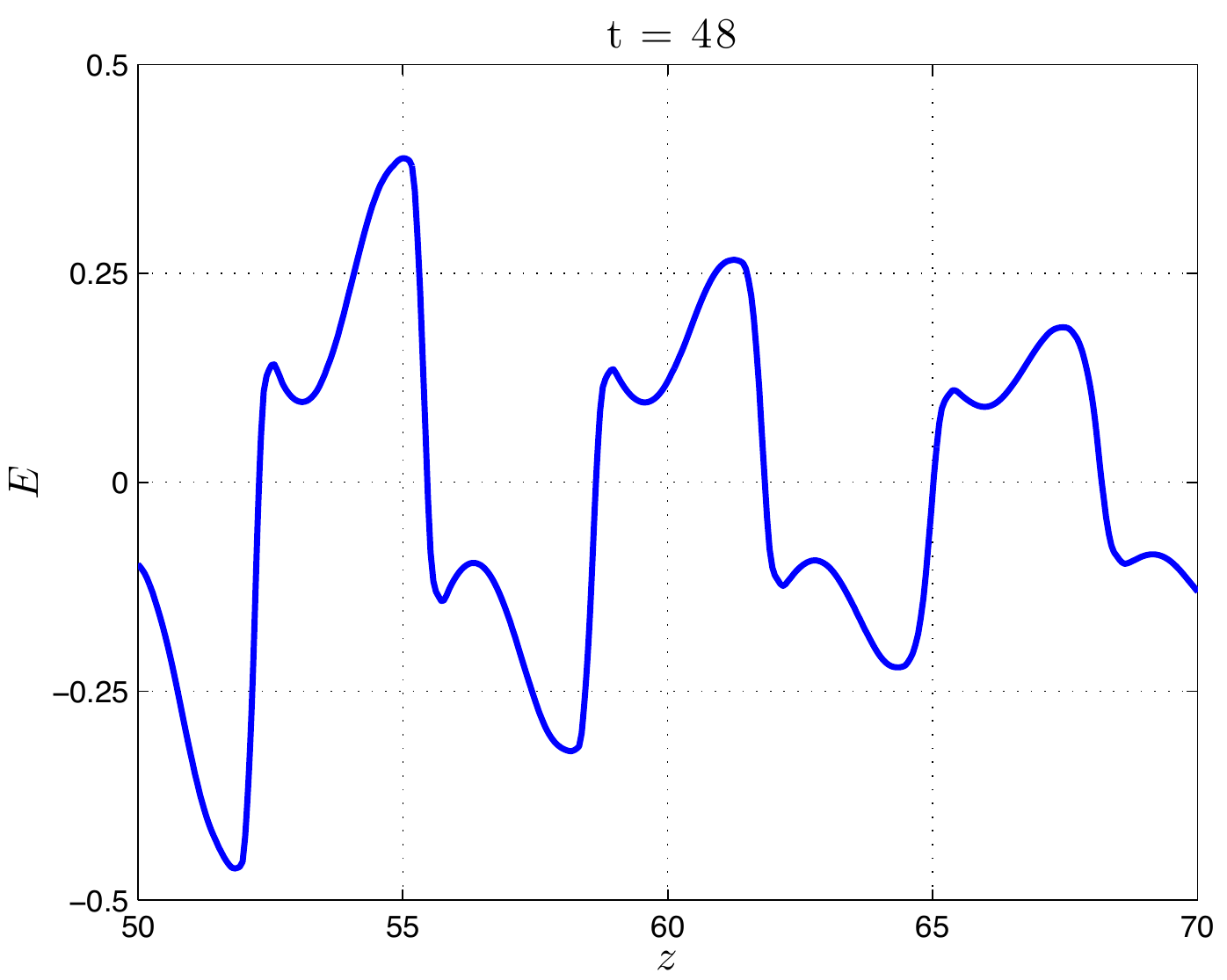}}
  \subfigure[]{\includegraphics[width=2.4in]{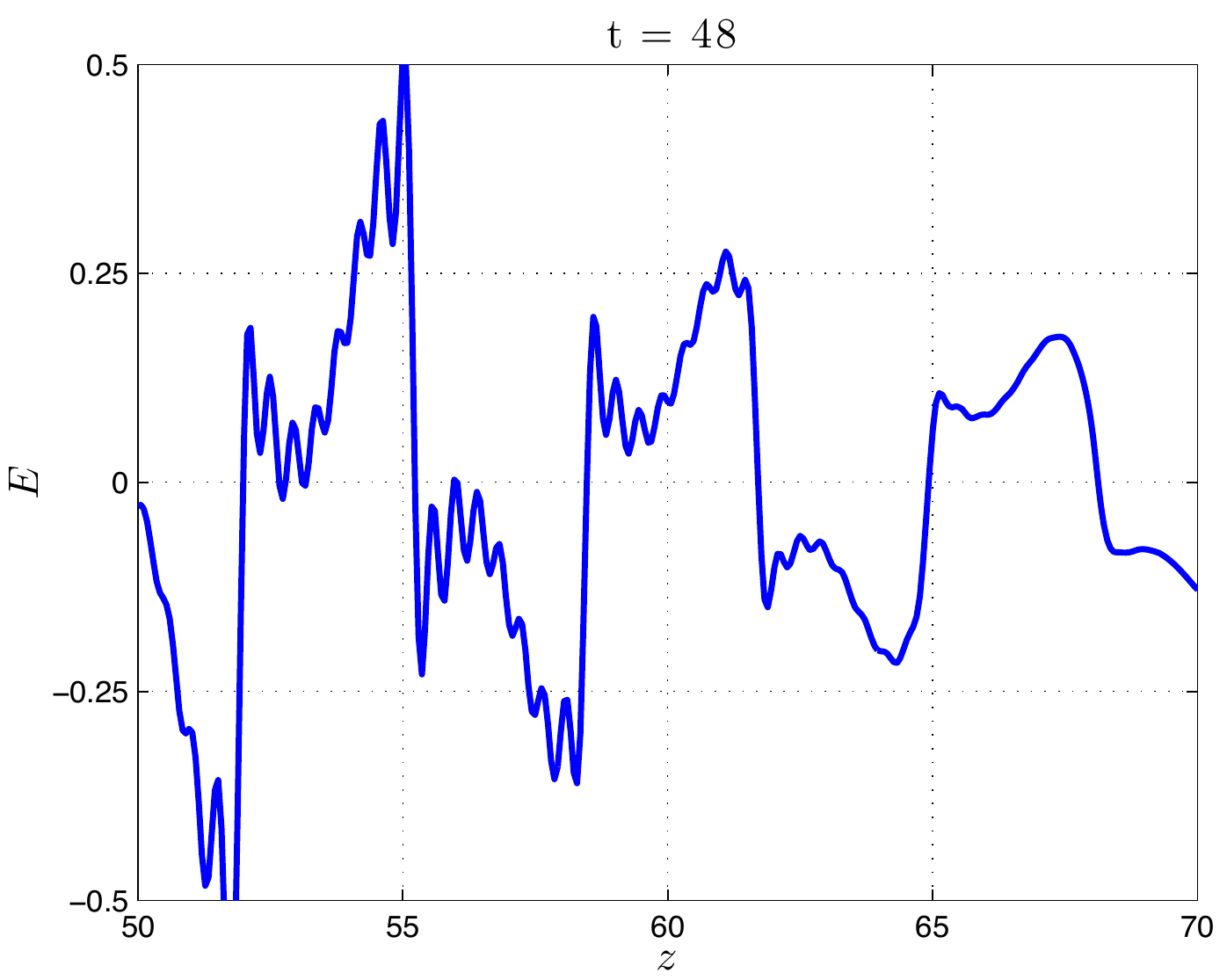}}
  \caption{On the left is a simulation of the Maxwell equations.  On
    the right is the simulation of a truncated asymptotic system.
    Both simulations were initiated with the same initial conditions.
    There is an indication of shock formation in the left.  On the
    right, we see that once sufficiently many harmonics are included,
    the Gibbs effect appears, confirming shock formation.}
  \label{f:intro_carrier_shock}
\end{figure}
\item Numerical solution of the nonlinear Maxwell equation
  \eqref{eq:maxwell1} is non-trivial due to the cubic nonlinearity.
  As a hyperbolic system, it is neither genuinely nonlinear nor
  linearly degenerate\cite{smoller1994swa,lax2006hpde,
    liu2000hyperbolic}.  To solve by finite volume methods, as we do,
  an explicit solution of the Riemann problem must be constructed.
  Details of this are given in Appendix \ref{sec:methods}.
  
 The  appropriate
  entropy condition could, in principle be derived from physical
  regularization mechanisms, which play the role of viscosity in gas
  dynamics. However, these mechanisms are not well understood.  
  However, such mechanisms and the appropriate notion of
  weak solution would respect thermodynamic principles,  which are built into our numerical scheme.


\item Using a nonlinear geometric optics expansion
  \cite{Hunter:1983vn,Hunter:1986kx,Majda:1984uq,Donnat:1997p97,Joly:1993p4514,Joly:1996p4513,Lannes:1998p4610},
  systematically keeping all resonances, we obtain nonlocal equations
  governing the interaction of all forward and backward propagating
  modes. Our asymptotic nonlocal system captures the slowly varying
  envelope of carrier-shock structures described above; see below.
  
  \bigskip

  Specifically, we introduce the general wave-form (much more general
  than a slowly varying envelope of a nearly monochromatic carrier
  plane wave), which includes all harmonics
  \begin{equation}
    E(z,t) = \eps^{1/2} \left(\ E^+(z - t, \eps z, \eps t ) + E^-(z +
      t, \eps z, \eps t)\  +\ \mathcal{O}(\epsilon)\ \right).
    \label{Eansatz}\end{equation}
  Let
  \begin{equation}
    \phi_\pm = z \mp  t,\ \ \eps t = T,\ \ {\rm and}\ \ \eps z = Z\ .
    \nn\end{equation}
  At leading order, the slow evolution of backward and forward
  components is governed by the coupled \emph{integro-differential} equations:
  \begin{subequations}
    \label{eq:integro_diff_intro}
    \begin{gather}
      \begin{split}
        \dT \Ep + \dZ \Ep&= \dphi  \mean{N( \phi_+ +   s) E^-( \phi_+ + 2  s,Z,T)}_s\\
        &\quad +\Gamma \dphi\left[
          \frac{1}{3}\paren{\Ep}^3+\Ep\mean{\paren{\Em}^2} \right],
      \end{split}
      \\
      \begin{split}
        \dT \Em - \dZ \Em&=-\dphi\mean{N(  \phi_- - s)E^+(  \phi_- -2 s, Z,T) }_s\\
        &\quad -\Gamma
        \dphi\left[\mean{\paren{\Ep}^2}\Em+\frac{1}{3}\paren{\Em}^3\right].
      \end{split}
    \end{gather}
  \end{subequations}
  Here $\mean{\cdot}$ is an averaging operation in the $\phi$
  argument;
  \begin{equation}
  \langle\ f\ \rangle\ \equiv\ \lim_{T\to\infty}\ \frac{1}{T}\ \int_0^T\ f(s)\ ds;
  \label{avg-def}
  \end{equation}
    see also section \ref{sec:asympt}. Equations
  \eqref{eq:integro_diff_intro} arise as a constraint on
  $E^\pm(\phi_\pm,Z,T)$ ensuring that the $\mathcal{O}(\epsilon)$
  error term \eqref{Eansatz} in remains small on large time scales:
  $T=\mathcal{O}(1)$ or equivalently $t=\mathcal{O}(\epsilon^{-1})$.
  
  Spatial variations in the refractive index, $N(z)$, give rise to a
  coupling of backward and forward waves. Indeed, if $N(z)\equiv0$ and
  one specifies data for the system \eqref{eq:integro_diff_intro} at
  $t=0$ with non-zero forward components ($E^+\ne0$) and, no backward
  components ($E^-=0$) then, formally, $E^-$ remains zero for all
  time, {\it i.e.} no backward waves are generated.  Continuing with
  this assumption of $N = 0$ and $E^-_0 = 0$, if we let $V(\phi, T) =
  E^+(\phi, Z_0 - T, T)$, with $Z_0$ arbitrary, then $V$ satisfies
  \[
  \partial_T V = \tfrac{\Gamma}{3} \partial_\phi(V^3).
  \]
  This generalized Burger's equation will gives rise to a finite time
  singularity.  We revisit this observation in the discussion, section
  \ref{sec:discussion}, when considering how singularities might
  appear when the linear coupling between backwards and forwards waves
  is restored, $N \neq 0$.

    


\item The nonlocal equations may also be written as an infinite system
  of coupled mode equations.  In the case where $E^\pm$ is $2\pi-$
  periodic in $\phi_\pm$, the integro-differential equation
  \eqref{eq:integro_diff_intro} reduces to an infinite system of
  coupled mode equations for the Fourier coefficients
  $\left\{E^\pm_p(Z,T) : p\in\mathbb{Z} \right\}$:
  \begin{subequations}\label{e:mode_intro}
    \begin{gather}
      \label{eq:mode_intro_p}
      \begin{split}
        \dT \Ep_p + \dZ \Ep_p = \mathrm{i}p N_{2p}{\Em_p} +
        \mathrm{i}p\frac{\Gamma}{3}&\left[\sum\Ep_q
          \Ep_r \Ep_{p-q-r} \right.\\
        &\quad\left.+3\paren{\sum \abs{\Em_q}^2} \Ep_p \right],
      \end{split}
      \\
      \label{eq:mode_intro_m}
      \begin{split}
        \dT \Em_p - \dZ \Em_p = \mathrm{i}p \bar{N}_{2p}{\Ep_p}
        +\mathrm{i}p\frac{\Gamma}{3} &\left[\sum \Em_q
          \Em_r \Em_{p-q-r} \right.\\
        &\quad \left.+3\paren{\sum \abs{\Ep_q}^2} \Em_p \right].
      \end{split}
    \end{gather}
  \end{subequations}
  We call this system the {\it extended nonlinear coupled mode
    equations} (xNLCME).
  xNLCME reduces to the classical NLCME if we neglect higher
  harmonics.
\item Simulations of successively higher dimensional mode truncations
of \eqref{e:mode_intro}  show improved resolution of the carrier shocks under a slowly
  varying envelope, whose scale is captured by a comparatively low
  order truncation.  Indeed, Figure \ref{f:intro_third_harmonic} (b)
  shows that inclusion of the third harmonic in the asymptotic system
  resolves the large scale feature, while inclusion of additional
  harmonics in Figure \ref{f:intro_carrier_shock} (b) shows the Gibbs
  effect, expected for a finite Fourier representation of a
  discontinuous function.  This demonstrates that our asymptotic
  analysis leads to equations capturing the essential features of
  nonlinear Maxwell.  However, if we consider how energy, initially
  only in the first harmonic, is redistributed in time, we see in
  Figure \ref{f:eng_intro} that most of the energy persists in the
  first harmonic.  This reflects the partial success of NLCME as a
  model for periodic nonlinear Maxwell.
\end{enumerate}

\begin{figure}
  \centering
  {\includegraphics[width=3in]{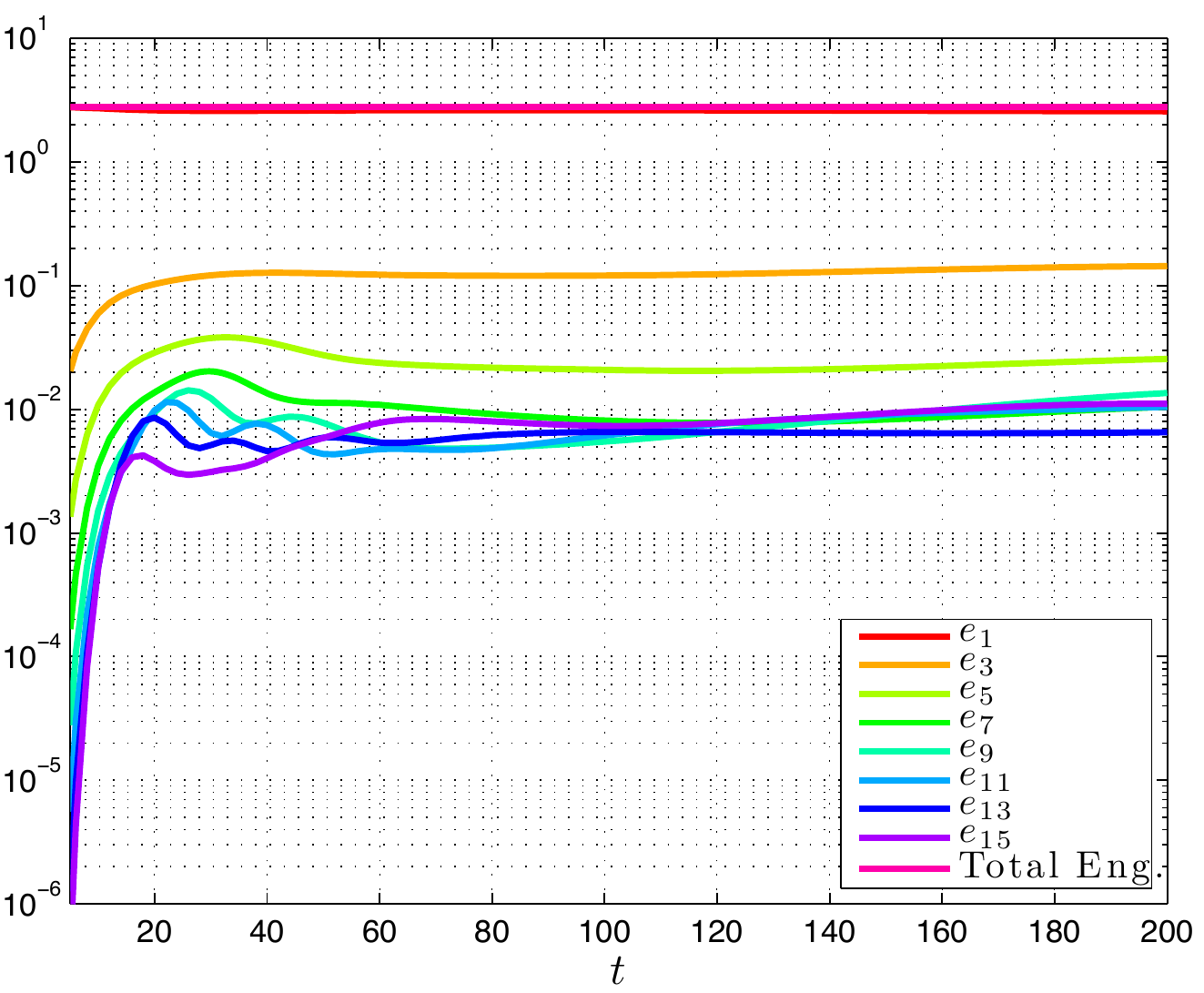}}
  \caption{Truncating \eqref{e:mode_intro} to odd harmonics
    $\abs{p}\leq 16$, we simulate the initial value problem an NLCME
    soliton in the first harmonic, and the others zero.  The above
    time series of the energy associated with each harmonic, $e_p$,
    shows that most of the energy continues to reside in the first
    harmonic.}
  \label{f:eng_intro}
\end{figure}are

\noindent{\bf Relation to previous work:}\  Some of the earliest
examinations on optical shocks can be found in Rosen,
\cite{Rosen:1965p6493}, and, DeMartini {\it et al.}
\cite{DeMartini:1967p6494}.  In these works, the authors applied the
method of characteristics to a unidicretional model.  
Kinsler and Kinsler {\it et al.} have continued to examine this
problem, and have developed an algorithm for detecting the onset of
shock formation, \cite{Kinsler:2007p6495,Kinsler:2007p6542}.  Carrier shocks were also
examined by Flesch, Moloney, \& Mlejnek, \cite{flesch1996cws}, for spatially homogeneous
Maxwell system with chromatic dispersion, modeled via a time-nonlocal
Lorentzian polarization response.  Ranka, Windeler, \& Stentz have found experimental evidence of
optical shocks, \cite{Ranka:2000p6690}.  In their work, a
  monochromatic pulse with sufficient power steepened and
  generated a broadband optical continuum.

Coherent structures in nonlinear and periodic media have also been
studied by LeVeque, LeVeque \& Yong, and Ketcheson
\cite{leveque2002fvm,leveque2003swl,Ketcheson:2009fk} in a model for
heterogeneous nonlinear {\it elastic} media.  They considered order
one solutions in high contrast, rapidly varying, periodic
structures. Their simulations yielded localized structures on the
scale of many periods with oscillations on the scale of the
period. For piecewise constant (discontinuous) periodic structures,
they have a discontinuous carrier shock-like character on the scale of
the period, though this is due to discontinuities in the medium, the
fluxes remain continuous. A two-scale (homogenization) expansion
yields a nonlinear dispersive equation, with solitary waves, similar
to the computed solution envelope.  In their physical regime, the
variations in the properties of the media and the nonlinearity are
$\mathcal{O}(1)$.  In contrast, we consider an asymptotic regime where
the constrast of the periodic structure and nonlinearity are of the
same order, $\mathcal{O}(\eps)$. Furthermore, the initial condition
has two scales (envelope and carrier scales), where the carrier wave
length is of the same order, indeed in resonance with, the periodic
structure.  These different scalings lead to different asymptotic
descriptions.  An early example of the interactions between nonlinearity and
a periodic structure was in atmospheric science, studied by Majda {\it
  et al.}, \cite{Majda:1999p2}.  In this work, a model of the
interaction of equatorial waves with topography gives rise to
nonsmooth profiles (in this case, solitary waves with corner
singularities).


Finally, systems of coupled modes have also been examined in prior
works, though the work is typically limited two just two harmonics,
such as a first and second harmonic system or a first and third
harmonic system. Such a system was studied by Tasgal, Band, \& Malomed
\cite{Tasgal:2005p6335}, who were able to find stable {\it polychromatic}
solitons in a first and third harmonic system.


%

\bigskip

An outline of this paper is as follows.  In Section \ref{sec:nlcme},
we review how NLCME arises as an approximation of nonlinear Maxwell.
Results of Maxwell Simulations, showing the coherent structures and
shocks, are given presented in Section \ref{sec:observations}.  We
then present our derivation of xNLCME in Section \ref{sec:asympt},
followed by simulations of this system in Section \ref{s:xnlcme_sims}.
We discuss all of these results in Section \ref{sec:discussion}.


{\bf Acknowledgements:} The authors would like to thank R.R. Rosales
for discussions during the early stages of this work on the use
nonlinear geometrical optics.  We also thank M. Pugh, D. Ketcheson, R.J. LeVeque,
and C. Sulem for helpful discussions. GS was supported in part by NSF-IGERT grant
DGE-02-21041, NSF-CMG grant DMS-05-30853, and NSERC. MIW was supported in part by
NSF grants DMS-07-07850 and DMS-10-08855. MIW would also like to
acknowledge the hospitality of the Courant Institute of Mathematical
Sciences, where he was on sabbatical during the preparation of this
article.

\section{Nonlinear Maxwell and NLCME}
\label{sec:nlcme}
In this section we briefly review how NLCME arises from nonlinear
Maxwell with a periodically varying index of refraction. We also
identify the mathematical inconsistency of NLCME as a description of the
wave-envelope.

First, we write the nonlinear Maxwell equation \eqref{eq:maxwell1} as
\begin{equation}
  \label{e:maxwell_scalar}
  \dt^2\paren{n(z)^2E + \chi E^3} = \dz^2 E
\end{equation}
with index of refraction
\begin{equation}
  n(z) = 1 + \eps N(z), \quad 0 < \epsilon \ll 1,
\end{equation}
where $N(z)$ is $2 \pi$ periodic with mean zero and Fourier series:
\begin{equation}
  N(z) = \sum_{p \in \mathbb{Z}\setminus\{0\}} N_p e^{\mathrm{i} p z}.
  \label{N-Fourier}
\end{equation}

We shall seek solutions which
incorporate \begin{inparaenum}[(i)] \item slow variations in time and
  space, due to the weak modulation about a constant refractive
  index; \item a scaling of the wave-field which seeks solutions in
  which the effects of dispersion and nonlinearity are in balance:
  \begin{equation}
    E^\epsilon(z,t)\ =\ \epsilon^{1\over2}\ \mathcal{E}^\epsilon(z,t;Z,T),\ \ Z=\epsilon z,\ \ T=\epsilon t.
    \label{Eassumptions}
  \end{equation}
\end{inparaenum}

Rewriting \eqref{e:maxwell_scalar} in terms of new variables dependent
$\mathcal{E}^\epsilon$ and independent $(z,t,Z,T)$ variables, we
obtain:
\begin{equation}
  \left(\partial_t^2-\partial_z^2\right) \mathcal{E}^\epsilon\ 
  +\ \epsilon\left(2\partial_t\partial_T \mathcal{E}^\epsilon-2\partial_z\partial_Z \mathcal{E}^\epsilon+2N(z) \mathcal{E}^\epsilon+\chi\left( \mathcal{E}^\epsilon\right)^3\ \right)\ +\ \mathcal{O}(\epsilon^3)\ =\ 0.
  \nn\end{equation}
Formally expanding $ \mathcal{E}^\epsilon$ as
\begin{equation}
  \mathcal{E}^\epsilon(z,t,Z,T)\ =\  \mathcal{E}_0(z,t,Z,T)\ +\ \epsilon\ \mathcal{E}_1(z,t,Z,T)\ +\ \dots
  \nn\end{equation}
we obtain the following hierarchy for $\mathcal{E}_j(z,t,Z,T)$, $j
\geq 0$:
\begin{equation}
  \begin{split}
    \mathcal{O}(\epsilon^0) & \left(\partial_t^2-\partial_z^2\right)
    \mathcal{E}_0=0 \\
    \mathcal{O}(\epsilon^1) &   \left(\partial_t^2-\partial_z^2\right) \mathcal{E}_1 = -2\partial_t\partial_T \mathcal{E}_0+2\partial_z\partial_Z \mathcal{E}_0-2N(z) \mathcal{E}_0-\chi\left( \mathcal{E}_0\right)^3 \\
    &\vdots\\
    \mathcal{O}(\epsilon^j) & \left(\partial_t^2-\partial_z^2\right)
    \mathcal{E}_j =\text{expressions in terms of}\ \mathcal{E}_l,\ \ 0\le l\le j-1 \\
    & \vdots
  \end{split}
  \label{hierarchy}
\end{equation}

Solving the $\mathcal{O}(\epsilon^0)$ equation yields:
\begin{equation}
  \mathcal{E}_0(z,t,Z,T) =
  \mathcal{E}^+(Z,T)e^{i(z- t)} + \mathcal{E}^-(Z, T)e^{-i(z+ t)} + \cc
  \label{svea}\end{equation}
Thus,   the leading order  consists of backward and forward propagating waves, 
modulated by the slow envelope amplitude functions
 $\mathcal{E}^\pm(Z,T)$, which are to be determined.
  
Substitution of \eqref{svea} into the $\mathcal{O}(\epsilon^1)$
equation for $\mathcal{E}_1$ yields the equation: {\begin{equation}
    \begin{split}
      & \left(\partial_t^2-\partial_z^2\right) \mathcal{E}_1\\
      &=\bracket{2i \partial_T \mathcal{E}^+ - 2i \partial_Z
        \mathcal{E}^+ - 2 N_2 \mathcal{E}^- - 3 \chi\paren{
          \abs{\mathcal{E}^+}^2 + 2 \abs{\mathcal{E}^-}^2} \mathcal{E}^+  } e^{i(z-   t)}\\
      &+\bracket{2i \partial_T \mathcal{E}^- - 2i \partial_Z
        \mathcal{E}^+ - 2 \bar{N}_2 \mathcal{E}^+ - 3 \chi\paren{
          \abs{\mathcal{E}^-}^2 + 2  \abs{\mathcal{E}^+}^2 } \mathcal{E}^-}e^{-i(z+   t)} \\
      &+\paren{\mathcal{E}^+}^3 e^{3i(z- t)} +\paren{\mathcal{E}^-}^3
      e^{-3i(z+ t)} + \cc + \text{non-resonant terms}
    \end{split}
    \label{E1eqn} 
  \end{equation}}
We have used that $N_0 = 0$ and 
\begin{equation}
  \begin{split}
    &N(z) (\mathcal{E}^+ e^{i(z- t)} + \mathcal{E}^- e^{-i(z+t)})\\
    &\quad = N_{-2} \mathcal{E}^+ e^{-i(z+ t)} + N_2 \mathcal{E}^-
    e^{i(z- t)} +\cc+ \text{non-resonant terms}.
  \end{split}
\end{equation}
Each term, explicitly written on the right hand side of \eqref{E1eqn}, is resonant with the
kernel of $\left(\partial_t^2-\partial_z^2\right)$ .  It follows that
the coefficients of {\bf all} harmonic plane waves: $e^{\pm i q (z-
  t)}$ and $e^{\pm i q (z+ t)},\ q\in\mathbb{Z}$ must vanish for
$\mathcal{E}_1$ to be bounded in $t$.  
 
The vanishing of the coefficients of $e^{i(z- t)}$ and $e^{-i(z+ t)}$
yields the nonlinear coupled mode equations (NLCME):
\begin{subequations}
  \label{eq:nlcme}
  \begin{align}
    \dT \mathcal{E}^+ +  \dZ \mathcal{E}^+ = i  N_2 \mathcal{E}^- + i\Gamma\paren{\abs{\mathcal{E}^+}^2 +  2 \abs{\mathcal{E}^-}^2 } \mathcal{E}^+,\\
    \dT \mathcal{E}^- - \dZ \mathcal{E}^- = i \bar{N}_{2}
    \mathcal{E}^+ + i \Gamma \paren{\abs{\mathcal{E}^-}^2 + 2
      \abs{\mathcal{E}^+}^2 } \mathcal{E}^-,
  \end{align}
\end{subequations}
where $ \Gamma \equiv \frac{3}{2}\chi$ and $\bar{N}_2=N_{-2}$.  The
initial value problem for \eqref{eq:nlcme} is well-posed
\cite{goodman01npl}.  NLCME also has explicit family of {\it
  gap-soliton} solutions; see Appendix \ref{s:nlcme_soliton}.

However, requiring $\mathcal{E}^\pm$ to satisfy \eqref{eq:nlcme} removes only
the lowest harmonic resonances. This is the approximation invoked in
the physics literature; see the survey \cite{desterke1994gs} and
references cited therein.

Note however that the remaining explicitly displayed terms on the right hand side of
\eqref{E1eqn} are resonant as well and induce linear in time
growth. If we choose to remove the resonant terms proportional to
$e^{3i(z- t)}$ and $e^{-3i(z+ t)}$ by including slow modulations of
these plane waves at $\mathcal{O}(\epsilon^0)$, nonlinearity and
parametric forcing through $N(z)$ will generate yet other resonant
harmonics.
 
{\it A leading order solution which does not generate resonant terms
  at higher order must contain \underline{all} harmonics. Thus, NLCME
  is mathematically inconsistent. In section \ref{sec:asympt} we derive
  an integro-differential equation, which consistently incorporates all
  resonances. As seen from our numerical
  and asymptotic studies,   this nonlocal nonlinear
  geometrical optics equation more accurately capture features on
  both small and large spatial scales, {\it e.g.} changes in the envelope
   due to higher harmonic generation, as well as carrier shock formation.  }



\section{Simulations of nonlinear periodic Maxwell}
\label{sec:observations}
In this section we discuss the results of numerical simulations, based
on the algorithms of Appendix \ref{sec:methods}, of the nonlinear and
periodic Maxwell equations \eqref{e:maxwell_scalar}.
\begin{itemize}
\item In section \ref{s:maxwell_solitons} we show that for Cauchy
  initial data derived from the classical NLCME soliton, there evolve
  spatially localized soliton-like states which persist on long time
  scales.  We discuss aspects of the large scale (envelope) structure
  of such states, which are consistent with the NLCME soliton, as well
  as significant deviations.
\item In section \ref{s:maxwell_shocks} we show that smoothness breaks
  down in finite time. In particular, we observe shock formation on
  the fast spatial scale of the carrier wave, while a slowly varying
  envelope evolves smoothly.
\end{itemize}

We begin by expressing \eqref{e:maxwell_scalar} as a first order
system:
\begin{equation}
  \dt \begin{pmatrix} n(z)^2 E +  \chi E^3 \\ B\end{pmatrix} +
  \dz \begin{pmatrix} -B\\-E\end{pmatrix}=0.
\end{equation}
We introduce the scaling $(E,B,D)^T = \epsilon^{1/2}
(\tilde{E},\tilde{B},\tilde{D})$, and expressing the equations in
terms of the variables: $(\tilde{D}, \tilde{B})$ coordinates.
Dropping tildes, this is
\begin{equation}
  \label{e:rescaled_maxwell}
  \dt \begin{pmatrix}D \\ B\end{pmatrix} +
  \dz \begin{pmatrix} -B\\-E(D,z)\end{pmatrix}=0. 
\end{equation}
where $E(D,z)$ is the unique real solution of
\begin{equation}
  \label{e:rescaled_closure}
  D = n(z)^2E + \eps \chi E^3
\end{equation}

\subsection{Soliton-like coherent structures}
\label{s:maxwell_solitons}

As is well known \cite{christodoulides-joseph:89,aceves1989sit} NLCME
has spatially localized gap soliton solutions. We use the analytical
expression for the gap soliton to generate Cauchy initial data,
$E(z,0),\ \partial_tE(z,0)$ for \eqref{e:maxwell_scalar} and
numerically simulate the evolution.
 
Using \eqref{svea} and the leading order approximation for the
magnetic field $B_1^\pm = \mp E_1^\pm$, NLCME soliton data (see
\eqref{eq:NLCME_soliton} in Appendix \ref{s:nlcme_soliton}) can be
seeded into Maxwell using
\begin{subequations}
  \begin{align}
    E &= \mathcal{E}^+(\eps z, \eps t)e^{i(z-t)} + \mathcal{E}^-(
    \eps z, \eps t) e^{-i(z+t)} + \cc,\\
    B &= - \mathcal{E}^+(\eps z, \eps t)e^{i(z-t)} +
    \mathcal{E}^+(\eps z, \eps t)e^{i(z-t)} + \cc
  \end{align}
\end{subequations}
We obtain $D$ via \eqref{e:rescaled_closure} and evaluate at $t=0$ to
get the initial condition.

For a spatially varying index of refraction, we take
\begin{equation}
  \label{e:index}
  N(z) = \tfrac{4}{\pi} \cos( 2 z), \quad {\it i.e.}\ N_2 = N_{-2} =\tfrac{2}{\pi},\ \ N_p=0,\ \ |p|\ne2,
\end{equation}
$\eps = 0.0625$, and $\chi=1$ ($\Gamma = \tfrac{3}{2}$). The results
of our simulations appear in a - d of Figures
\ref{f:traveling_soliton} and \ref{f:standing_soliton}.  While there
is attenuation in amplitude and some dispersive spreading of energy,
the solution remains spatially localized over long time intervals. Not
only is there a persistence of the localization (with the periodic
medium), but also there is good pointwise agreement with the NLCME
approximation; see Figure \ref{f:traveling_soliton_zoom}.

Frames e - h of Figures \ref{f:traveling_soliton} and
\ref{f:standing_soliton} display the corresponding results in the
absence of a periodic structure, {\it i.e.} $N(z) \equiv 0$.  The
delocalization, dispersive spreading and attenuation of the wave
amplitude is greatly enhanced. To understand this heuristically, note
that a gap soliton is a localized state whose frequency lies in the
spectral gap of the linearized PDE at the zero solution. A focusing
nonlinearity adds a (self-consistent) potential well, creating a
(nonlinear) defect mode with frequency lying in this spectral gap. If
$N(z)\equiv0$ then the linearization at the zero state has {\it no}
spectral gap. Thus, a oscillating with the gap soliton frequency would
couple to radiation modes and dispersively spread and attenuate.
This mechanism is discussed, for example, in \cite{Soffer-Weinstein:98}.


\begin{figure}
  \centering

    Simulations with varying refractive index, \eqref{e:index}:

  \subfigure[]{\includegraphics[width=2.2in]{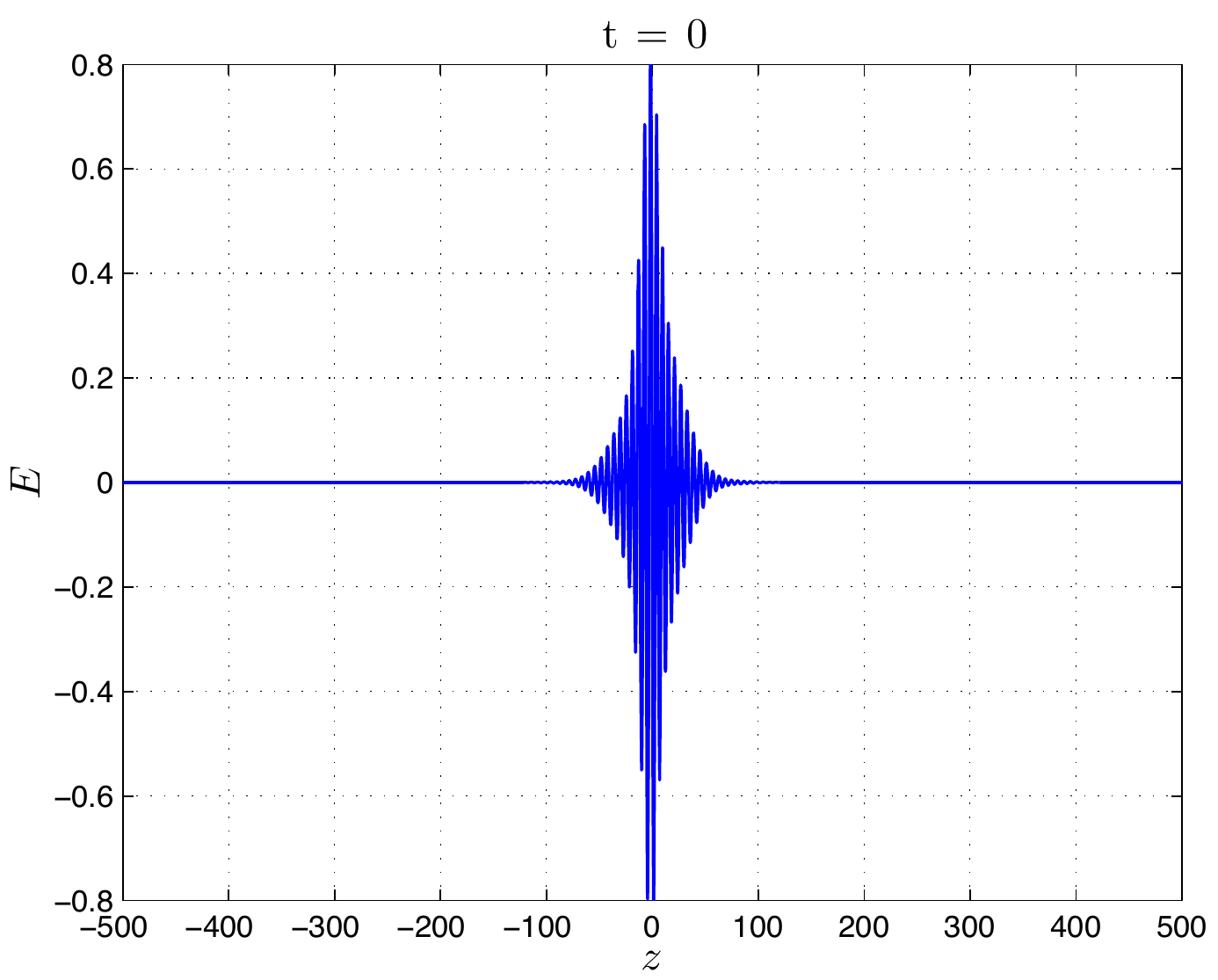}}
  \subfigure[]{\includegraphics[width=2.2in]{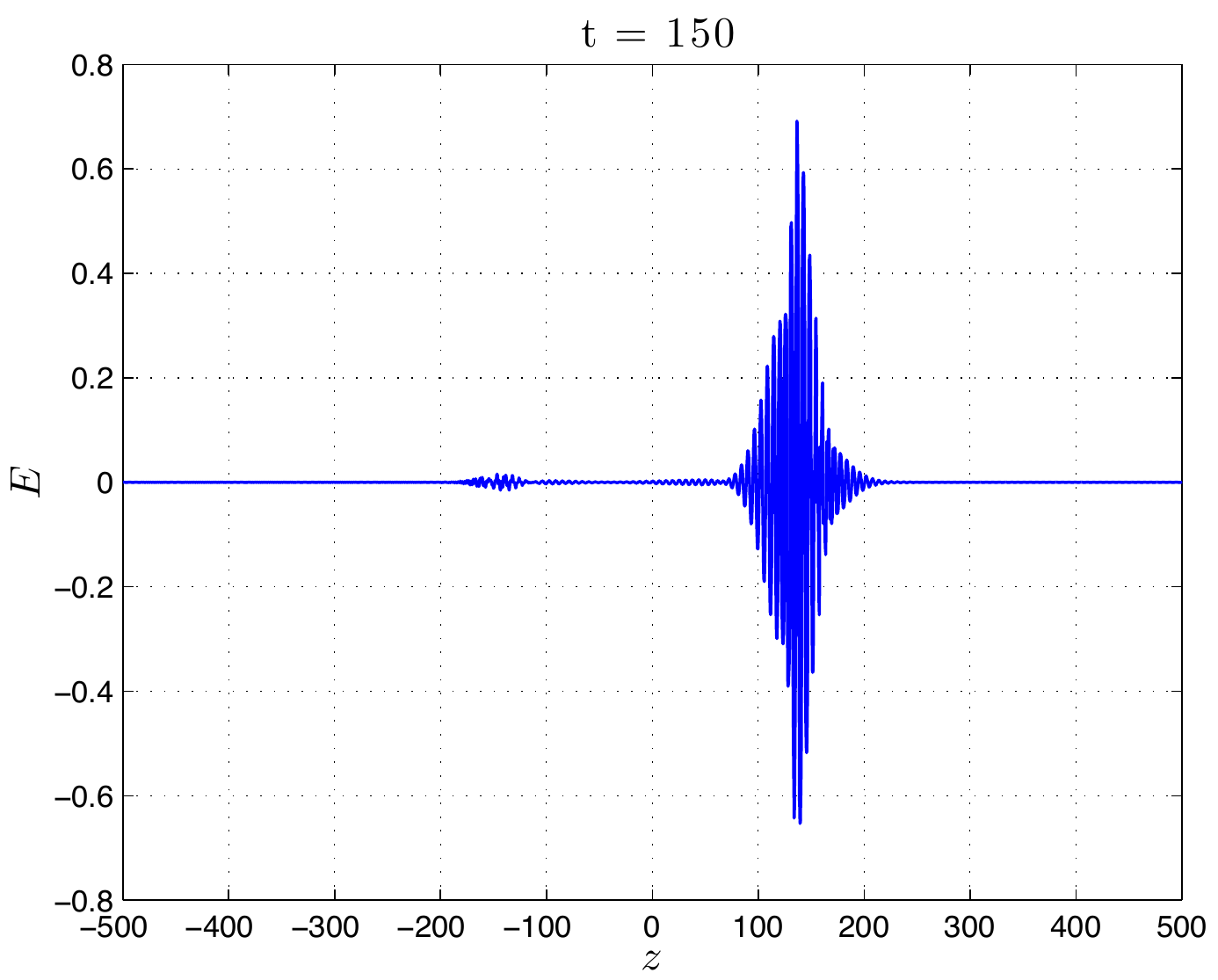}}

  \subfigure[]{\includegraphics[width=2.2in]{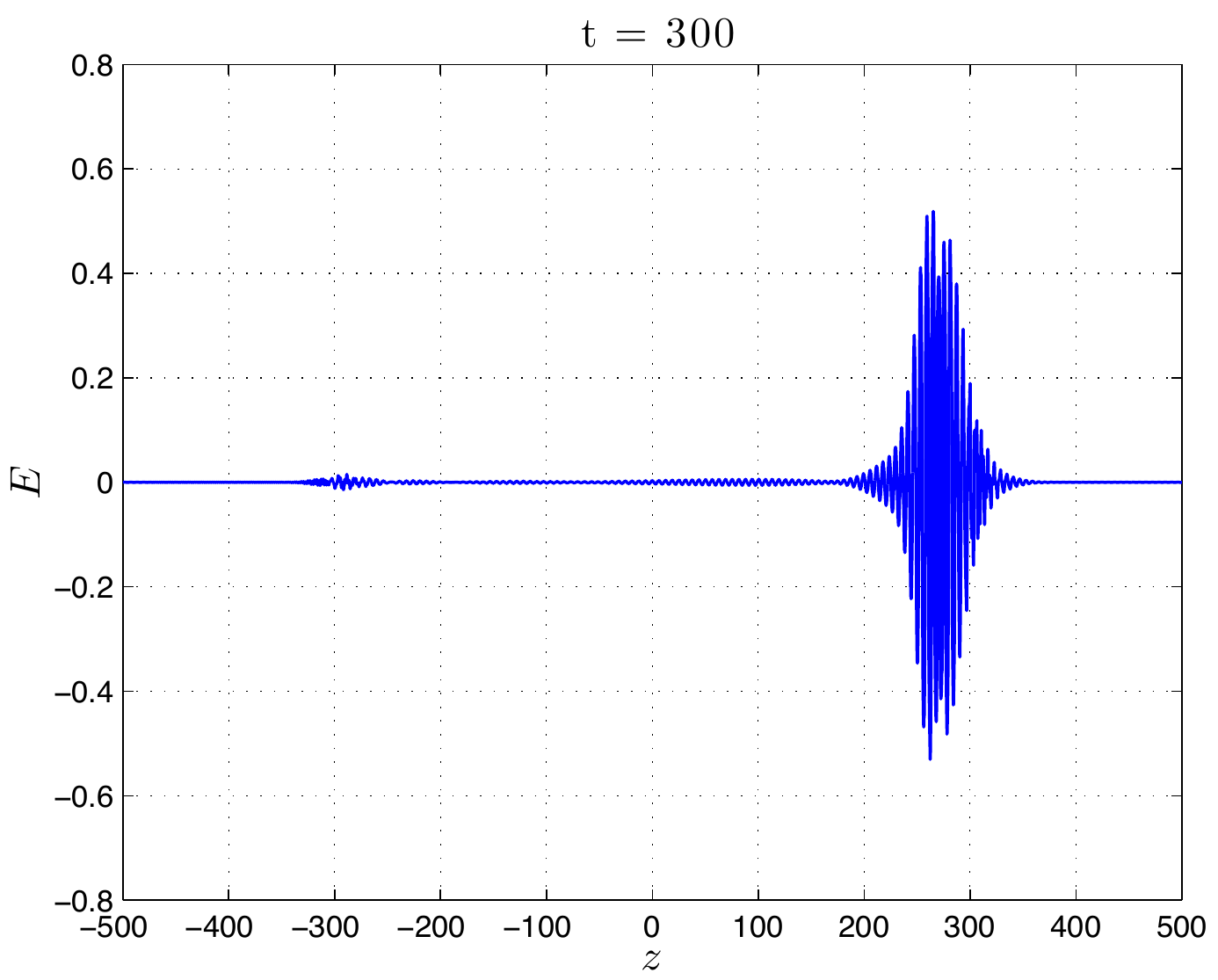}}
  \subfigure[]{\includegraphics[width=2.2in]{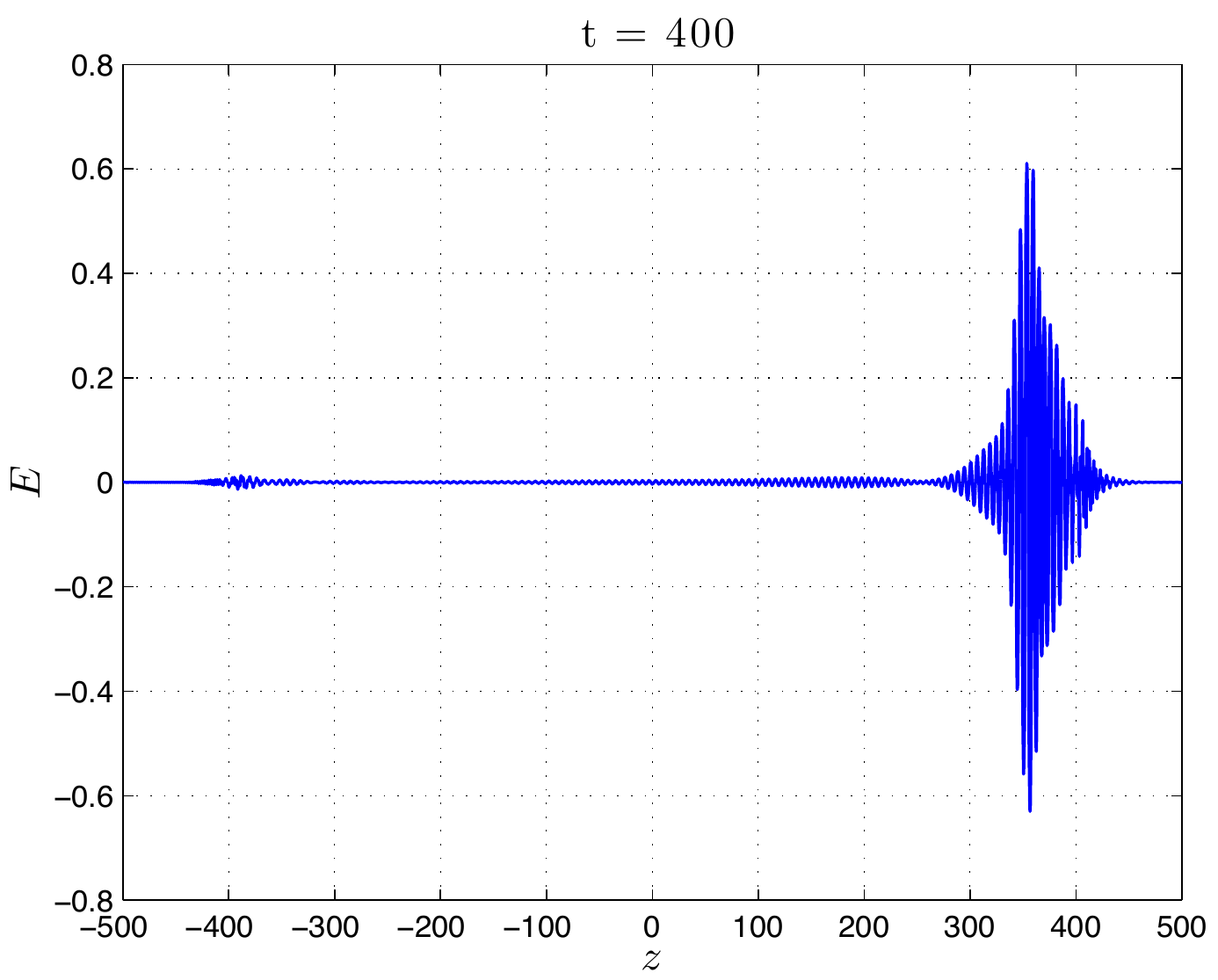}}

    {Simulations with constant refractive index, $N(z) = 0$:}

  \subfigure[]{\includegraphics[width=2.2in]{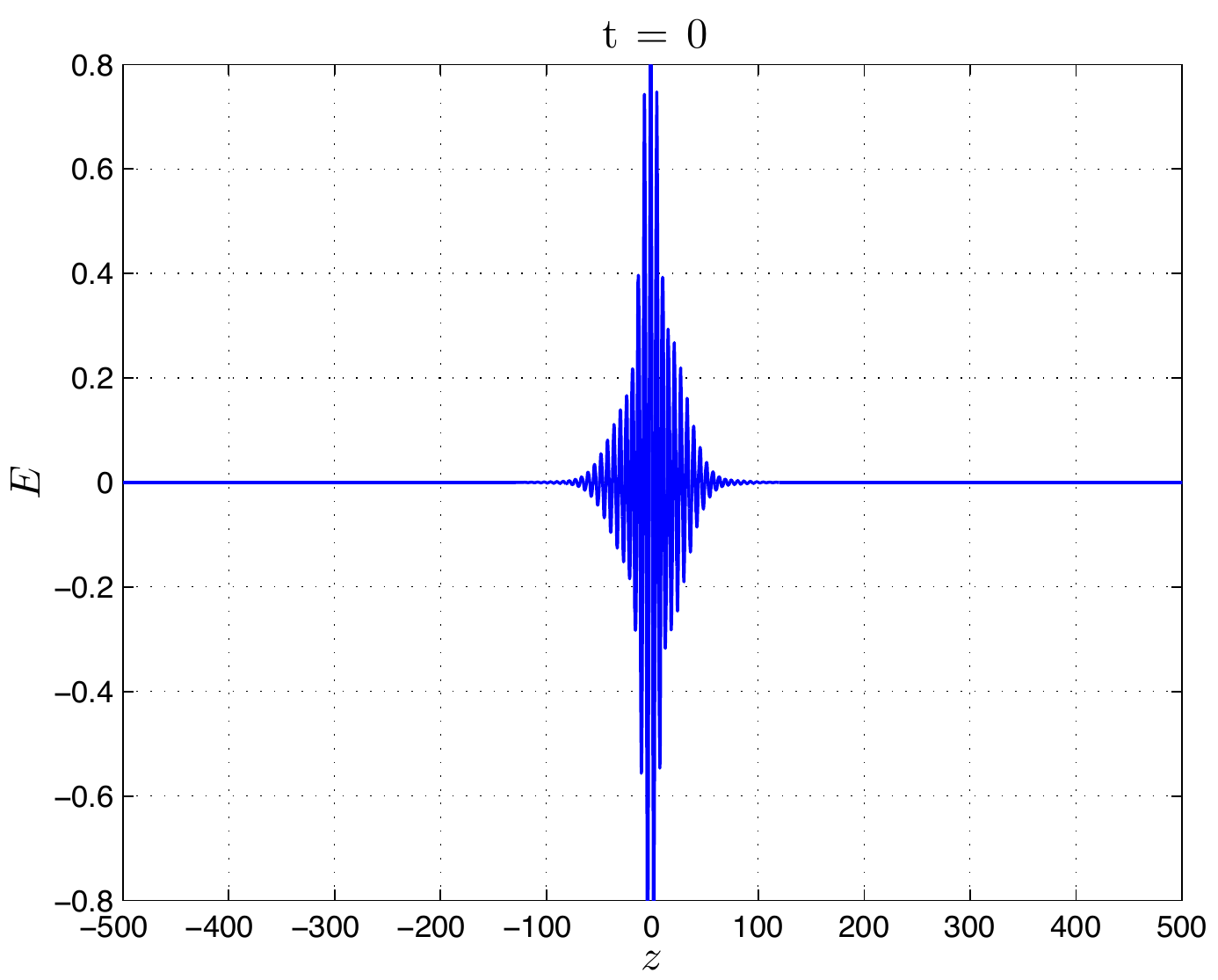}}
  \subfigure[]{\includegraphics[width=2.2in]{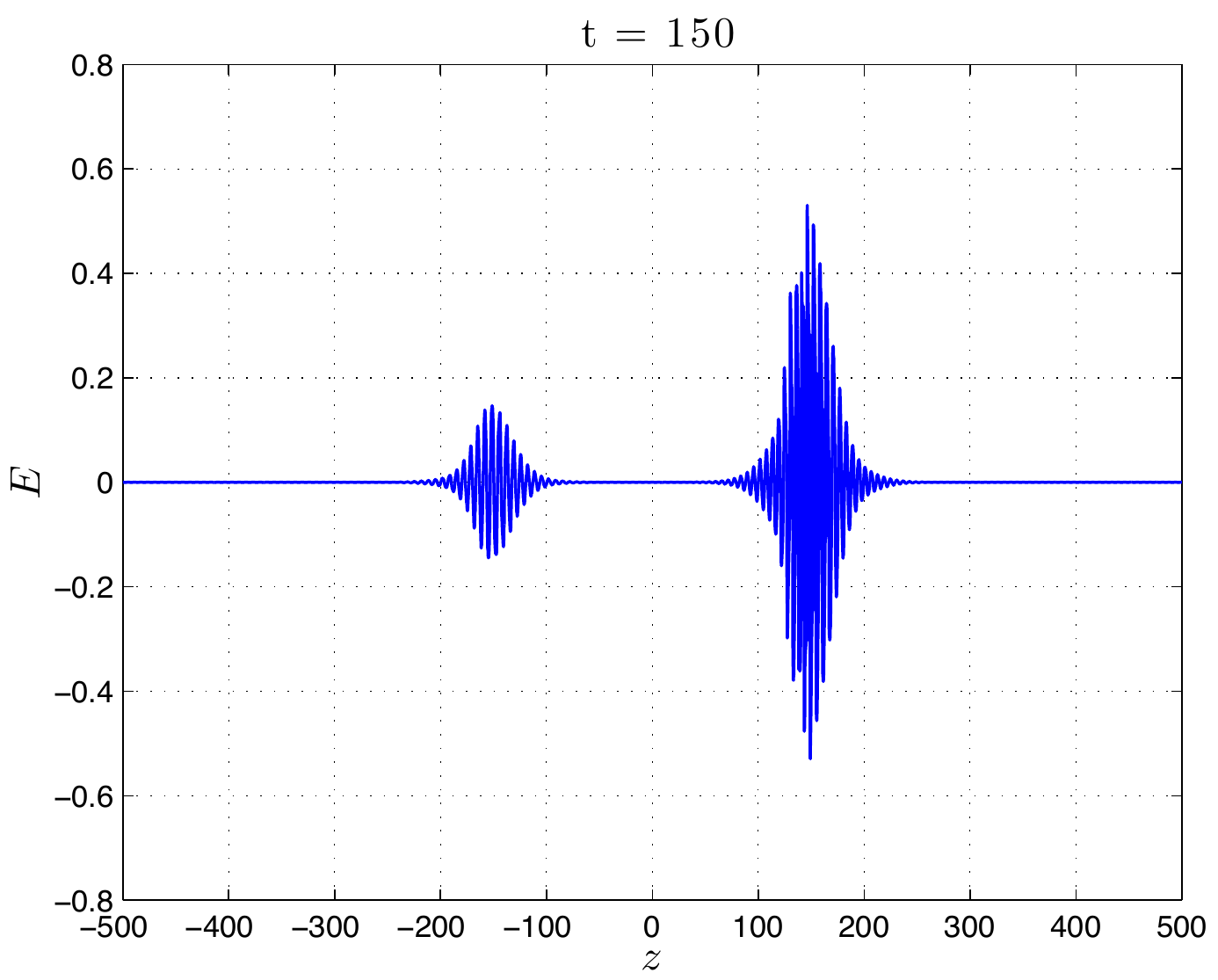}}

  \subfigure[]{\includegraphics[width=2.2in]{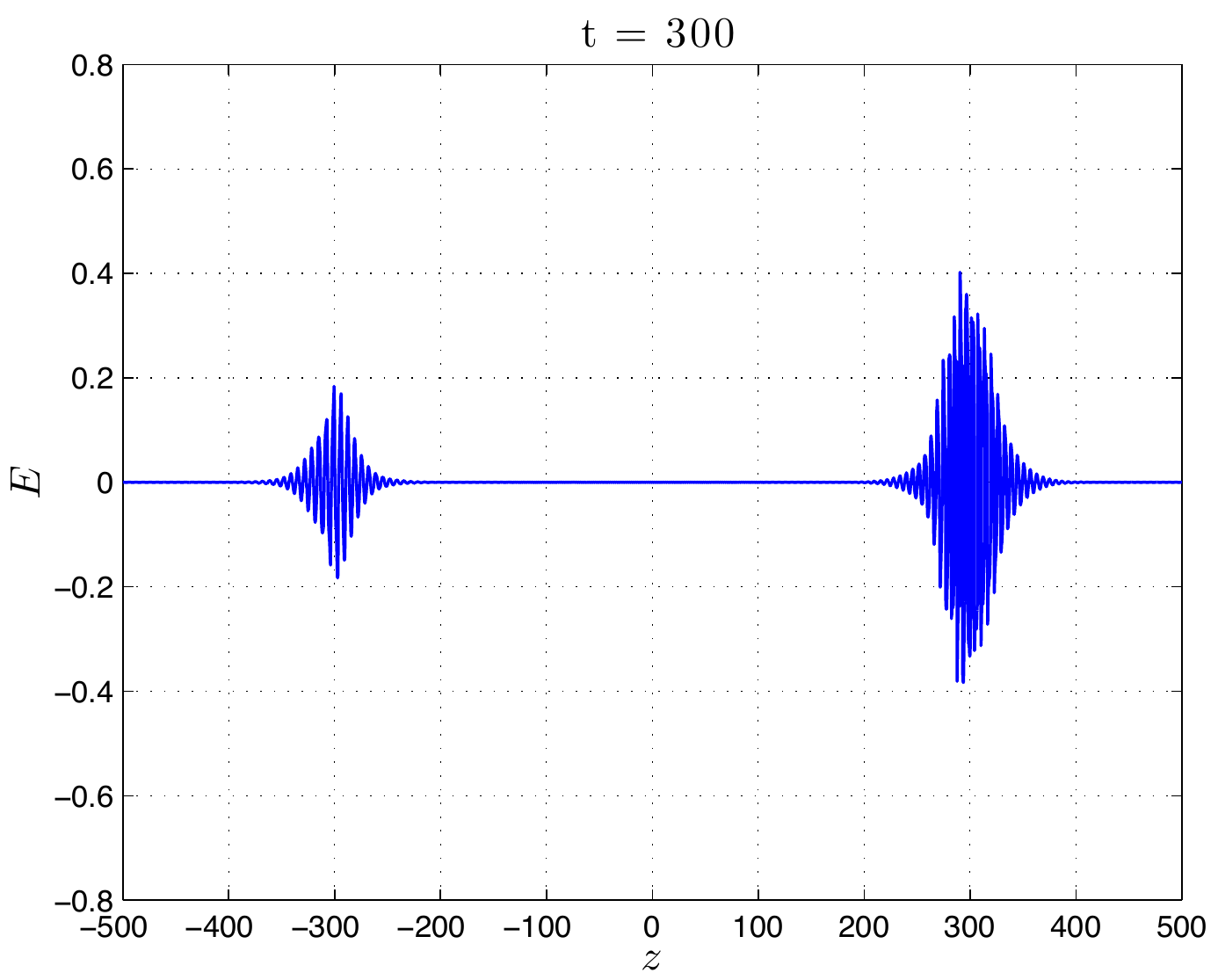}}
  \subfigure[]{\includegraphics[width=2.2in]{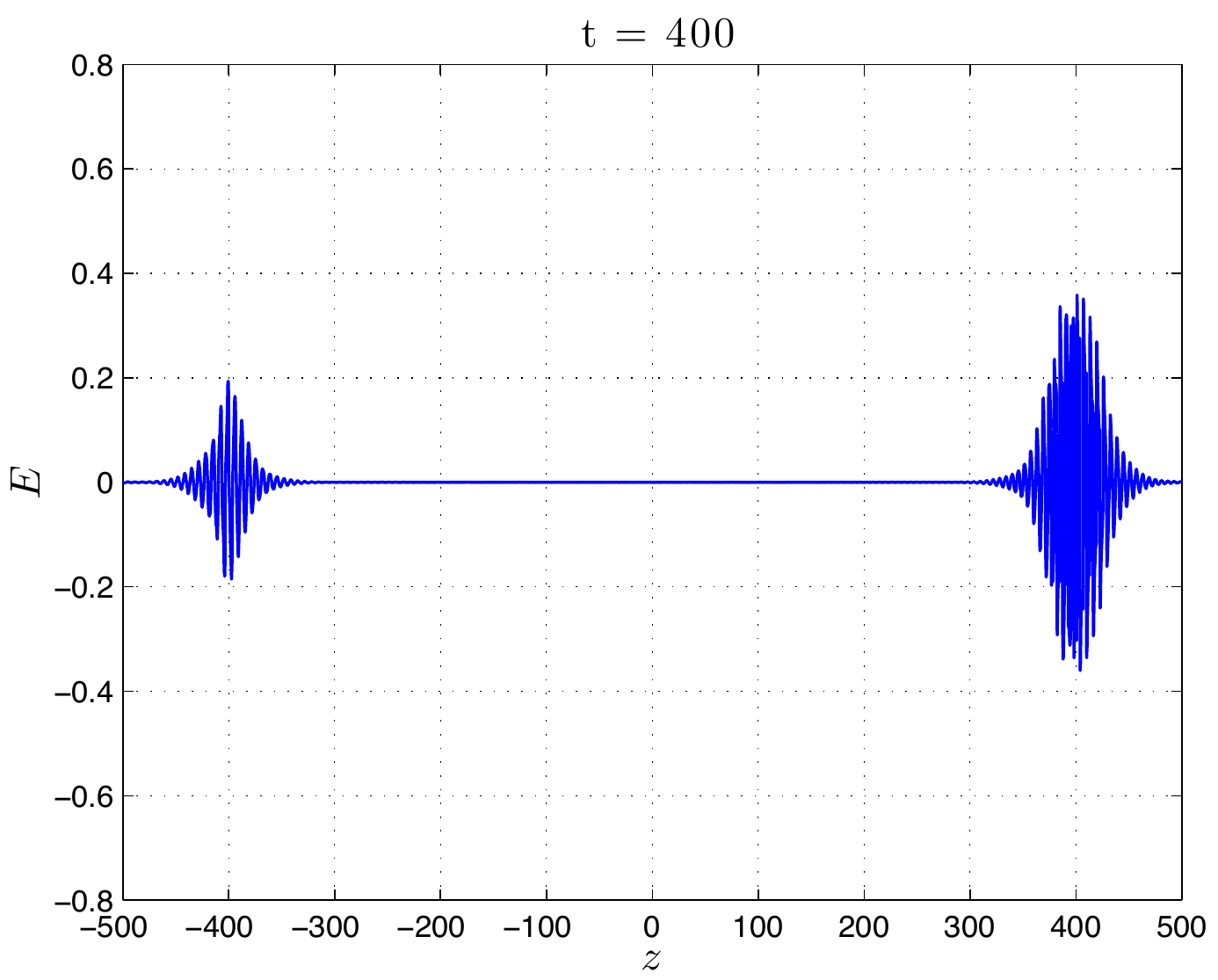}}

  \caption{Rescaled Maxwell equation, \eqref{e:rescaled_maxwell},
    time-evolution for data generated by the NLCME soliton with parameters
    $v=.9$ and $\delta = .9$; see \eqref{eq:NLCME_soliton}.  The
    solutions are computed with 20000 grid points on the domain
    $[-500,500]$. 
    }

  \label{f:traveling_soliton}

\end{figure}

\begin{figure}
  \centering

    {Simulations with varying refractive index, \eqref{e:index}:}

  \subfigure[]{\includegraphics[width=2.2in]{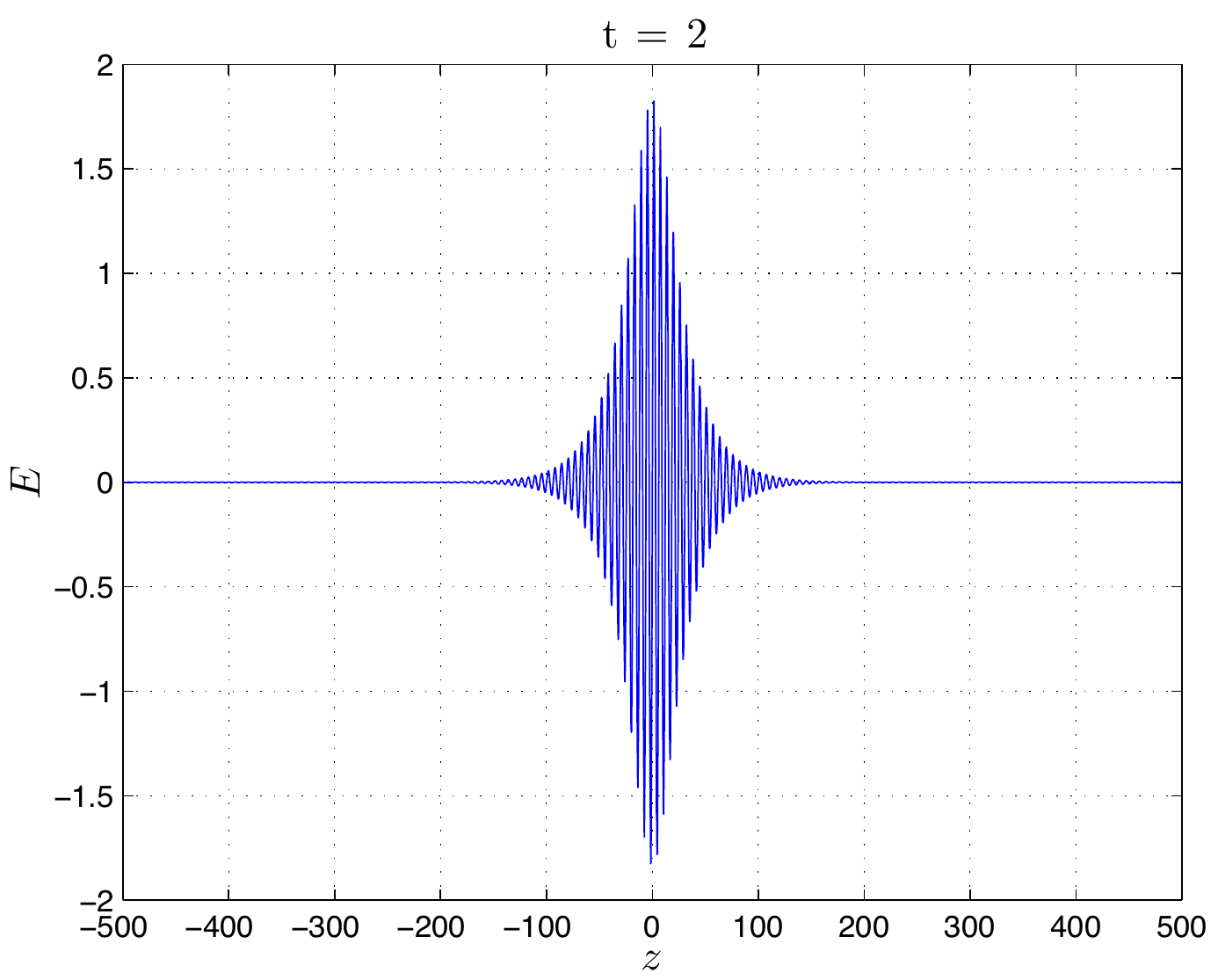}}
  \subfigure[]{\includegraphics[width=2.2in]{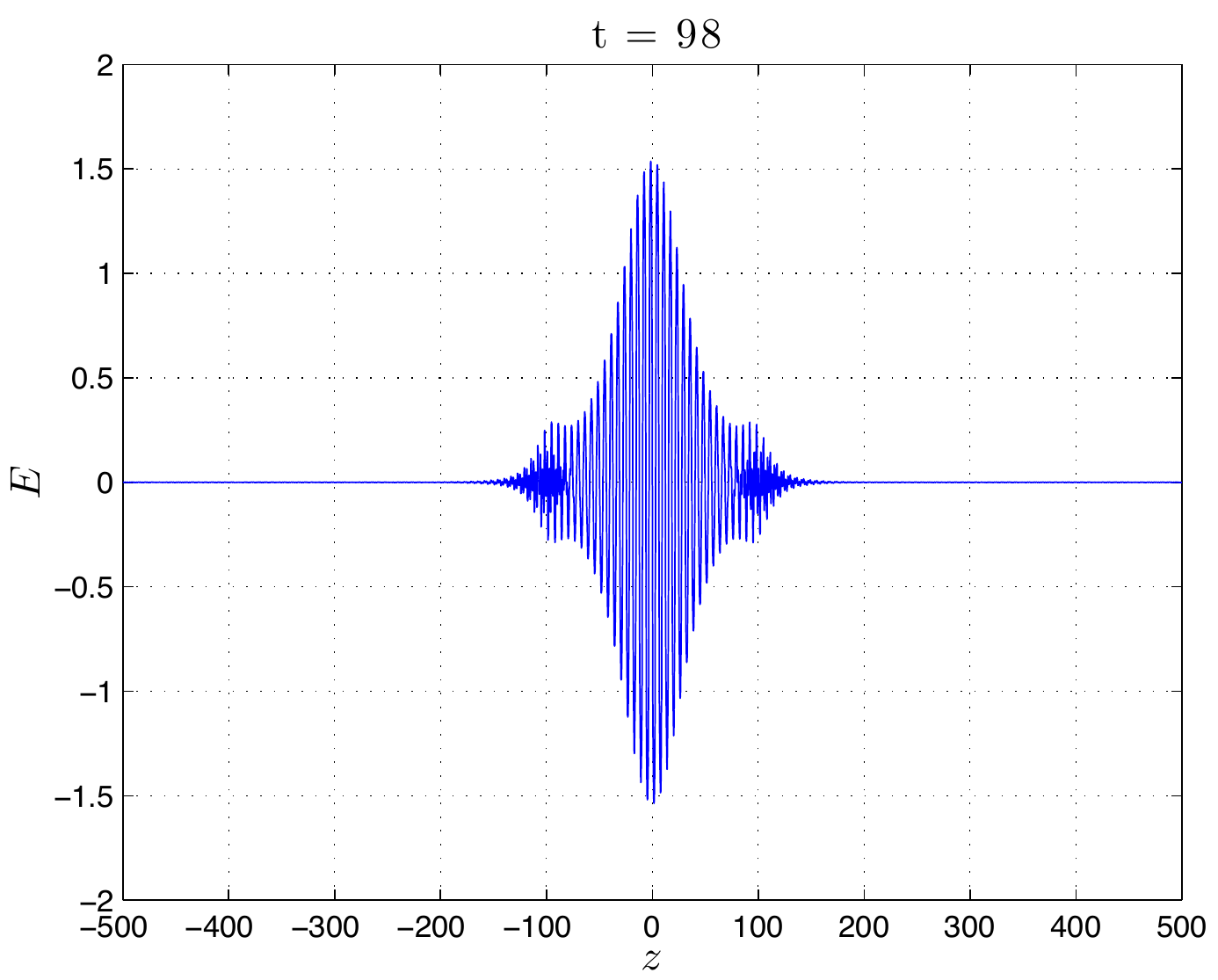}}

  \subfigure[]{\includegraphics[width=2.2in]{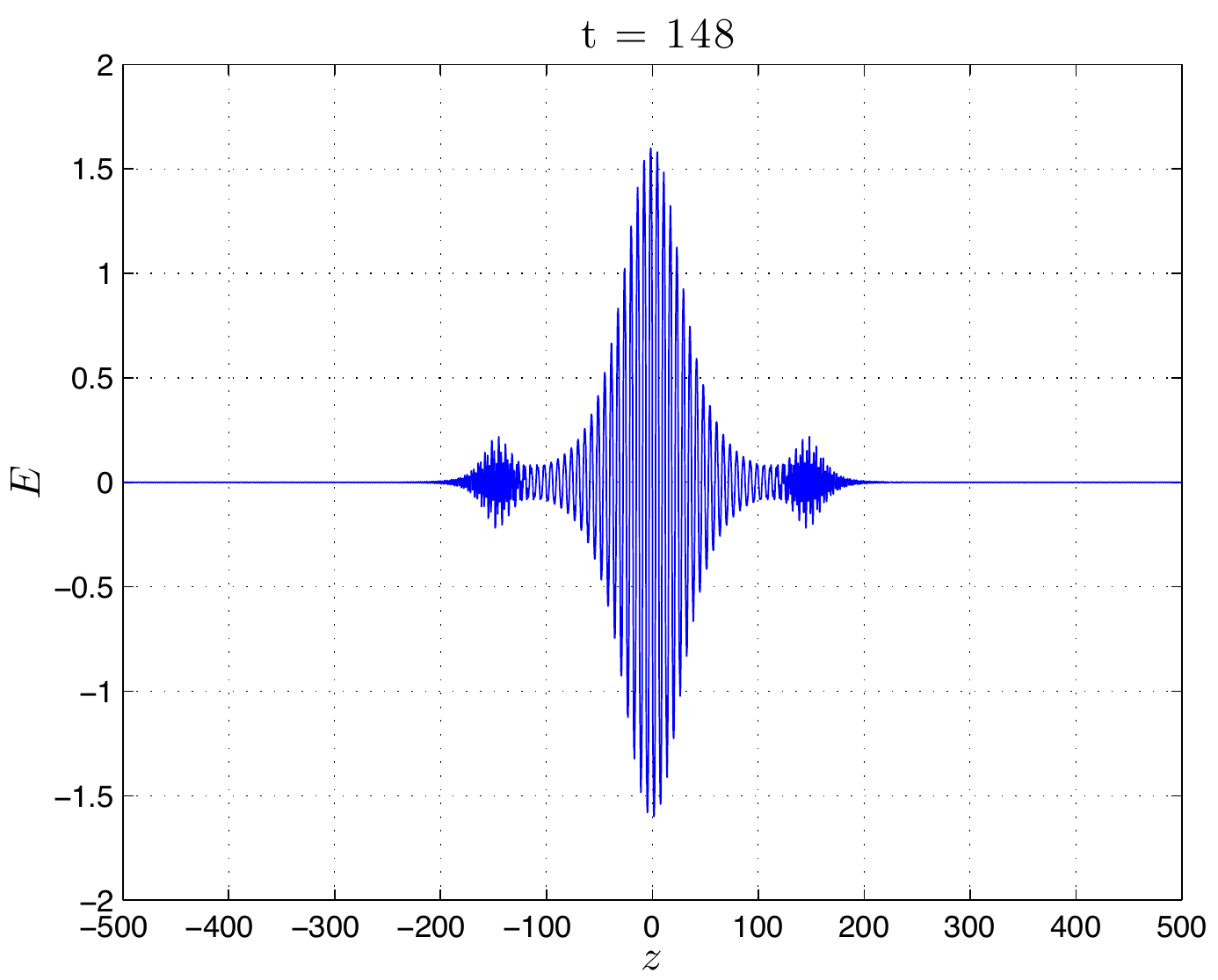}}
  \subfigure[]{\includegraphics[width=2.2in]{figs/_plots_soliton2ncos1_print_20000/matlabfig396}}

    {Simulations with constant refractive index, $N(z) = 0$:}

  \subfigure[]{\includegraphics[width=2.2in]{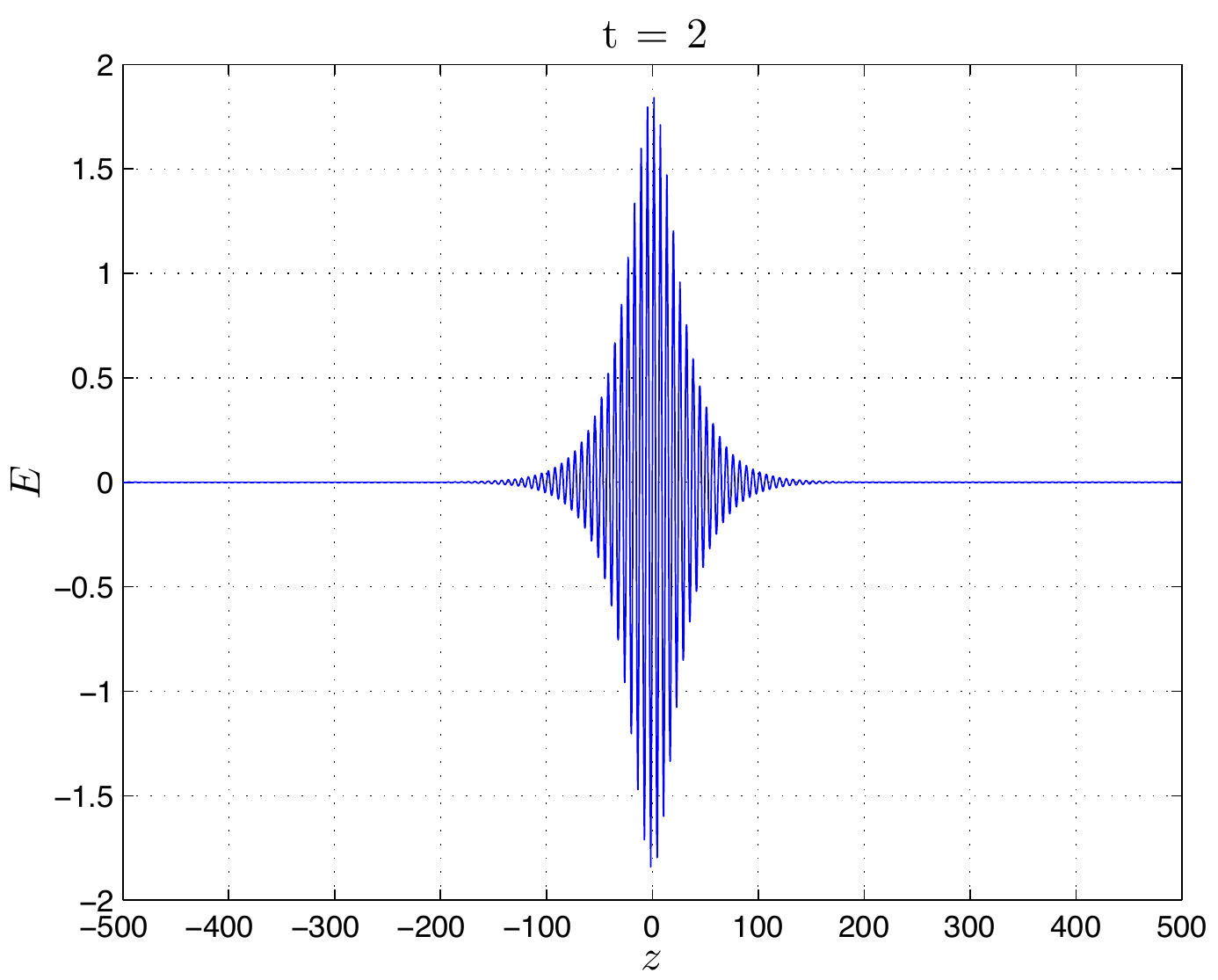}}
  \subfigure[]{\includegraphics[width=2.2in]{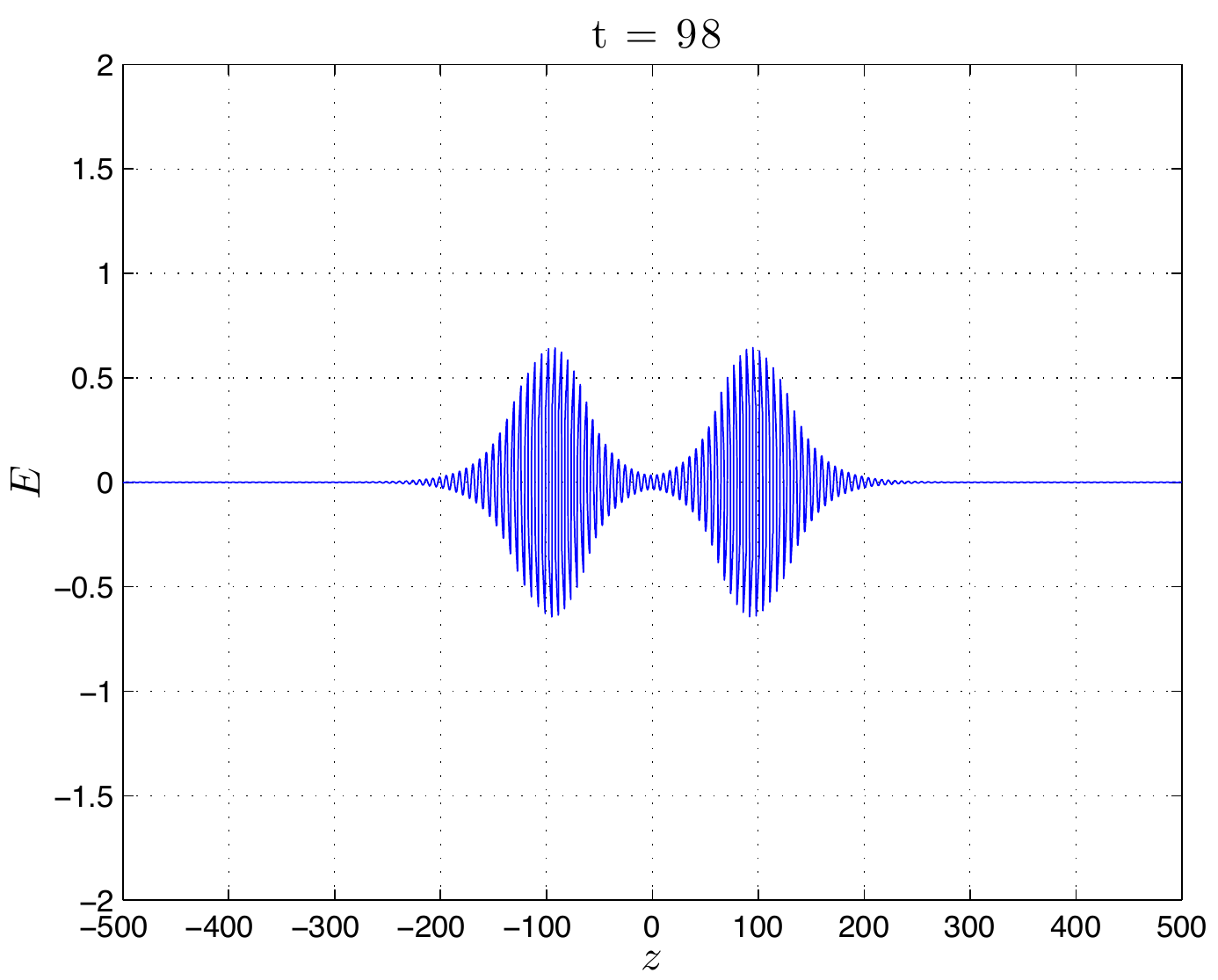}}

  \subfigure[]{\includegraphics[width=2.2in]{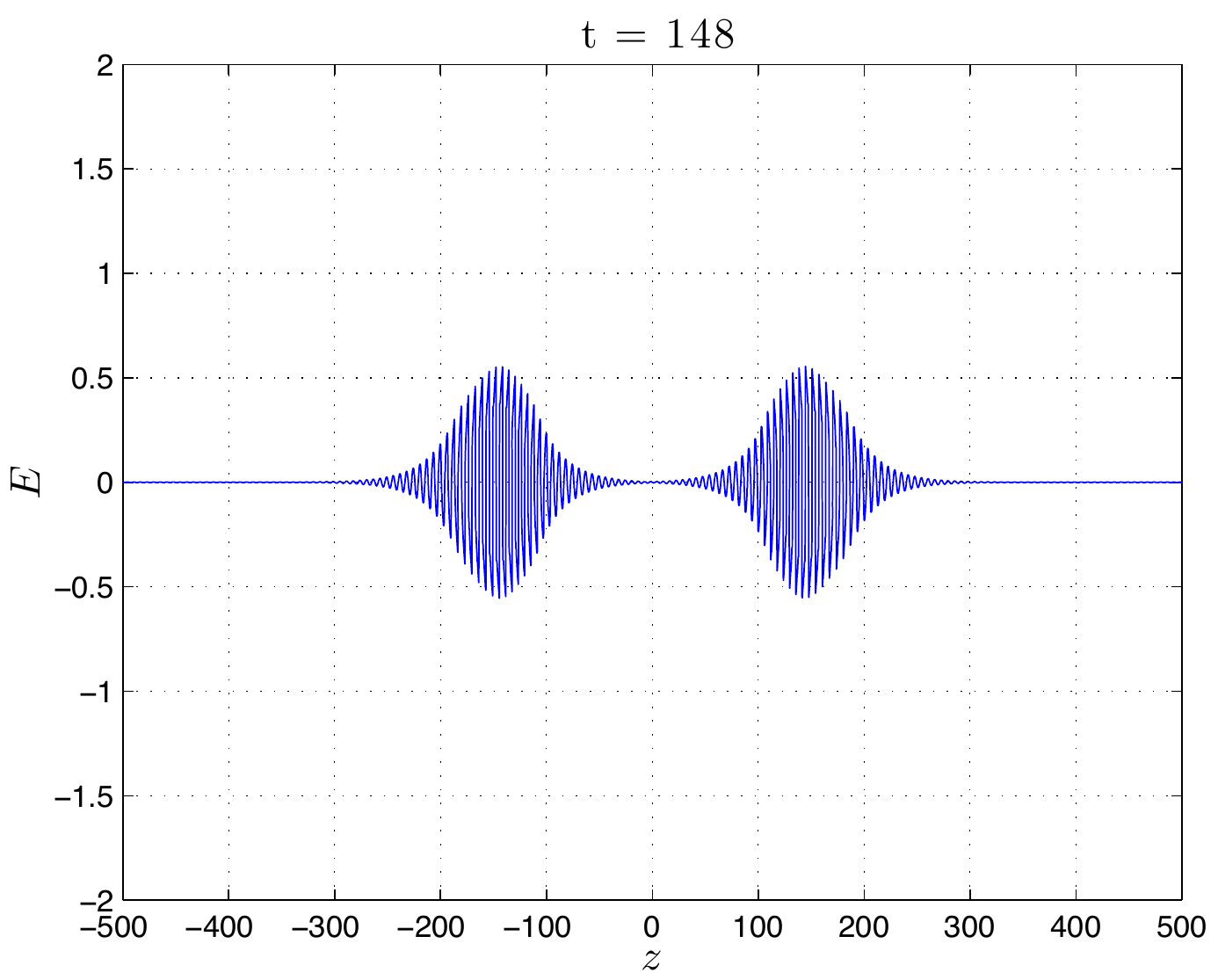}}
  \subfigure[]{\includegraphics[width=2.2in]{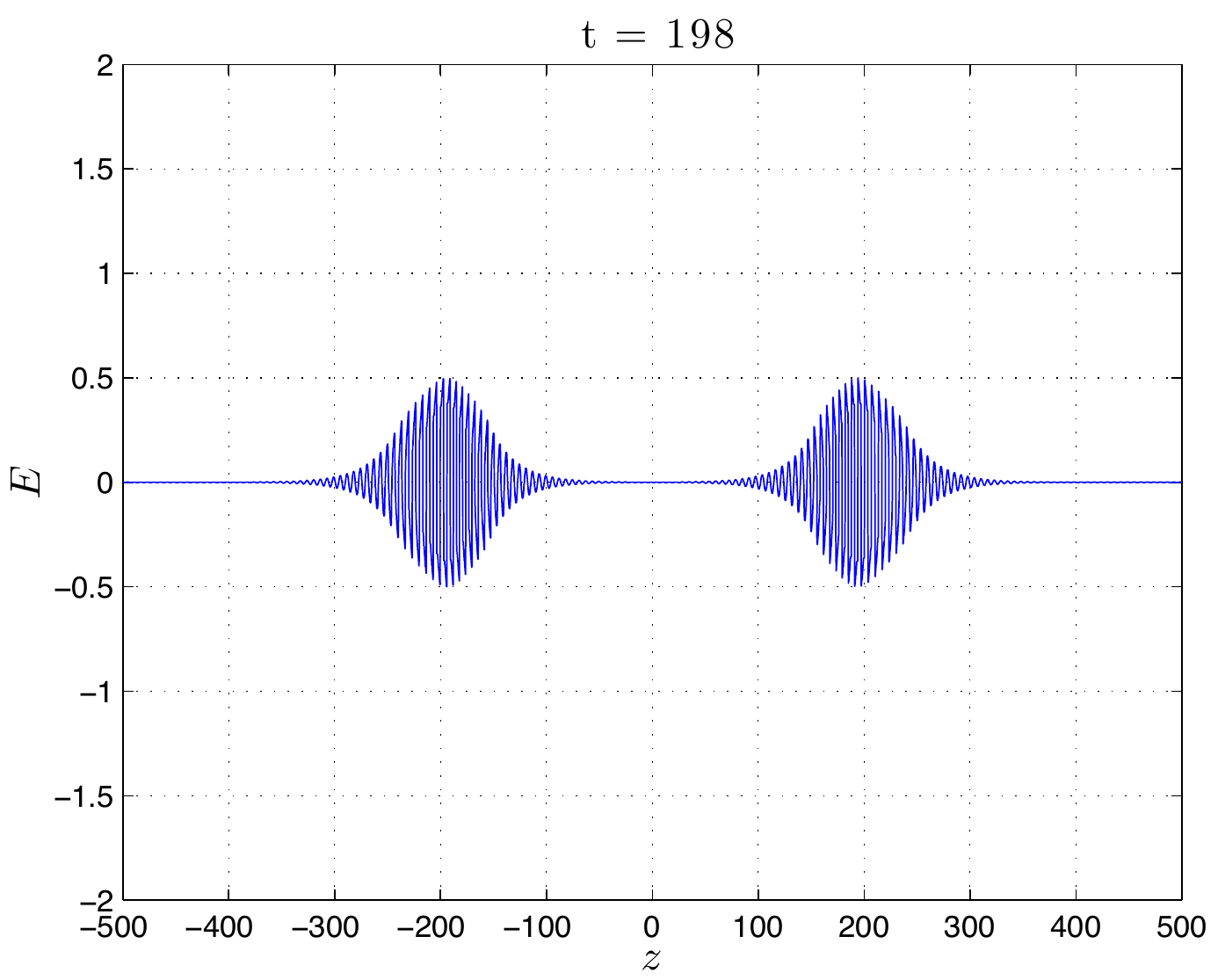}}

  \caption{Solution of rescaled nonlinear periodic Maxwell equation,
    \eqref{e:rescaled_maxwell}, for initial data generated by the
    NLCME soliton with parameters $v=0$ and $\delta =\pi/2$; see
    \eqref{eq:NLCME_soliton}.The solutions are computed with
    20000 grid points on the domain $[-500,500]$.
}
  \label{f:standing_soliton}
\end{figure}

\begin{figure}
  \centering
  \subfigure[]{\includegraphics[width=2.35in]{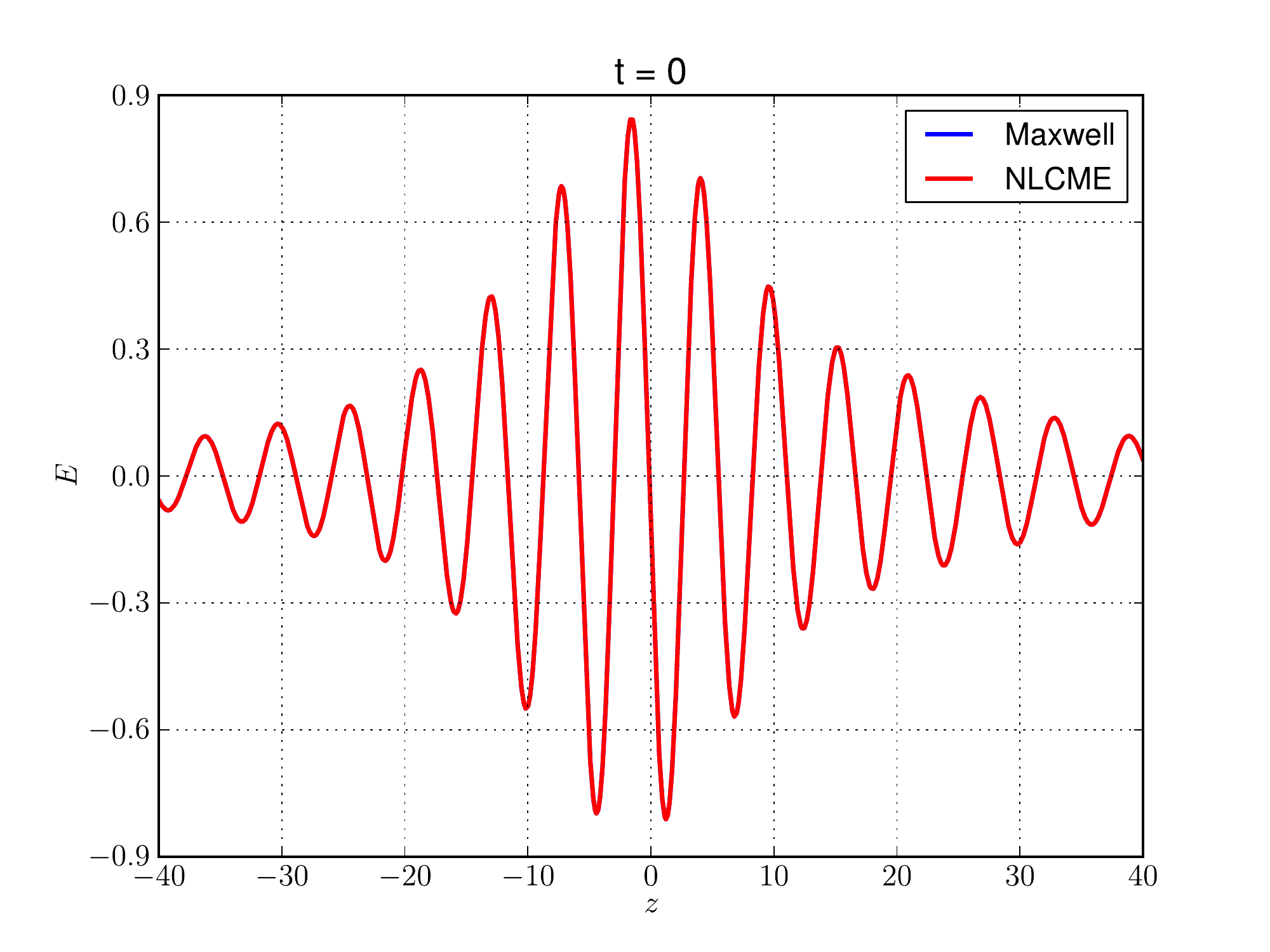}}
  \subfigure[]{\includegraphics[width=2.35in]{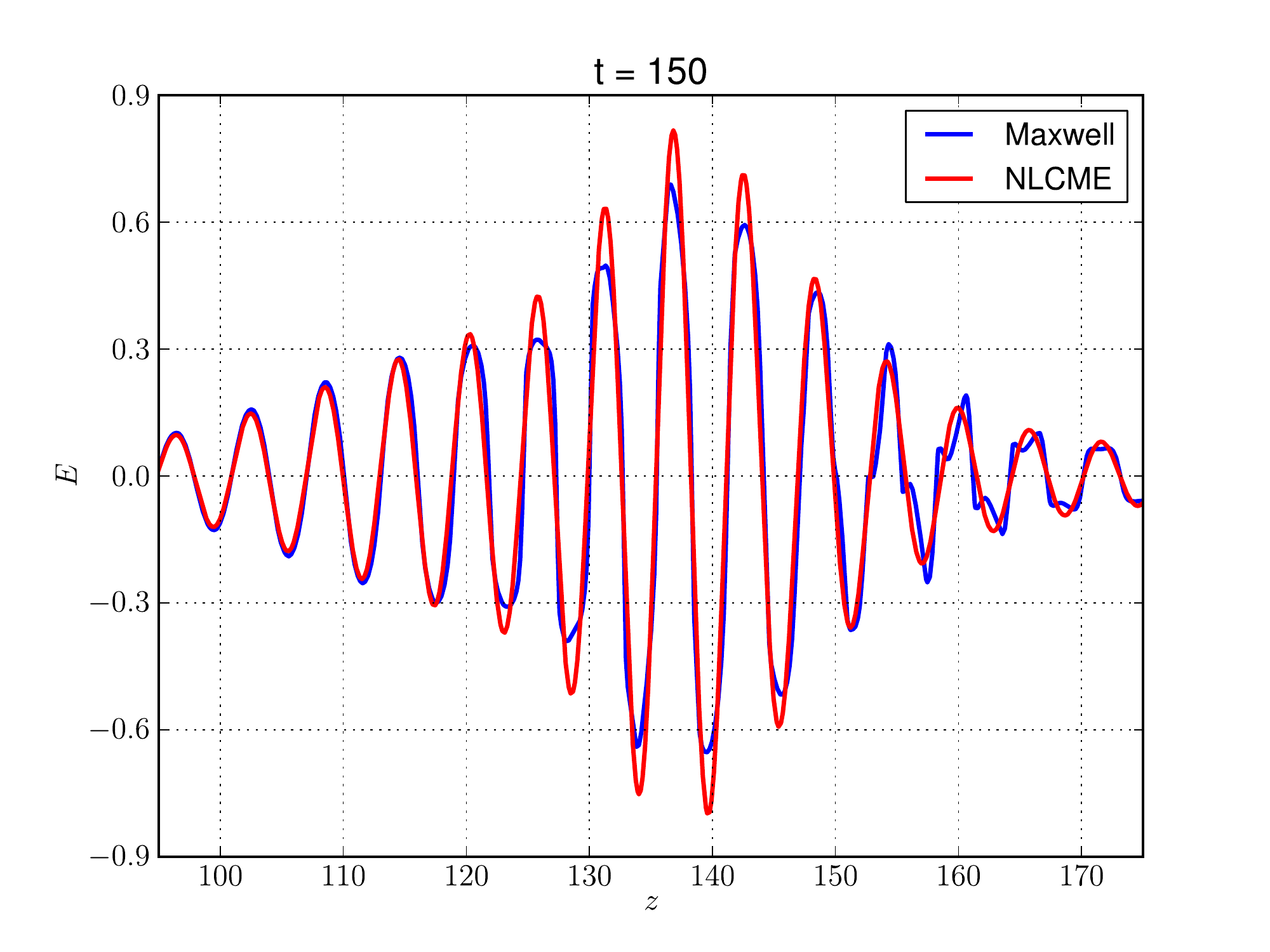}}

  \subfigure[]{\includegraphics[width=2.35in]{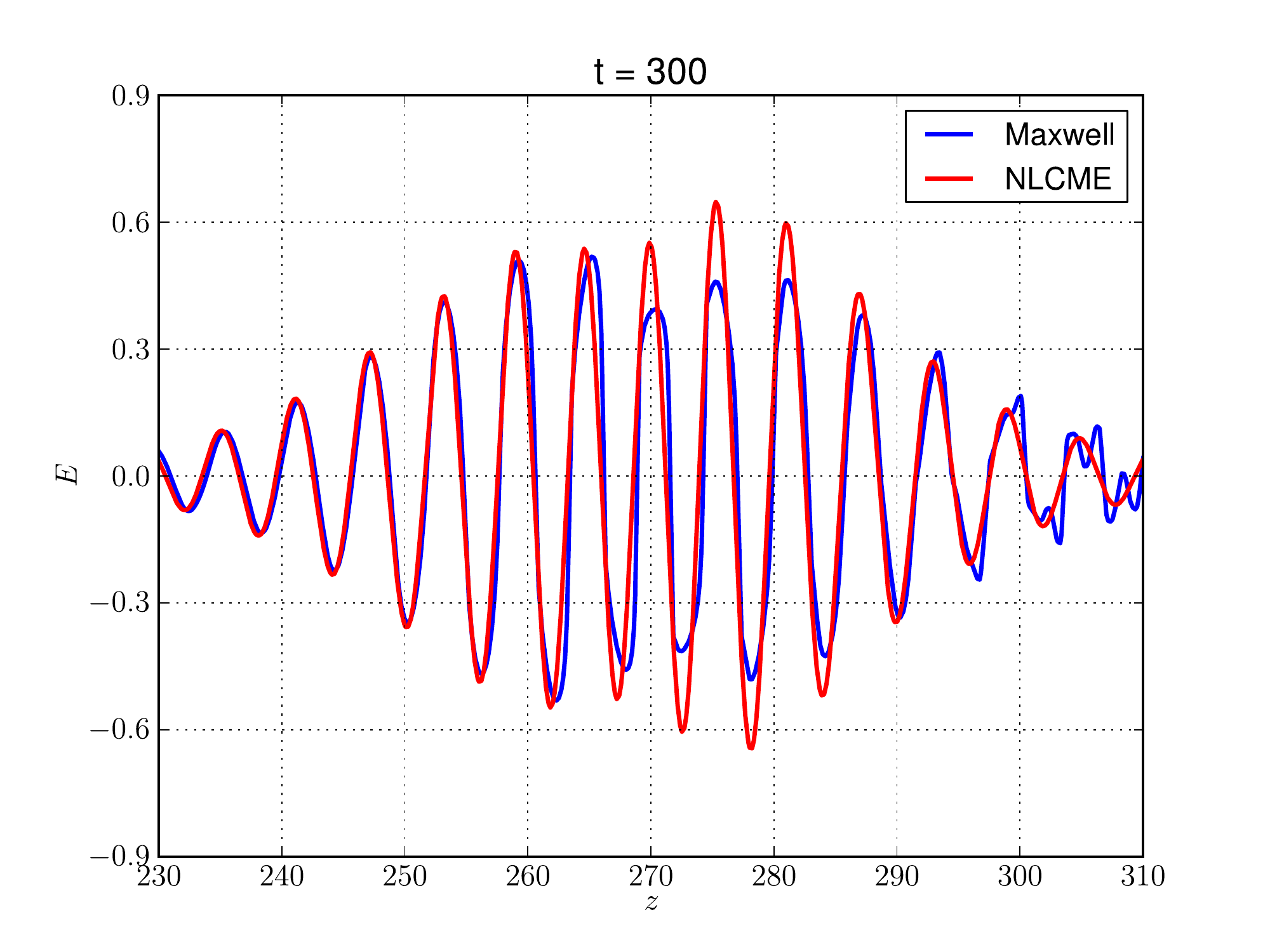}}
  \subfigure[]{\includegraphics[width=2.35in]{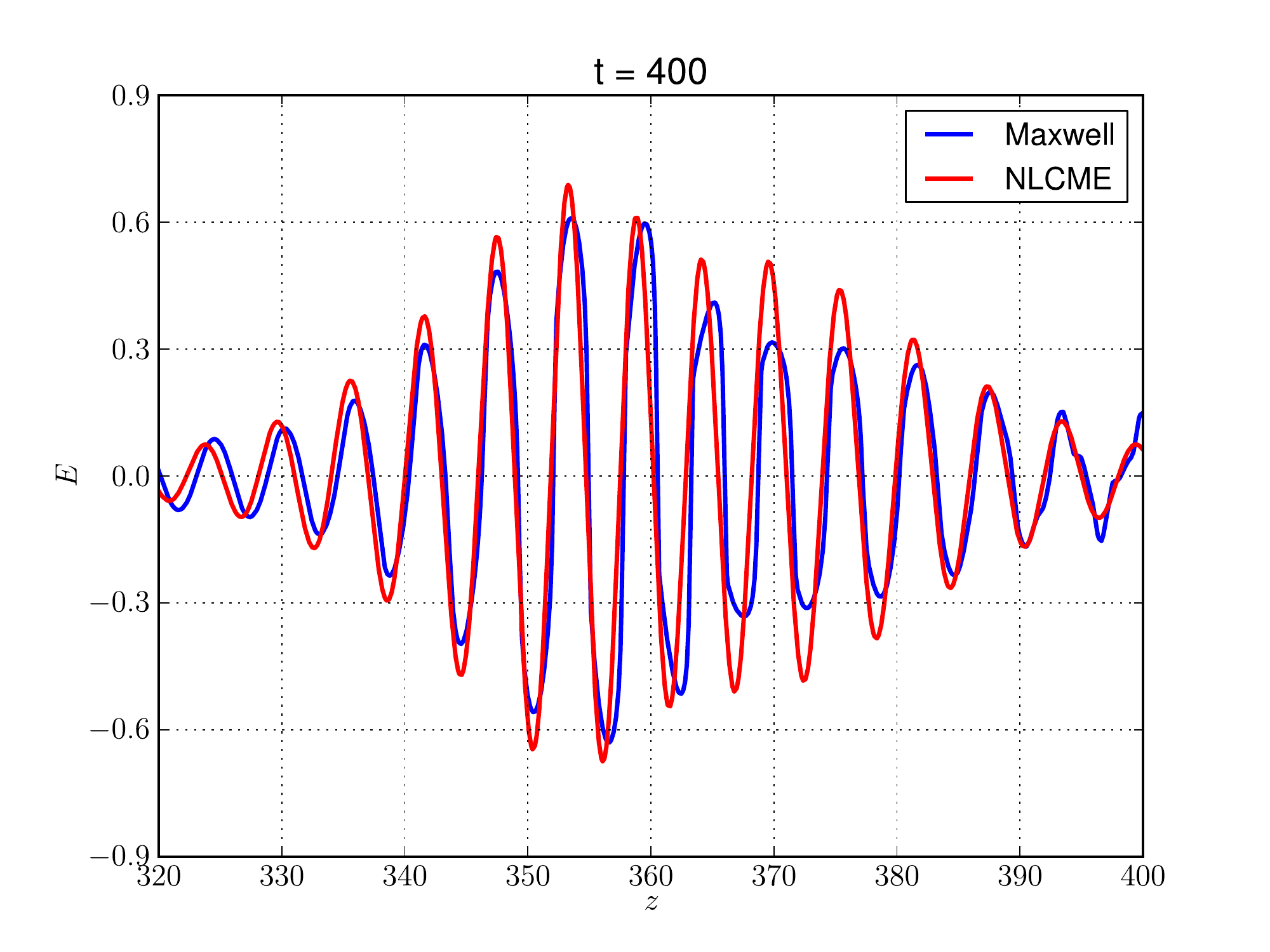}}
  \caption{Comparison of the solution appearing in Figure
    \ref{f:traveling_soliton} a - d, with the exact NLCME soliton.}
  \label{f:traveling_soliton_zoom}
\end{figure}

We note that it is also essential that the data be properly prepared to see a
persistence of localization.  For the initial condition
\begin{subequations}\label{e:sech_data}
  \begin{align}
    D &=  0.5  \cos(z) \sech(\eps z),\\
    B & = - D,
  \end{align}
\end{subequations}
we see in Figure \ref{fig:primitivesech} substantial spreading.  This
data mimics the gap soliton's amplitude, slowly varying envelope, and
carrier wave, but is apparently too far outside the basin of
attraction to converge to a localized state.  Similar results were
observed with Gaussian wave packet initial conditions.

\begin{figure}
  \centering
  \subfigure[]{\includegraphics[width=2.35in]{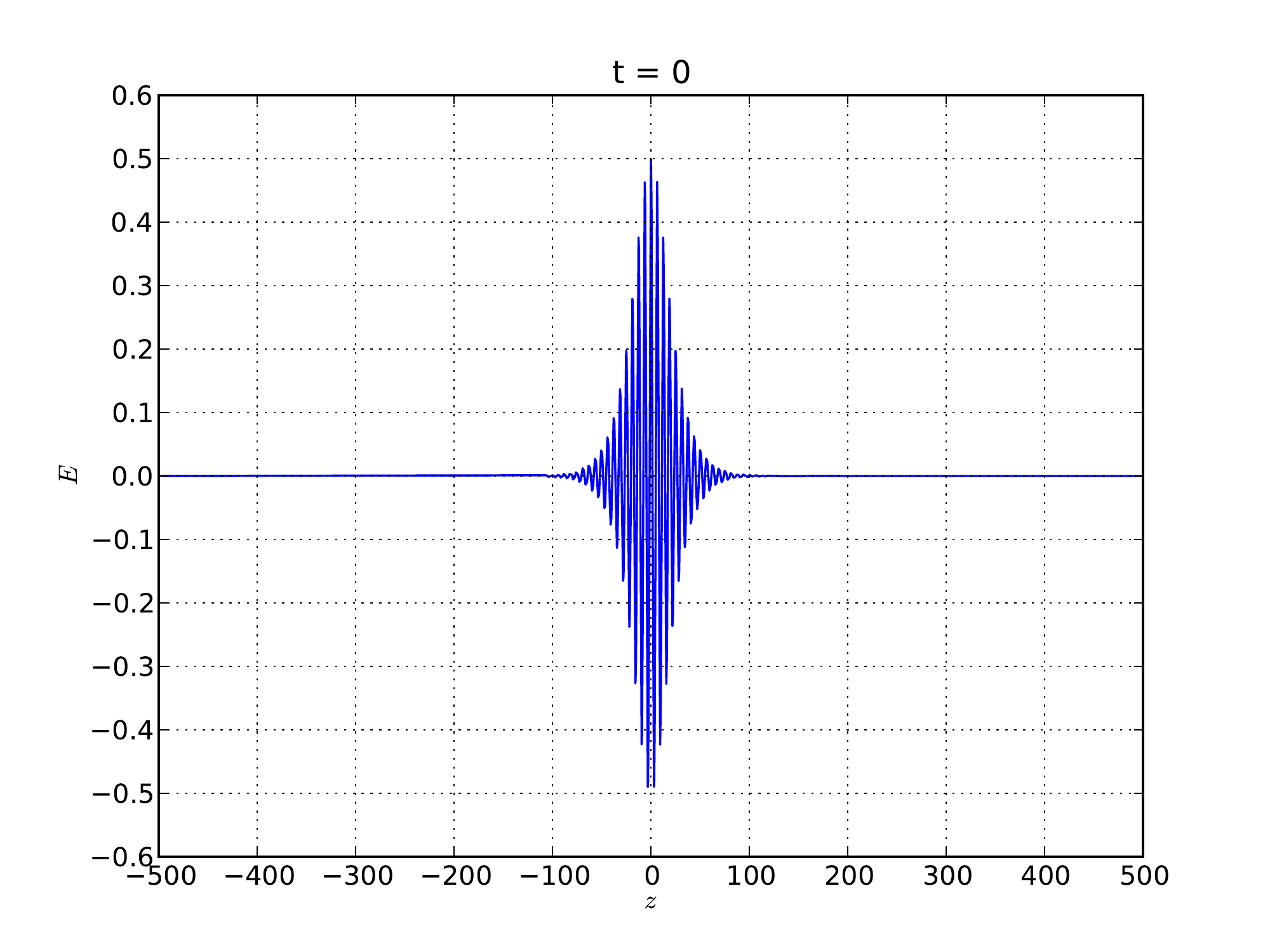}}
  \subfigure[]{\includegraphics[width=2.35in]{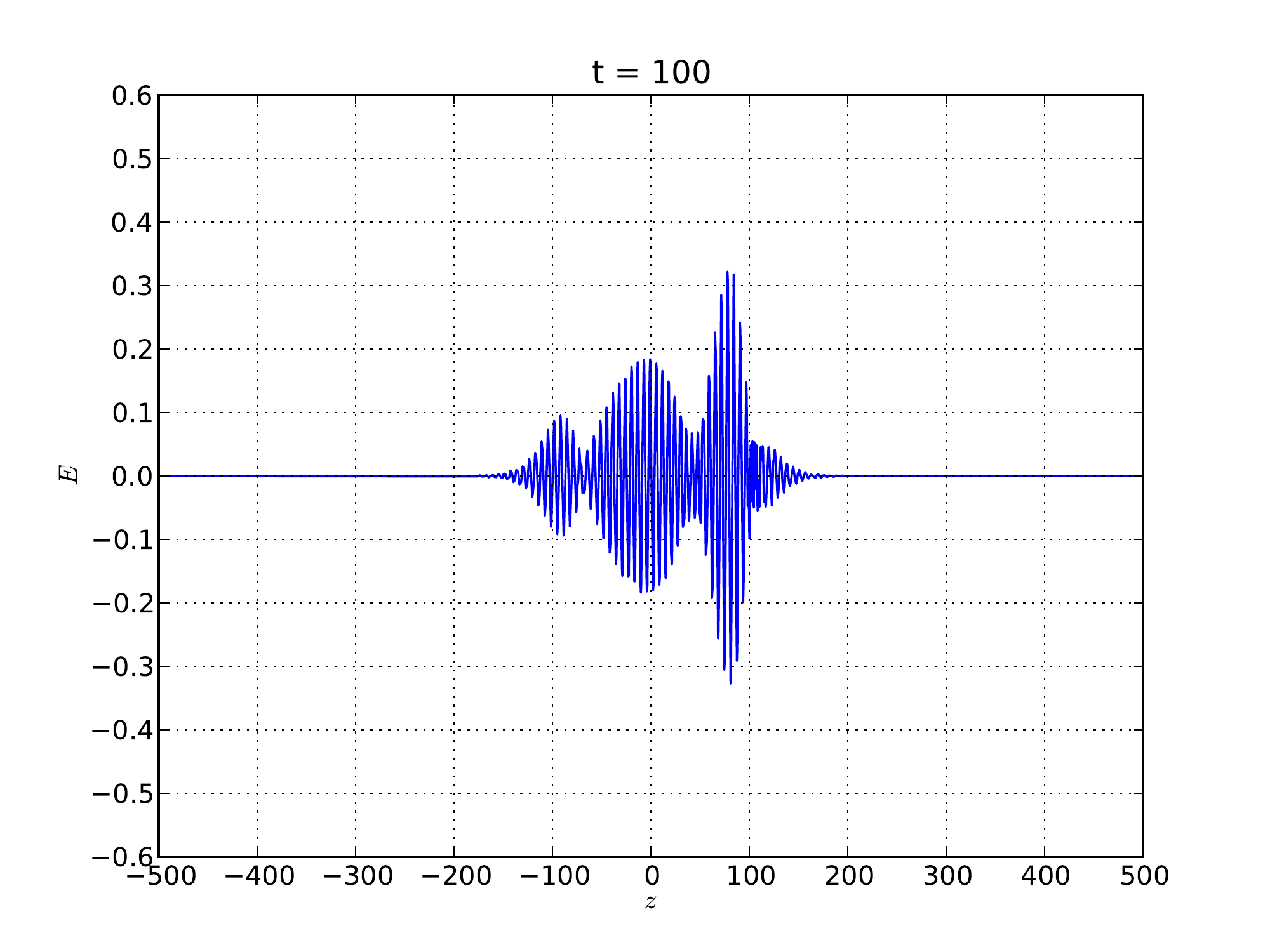}}

  \subfigure[]{\includegraphics[width=2.35in]{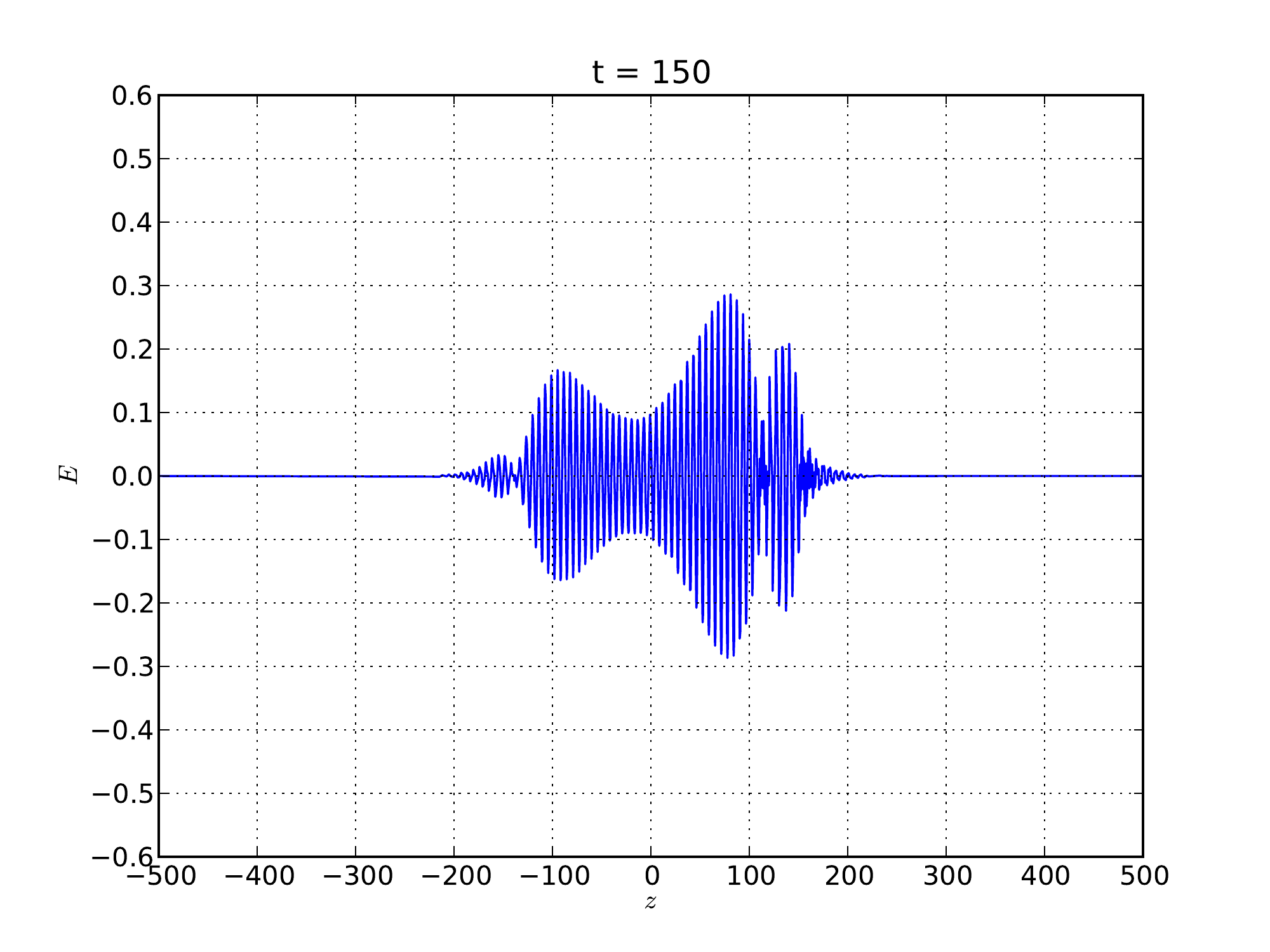}}
  \subfigure[]{\includegraphics[width=2.35in]{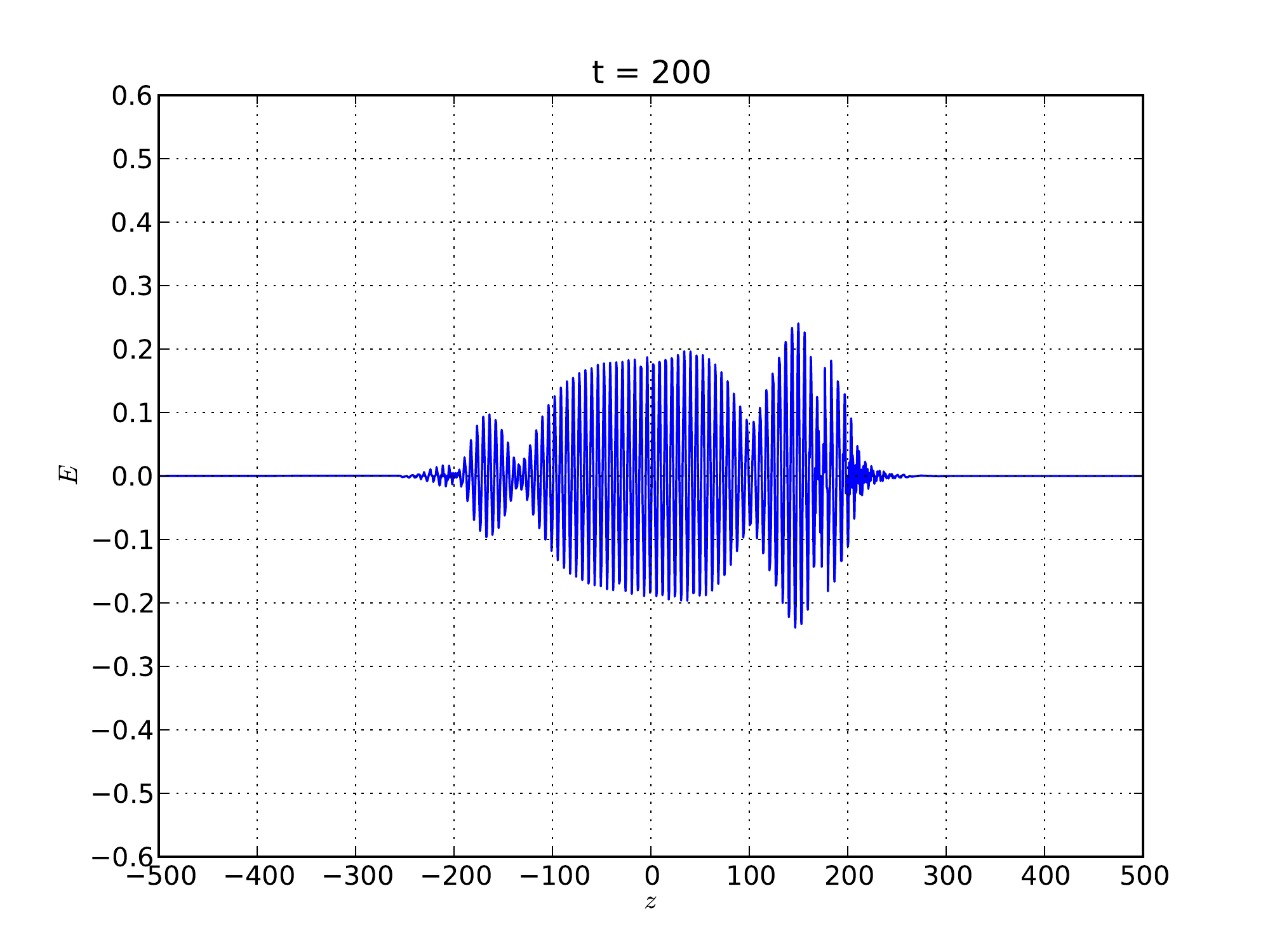}}

  \caption{Solution of rescaled nonlinear periodic Maxwell equation,
    \eqref{e:rescaled_maxwell} with periodic refractive index
    \eqref{e:index}, for initial data \eqref{e:sech_data} .  In
    contrast to the NLCME soliton data, the shape of the solution does
    not persist. The solution is computed with 20000 grid points on
    the domain $[-500,500]$.}
  \label{fig:primitivesech}

\end{figure}

\subsection{ Envelope carrier-shock trains}
\label{s:maxwell_shocks}

Although ther the slowly varying NLCME envelope
{\it shape} is robust, for the nonlinear Maxwell time-evolution, 
 there is evidence of nonlinear steepening and shock formation
on the short (carrier) microstructure spatial scale.  Thus, the nearly monochromatic
slowly varying envelope approximation of NLCME is violated.

 Figure \ref{f:shock_comparison} displays the time-evolution  for 
  (a) moving and (b) stationary  NLCME - gap soliton data.
  For each initial condition, the nonlinear Maxwell evolution is simulated 
  for different grid spacings. As we increase the number of
grid points, sharp features are better resolved by the shock capturing
algorithm.  One can also examine the Fourier transform of the output
to see that we obtain an algebraically decaying solution in wave
number, with peaks at the integer wave number values.

In summary, our observations support the emergence of an
{\it envelope carrier-shock train};
persistence of
coherent, slowly varying, wave envelope 
and shock formation on the carrier scale.

  \begin{figure}
    \centering

    \subfigure[$v=.9$, $\delta =
    .9$]{\includegraphics[width=2.4in]{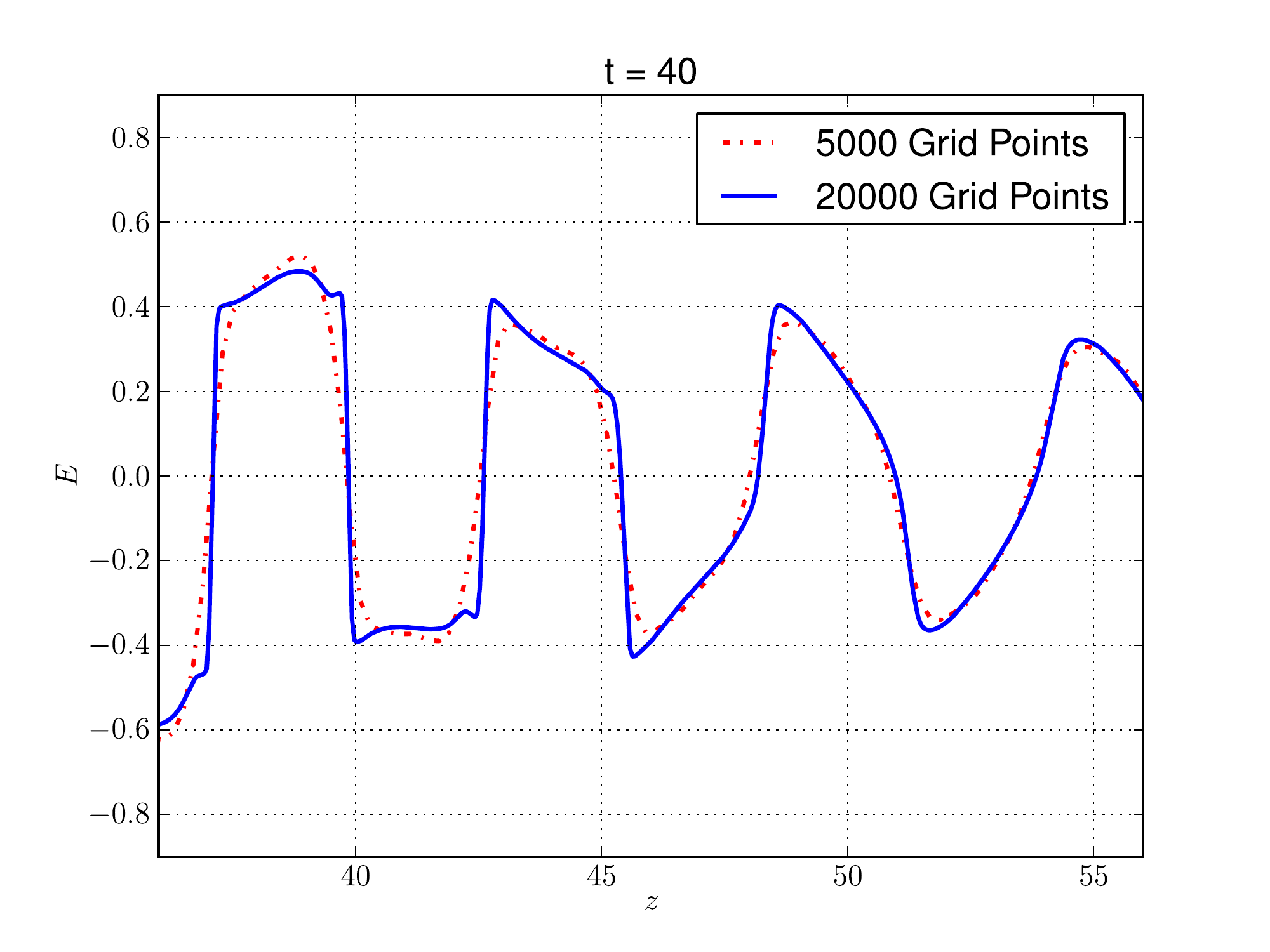}}
    \subfigure[$v=0$, $\delta =
    \pi/2$]{\includegraphics[width=2.4in]{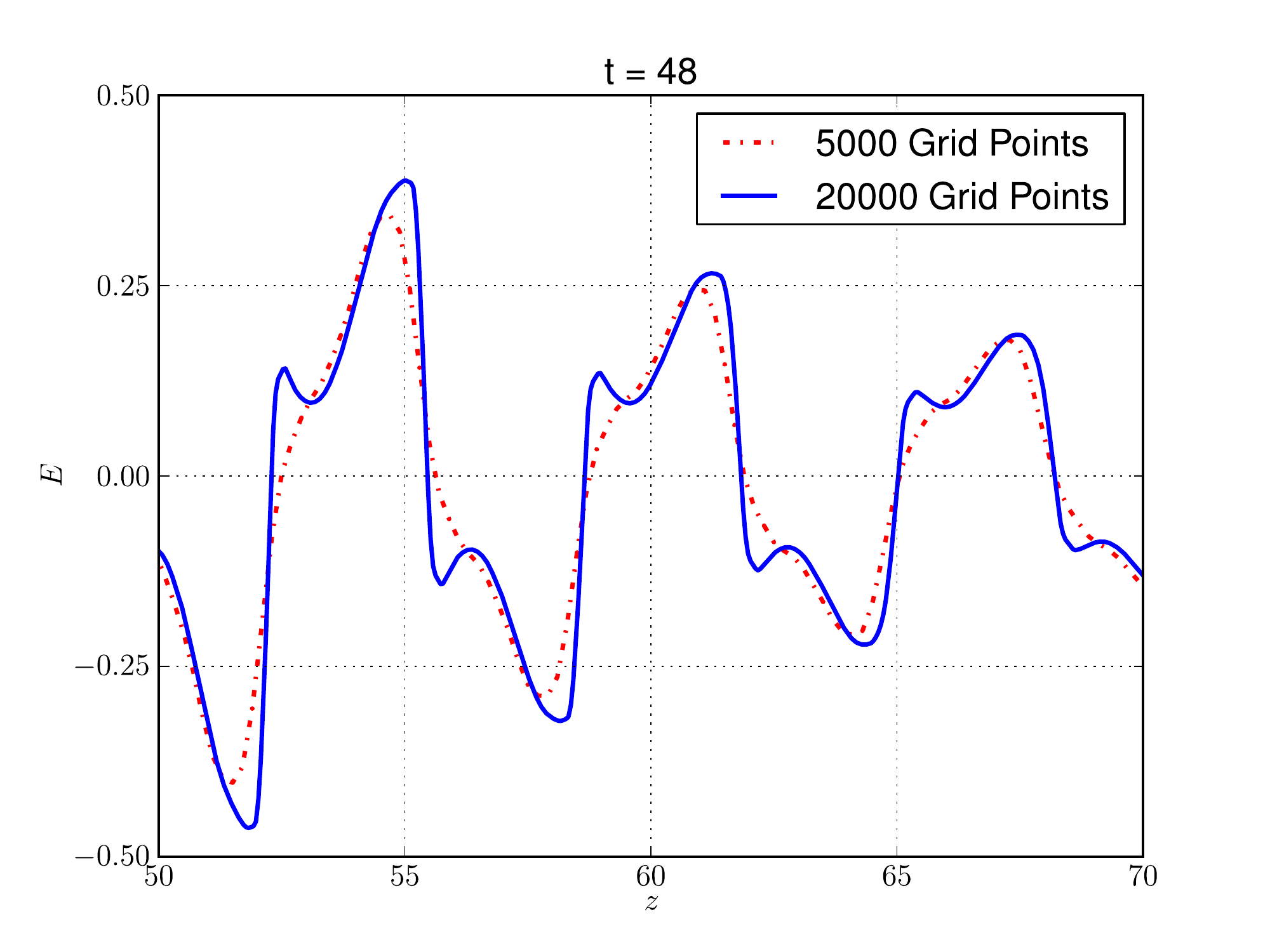}}

    \caption{Increasing the number of grids points better resolves the
      shocks in the carrier wave. For NLCME soliton data with the
      indicated $v$ and $\delta$ (see \eqref{eq:NLCME_soliton}), with
      index of refraction $N(z)$ given by \eqref{e:index}).}
    \label{f:shock_comparison}
  \end{figure}

  \section{Resonant nonlinear geometrical optics and nonlinear
    spatially inhomogeneous Maxwell equations}
  \label{sec:asympt}

   In this section we
  derive a system of equations,  which incorporates 
  all wave-resonances and which our numerical simulations show,
  captures the key features of the nonlinear Maxwell time-evolution,
   in particular, the presence of robust envelope carrier-shock train solutions.
  We derive this system, for general non-homogeneous media,  using a nonlinear geometrical optics expansion;
  see, for example,
  \cite{Hunter:1983vn,Hunter:1986kx,Majda:1984uq}. The equations
  obtained are the general integro-differential equations
  \eqref{eq:integro_diff_intro}. In the case of a periodic medium,
  they reduce to an infinite set of local equations, which we
  call the {\it extended nonlinear coupled mode equations}
  (xNLCME). If, in xNLCME, we neglect all but principle resonances,
   xNLCME
  reduces to NLCME.
 
  As we shall see, in our numerical simulations of increasing high
  dimensional truncations of xNLCME (section \ref{s:xnlcme_sims}),
  this theory appears to accommodate the observed carrier shocks and
  large scale coherent structures.

\subsection{Nonlinear geometric optics expansion}
In contrast to the ansatz of Section \ref{sec:nlcme}, we assume the
more general form
\begin{equation}
  \label{eq:expansion_ansatz}
  \mathbf{u}(z,t) = \mathbf{u}^{(0)}(z,t,Z,T) + \eps \mathbf{u}^{(1)}(z,t,Z,T)+\eps^2 \mathbf{u}^{(2)}(z,t,Z,T)+\ldots.
\end{equation}
where ${\bf u} = ( E, B)^T$ and $Z = \eps z$, $T = \eps t$.  Inserting
\eqref{eq:expansion_ansatz} into \eqref{e:rescaled_maxwell},
\eqref{e:rescaled_closure}, the first order system
\begin{equation*}
  \dt \begin{pmatrix} n(z)^2E + \eps \chi E^3 \\ B\end{pmatrix} +
  \dz \begin{pmatrix} -B\\-E\end{pmatrix}=0. 
\end{equation*}
we expand to get,
\begin{equation*}
  \begin{split}
    \paren{ \dt +B^{(0)}\dz}\bold{u}^{(0)}+
    \eps&\left[\paren{ \dt+B^{(0)}\dz}\bold{u}^{(1)} +\paren{ \dT+B^{(0)}\dZ}\bold{u}^{(0)}\right.\\
    &\left.\quad
      +A^{(1)}(z,\mathbf{u})\dt\bold{u}^{(0)}\right]=\bigo(\eps^2)
  \end{split}
\end{equation*}
with matrices
\begin{equation}
  B^{(0)}  = \begin{pmatrix} 0 & -1 \\ -1 & 0 \end{pmatrix}, \quad
  A^{(1)} = \begin{pmatrix} 2  N(z) + 3 \chi E^2 & 0 \\ 0 &
    0 \end{pmatrix}
\end{equation}

At $\mathcal{O}(\epsilon^0)$,
\begin{equation}
  \paren{\dt +B^{(0)}\dz}\bold{u}^{(0)} =0.
\end{equation}
Solving this as the generalized Eigenvalue problem,
\begin{equation*}
  \paren{B^{(0)} -\lambda I} \bold{r}=0
\end{equation*}
the solutions are:
\begin{equation}
  \lambda_{\pm} = \pm 1,\quad \bold{r}_{\pm} = \begin{pmatrix} 1 \\ \mp 1\end{pmatrix}.
\end{equation}
The corresponding left eigenvectors are
\begin{equation}
  \mathbf{l}_{\pm} = \tfrac{1}{2}\begin{pmatrix} 1 & \mp 1 \end{pmatrix}
\end{equation}
With this normalization, $\mathbf{l}_i \Azero \mathbf{r}_j =
\delta_{i,j}$.  The leading order fields are then
\begin{subequations}\label{eq:fields0}
  \begin{align}
    \mathbf{u}^{(0)} &= E^+(\phi_+,Z,T) \mathbf{r}_+ + E^-(\phi_-,Z,T) \mathbf{r}_-,\\
    E^{(0)} & =  E^+(\phi_+,Z,T) + E^-(\phi_-,Z,T),\\
    \phi_\pm& = z\mp t.
  \end{align}
\end{subequations}
This expression is much more general than \eqref{svea} used in the
derivation of NLCME.

At $\mathcal{O}(\eps)$, the equation is
\begin{equation}
  \label{eq:order1}
  \begin{split}
    \paren{\dt +B^{(0)}\dz} \mathbf{u}^{(1)} = - \paren{\dT
      +B^{(0)}\dZ} \mathbf{u}^{(0)} - A^{(1)}(z,\mathbf{u}^{(0)})\dt
    \mathbf{u}^{(0)}
  \end{split}
\end{equation}
If we assume
\begin{equation}
  \mathbf{u}^{(1)}(z,t) =m^+(z,t) \mathbf{r}_+ +m^-(z,t) \mathbf{r}_-,
\end{equation}
and substitute into \eqref{eq:order1}, then left multiply by
$\mathbf{l}_+$ and then by $\mathbf{l}_-$, we get the two equations
\begin{align}
  -\paren{\dt m^+ + \dz m^+}&= \dT E^+ + \dZ E^++ \mathbf{l}_+
  A^{(1)}(\mathbf{u}^{(0)}) \times\nn \\
  &\quad \paren{-\partial_{\phi_+ }E^+\bold{r}_+
    +\partial_{\phi_-}E^-\bold{r}_- },\label{mplus}\\
  &\nn\\
  -\paren{\dt m^- - \dz m^-}&= \dT E^- - \dZ E^-+ \mathbf{l}_-
  A^{(1)}(\mathbf{u}^{(0)}) \times \nn\\
  &\quad \paren{-\partial_{\phi_+ }E^+\bold{r}_+
    +\partial_{\phi_-}E^-\bold{r}_- }.
  \label{mminus}
\end{align}
The last term is the same in both equations,
\begin{equation}
  \begin{split}
    \bold{l}_\pm &A^{(1)}(\bold{u}^{(0)})\paren{-\partial_{\phi_+
      }E^+\bold{r}_+ +\partial_{\phi_-}E^-\bold{r}_-  }\\
    \quad &= \tfrac{1}{2}\paren{2 N(z) +
      3\chi{E^{(0)}}^2 } \paren{- \partial_{\phi_+ }E^+ +\partial_{\phi_-}E^-  }\\
    \quad   &={N(z)}\paren{-\partial_{\phi_+ } E^+ + \partial_{\phi_-} E^-} \\
    &\quad\quad + \tfrac{3}{2} \chi\paren{E^+ + E^-}^2 \paren
    {-\partial_{\phi_+ } E^+ + \partial_{\phi_-} E^-}.
  \end{split}
\end{equation}

Integration of \eqref{mplus} along the characteristic $\partial_t z_+
=1$ from $t=0$ to $t=L$, yields
\begin{equation}
  \begin{split}
    -&\paren{m^+(z_+(L),L)-m^+(z_+(0),0) } =\\
    &\int_{0}^L \dT E^+(Z,T, z_+(0)) + \dZ E^+(Z,T,z_+(0))ds\\
    &- \int_{0}^L {N(z_+(s))} \partial_{\phi_+ } E^+(Z,T,z_+(0))ds\\
    & + \int_{0}^L {N(z_+(s))}\partial_{\phi_- } E^-(Z,T,z_+(s)+s)ds\\
    &-\int_0^L \left[\tfrac{3}{2} \chi \paren{
        E^+(Z,T,z_+(0))-E^-(Z,T,z_+(s)+s)}^2\right. \\
    &\qquad \left. \times\partial_{\phi_+ } E^+(Z,T,z_+(0))\right]ds\\
    &+\int_0^L\left[ \tfrac{3}{2} \chi \paren{
        E^+(Z,T,z_+(0))-E^-(Z,T,z_+(s)+s)}^2 \right.\\
    &\qquad \left. \times\partial_{\phi_-} E^-(Z,T,z_+(s)+s)\right]ds.
  \end{split}
\end{equation}
Similarly, integration of \eqref{mminus} along the characteristic
$\partial_t z_- = -1$, yields
\begin{equation}
  \begin{split}
    -&\paren{m^-(z_-(L),L)-m^-(z_-(0),0) } = \\
    &\int_{0}^L \dT E^-(Z,T, z_+(0)) -\dZ E^-(Z,T,z_+(0))ds\\
    &  -\int_{0}^L {N(z_-(s))} \partial_{\phi_+ } E^+(Z,T,z_-(s)-s)ds\\
    &   + \int_{0}^L  {N(z_-(s))}\partial_{\phi_- } E^-(Z,T,z_-(0))ds\\
    &     -\int_0^L \left[\tfrac{3}{2} \chi \paren{ E^+(Z,T,z_-(s)-s)-E^-(Z,T,z_-(0))}^2\right.\\
    &\qquad \left.\times\partial_{\phi_+ } E^+(Z,T,z_-(s)- s)\right]ds\\
    & +\int_0^L \left[\tfrac{3}{2} \chi  \paren{E^+(Z,T,z_-(s)- s)-E^-(Z,T,z_-(0))}^2 \right.\\
    &\qquad \left.\times\partial_{\phi_-} E^-(Z,T,z_-(0))\right]ds.
  \end{split}
\end{equation}

Necessary conditions for $m_\pm$ to grow sublinearly in $t$ as $t\to
\infty$ are the solvability conditions:
\begin{subequations}
  \begin{gather}
    \begin{split}
      &\dT E^+(Z,T, z_+(0)) + \dZ E^+(Z,T,z_+(0))=\\
      &-\lim_{L\to\infty}\frac{1}{L}\int_{0}^L {N(z_+(s))}\partial_{z_+(0) } E^-(Z,T,z_+(s)+s)ds\\
      &+ \lim_{L\to\infty}\frac{1}{L}\int_0^L \left[\tfrac{3}{2}
        \chi \paren{ E^+(Z,T,z_+(0))
          +E^-(Z,T,z_+(s)+s)}^2\right.\\
      &\qquad \left. \times \partial_{z_+(0) }
        E^+(Z,T,z_+(0))\right]ds,
    \end{split}
    \\ \\
    \begin{split}
      &\dT E^-(Z,T, z_-(0)) -\dZ E^-(Z,T,z_-(0))=\\
      &\lim_{L\to\infty}\frac{1}{L}\int_{0}^t {N(z_-(s))} \partial_{z_-(0) } E^+(Z,T,z_-(s)-s)ds\\
      &- \lim_{L\to\infty}\frac{1}{L}\int_0^L \left[\tfrac{3}{2}
        \chi \paren{
          E^+(Z,T,z_-(s)-s)-E^-(Z,T,z_-(0))}^2\right.\\
      &\qquad\left.\partial_{z_-(0) } E^-(Z,T,z_-(0))\right]ds.
    \end{split}
  \end{gather}
\end{subequations}

Given $(z,t)$, $z_+(0) = z- t = \phi_+$ and $z_-(0) = z+ t =\phi_-$.
Defining
\begin{equation}
  \label{eq:mean_infty}
  \mean{f} = \lim_{L\to \infty} \frac{1}{L}\int_{0}^L f(s)ds
\end{equation}
the equations may be compactly expressed as:
\begin{subequations}
  \label{e:E_full}
  \begin{gather}
    \label{eq:Ep_full}
    \begin{split}
      &\dT E^+ + \dZ E^+=-\mean{N( \phi_+ +  s) \partial_\phi E^-(\phi_+ + 2  s}_s\\
      &\quad +\tfrac{3}{2} \chi\paren{\paren{E^+}^2 + 2 E^+ \mean{E^-}
        + \mean{\paren{E^-}^2}}\partial_\phi E^+,
    \end{split}
    \\
    \label{eq:Em_full}
    \begin{split}
      &\dT E^- - \dZ E^-=\mean{N(\phi_- -  s)\partial_\phi E^+(\phi_- - 2 s)}_s\\
      &\quad - \tfrac{3}{2}\chi\paren{\paren{E^-}^2 + 2 E^- \mean{E^+}
        + \mean{\paren{E^+}^2}} \partial_\phi E^-.
    \end{split}
  \end{gather}
\end{subequations}
It is important to recognize that the arguments of the fields in
\eqref{eq:Ep_full} are $\phi_+ = z - t$, $Z$, and $T$, while in
\eqref{eq:Em_full}, they are $\phi_- = z + t$, $Z$, and $T$.  As in
our derivation of NLCME in Section \ref{sec:nlcme}, $\Gamma \equiv
\tfrac{3}{2}\chi$.  With this notation, \eqref{e:E_full} can be
rewritten, after an integration by parts, in conservation law form,
\begin{subequations}\label{e:E_con}
  \begin{gather}
    \label{eq:Ep_con}
    \begin{split}
      &\dT \Ep + \dZ \Ep= \dphi \mean{n_1( \phi_+ + s) E^-(\phi_+ + 2
        s}_s\\
      &\quad +\Gamma \dphi\bracket{\frac{1}{3}\paren{\Ep}^3
        + \paren{\Ep}^2 \mean{\Em} + \Ep\mean{\paren{\Em}^2}},
    \end{split}
    \\
    \label{eq:Em_con}
    \begin{split}
      &\dT \Em -\dZ \Em=-\dphi\mean{n_1( \phi_- - s)E^+(
        \phi_- - 2   s) }_s\\
      &\quad -\Gamma \dphi\bracket{\mean{\paren{\Ep}^2
        }\Em+\mean{\Ep}\paren{\Em}^2+\frac{1}{3}\paren{\Em}^3}.
    \end{split}
  \end{gather}
\end{subequations}
Equations \eqref{e:E_con} corresponds to the integro-differential
equations of the introduction, if we omit the $\mean{E^\pm}$ terms.
Since $\mean{E^\pm}$ is time-invariant (see section \ref{s:hamiltonian_structure})
by choosing initial conditions for which $\mean{E^\pm}(T=0)=0$, 
these terms can be dropped from \eqref{e:E_con}.
  Finally, note that \eqref{e:E_full} are applicable to a {\it
  general heterogeneous dielectric material} with the appropriate
scalings.

\subsection{Periodic Media and xNLCME}
  
We now specialize to the periodic case.  Assume now that $N(z+2\pi)=N(z)$. Then
\eqref{e:E_con} is invariant under the discrete translation: $\phi\mapsto\phi+2\pi$, {\it i.e.}
\begin{subequations}
  \begin{align}
    \Ep(\phi, Z,T) &\mapsto \Ep(\phi + 2\pi, Z,T)\\
    \Em(\phi, Z,T) &\mapsto \Em(\phi + 2\pi, Z,T)\ .
  \end{align}
\end{subequations}

Thus, under the assumption of existence and uniqueness of solutions to
\eqref{e:E_con}, if the initial data are $2\pi$ in the $\phi$
argument, then the solutions remain $2 \pi $ periodic in $\phi$.  In
the periodic setting, the averaging operator, \eqref{eq:mean_infty},
simplifies to
\[
\mean{f} = \frac{1}{2\pi}\int_{0}^{2\pi} f(s)ds.
\]

We now expand $N(z)$ and $E^\pm$ in Fourier series,
\begin{align}
  N(z) &= \sum_{p \in \mathbb{Z}} N_p e^{\mathrm{i} p
    z},\\\label{e:efield_fourier} E^\pm(\phi ,Z,T) &= \sum_{p}
  E^{\pm}_p(Z,T) e^{\pm\mathrm{i} p \phi},
\end{align}
where $\bar{N}_p= N_{-p}$ and $\bar{E}^\pm_p= E^\pm_{-p}$ since $N$
and $E^\pm$ are real valued. In this case, the system of Fourier
coefficients $\{E^\pm_p(Z,T): p\in \mathbb{Z}\}$ satisfy the
infinite system of {\it extended nonlinear coupled mode equations}
(xNLCME):
\begin{subequations}\label{e:E_mode}
  \begin{gather}
    \label{eq:Ep_mode}
    \begin{split}
      \dT \Ep_p + \dZ \Ep_p &= \mathrm{i}p N_{2p}{\Em_p} +
      \mathrm{i}p\frac{\Gamma}{3}\left[\sum_{q,r} \Ep_q
        \Ep_r \Ep_{p-q-r}\right.\\
      &\quad \left. + 3 \Em_0 \sum_q \Ep_q \Ep_{p-q} +3\paren{\sum_q
          \abs{\Em_q}^2} \Ep_p\right],
    \end{split}
    \\
    \label{eq:Em_mode}
    \begin{split}
      \dT \Em_p - \dZ \Em_p &= \mathrm{i}p
      \bar{N}_{2p}{\Ep_p}+\mathrm{i}p\frac{\Gamma}{3} \left[\sum_{q,r} \Em_q
        \Em_r \Em_{p-q-r}\right.\\
      &\quad \left. + 3 \Ep_0 \sum_q \Em_q \Em_{p-q}+3\paren{\sum_q
          \abs{\Ep_q}^2} \Em_p\right].
    \end{split}
  \end{gather}
\end{subequations}

\subsection{Conservation Laws and Hamiltonian Structure}
\label{s:hamiltonian_structure}

Equation \eqref{e:E_con}, and alternatively \eqref{e:E_mode}, have  two
conservation laws:
\begin{prop}
  Assume that $E^\pm$ is a sufficiently smooth and sufficiently $Z-$ decaying solution of \eqref{e:E_con}
   and that $\{E_p(Z,T)\}_{p\in\mathbb{Z}}$ is the corresponding solution of xNLCME. Then, 
  \begin{subequations}
    \begin{gather}
     \frac{d}{dT}\ \int \mean{E^+(\cdot,T)}\ dZ = \ \frac{d}{dT}\ \int E_0^+(\cdot,T)\ dZ\ =\ 0\\
      \frac{d}{dT}\ \int \mean{E^+(\cdot,T)}\ dZ =   \frac{d}{dT}\ \int E_0^-(\cdot,T)\ dZ\ =\ 0\\
       \frac{d}{dT}\ \int \mean{(\Ep)^2(\cdot,T)}\ +\  \mean{(\Em)^2(\cdot,T)}\ dZ\\
        =\  
       \frac{d}{dT}\ \sum_p \int \abs{\Ep_p(\cdot,T)}^2 +
      \abs{\Em_p(\cdot,T)}^2\ dZ\ =\ 0\ . 
    \end{gather}
  \end{subequations}
\end{prop}
\begin{proof}
  Setting $p=0$ in \eqref{e:E_mode},
  \begin{align*}
    \dT \Ep_0 + \dZ \Ep_0&=0,\\
    \dT \Em_0 - \dZ \Em_0&=0.
  \end{align*}
  Integrating in $Z$ establishes the first two conservation laws in
  terms of the Fourier modes.  Integrating \eqref{e:efield_fourier} in
  $\phi$ over $[0, 2\pi)$ relates $\mean{E^\pm}$ to $E_0^\pm$.

\bigskip

  Multiplying \eqref{eq:Ep_mode} by $\bar{E}^+_p$, summing over $p$,
  and adding its complex conjugate,
  \begin{equation*}
    \begin{split}
      \sum \partial_T \abs{E_p^+}^2 + \partial_Z \abs{E_p^+}^2 &= \sum_p
      ip N_{2p} E_p^- \bar{E}_p^+ +\ \frac{\Gamma}{3}\ \sum_p\  ip \left[\sum_{q,r}
        \Ep_q
        \Ep_r \Ep_{-p}\Ep_{p-q-r}\right.\\
      &\quad + 3 \Em_0 \sum_q \Ep_q \Ep_{p-q}\Ep_{-p} \\
      &\quad \left.+3\paren{\sum_q
          \abs{\Em_q}^2} \ \abs{\Ep_p}^2\ \right] + \cc
    \end{split}
  \end{equation*}
  The quartic terms will all vanish.  Consider the first quartic term,
  and note that
\begin{equation}
\begin{split}
\sum_{p,q,r} p \Ep_q\Ep_r \Ep_{-p}\Ep_{p-q-r} &= \sum_{k_1 + k_2 + k_3
  + k_4 = 0} k_1 \Ep_{k_1}\Ep_{k_2} \Ep_{k_3}\Ep_{k_4}\\
&=\sum_{k_1 + k_2 + k_3
  + k_4 = 0} k_2 \Ep_{k_1}\Ep_{k_2} \Ep_{k_3}\Ep_{k_4}\\
&=\sum_{k_1 + k_2 + k_3
  + k_4 = 0} k_3 \Ep_{k_1}\Ep_{k_2} \Ep_{k_3}\Ep_{k_4}\\
&=\sum_{k_1 + k_2 + k_3
  + k_4 = 0} k_4 \Ep_{k_1}\Ep_{k_2} \Ep_{k_3}\Ep_{k_4}
\end{split}
\end{equation}
Hence,
\begin{equation}
\begin{split}
&\sum_{p,q,r} p \Ep_q\Ep_r \Ep_{-p}\Ep_{p-q-r}  \\
&\quad= \frac{1}{4}\sum_{k_1 + k_2 + k_3
  + k_4 = 0} ({k_1 + k_2 + k_3
  + k_4 }) \Ep_{k_1}\Ep_{k_2} \Ep_{k_3}\Ep_{k_4}=0
\end{split}
\end{equation}
The second quartic term vanishes using a similar analysis.  The last
quartic term,
\begin{equation}
\sum_p p \paren{\sum_q \abs{\Em_q}^2}\abs{\Ep_p}^2
\end{equation}
will vanish because the $p$ and $-p$ terms will cancel one another.
Similar analysis holds for \eqref{eq:Em_mode}, leaving
  us with the two equations
  \begin{align}
    \sum \partial_T \abs{E_p^+}^2 + \partial_Z \abs{E_p^+}^2 &= \sum
    ip N_{2p}
    E_p^- \bar{E}_p^+ - ip \bar{N}_{2p} \bar{E}_p^- E_p^+,\\
    \sum \partial_T \abs{E_p^-}^2 - \partial_Z \abs{E_p^-}^2 &= \sum
    ip \bar{N}_{2p} \bar{E}_p^- {E}_p^+ - ip {N}_{2p} {E}_p^-
    \bar{E}_p^+.
  \end{align}
    Summing these two, and integrating in $Z$ gives the $L^2$
  conservation law.
\end{proof} {\bf To simplify our analysis we assume $E^\pm_0$ are
  initially zero from here on}.  The equations reduce to
\begin{subequations}\label{e:E_con_zeromean}
  \begin{gather}
    \label{eq:Ep_con1}
    \begin{split}
      \dT \Ep + \dZ \Ep&= \dphi  \mean{N( \phi_+ +   s) E^-(\phi_+ + 2s}_s\\
      &\quad +\Gamma \dphi\bracket{\frac{1}{3}\paren{\Ep}^3 +
        \Ep\mean{\paren{\Em}^2}},
    \end{split}
    \\
    \label{eq:Em_con1}
    \begin{split}
      \dT \Em - \dZ \Em&=-\dphi\mean{N( \phi_- -   s)E^+( \phi_- - 2 s) }_s\\
      &\quad -\Gamma \dphi\bracket{\mean{\paren{\Ep}^2
        }\Em+\frac{1}{3}\paren{\Em}^3}.
    \end{split}
  \end{gather}
\end{subequations}
and
\begin{subequations}\label{e:E_mode_zeromean}
  \begin{gather}
    \label{eq:Ep_mode1}
    \begin{split}
      \dT \Ep_p + \dZ \Ep_p = \mathrm{i}p N_{2p}{\Em_p} +
      \mathrm{i}p\frac{\Gamma}{3}&\left[\sum\Ep_q
        \Ep_r \Ep_{p-q-r} \right.\\
      &\quad\left.+3\paren{\sum \abs{\Em_q}^2} \Ep_p \right],
    \end{split}
    \\
    \label{eq:Em_mode1}
    \begin{split}
      \dT \Em_p - \dZ \Em_p = \mathrm{i}p \bar{N}_{2p}{\Ep_p}
      +\mathrm{i}p\frac{\Gamma}{3} &\left[\sum \Em_q
        \Em_r \Em_{p-q-r} \right.\\
      &\quad \left.+3\paren{\sum \abs{\Ep_q}^2} \Em_p \right].
    \end{split}
  \end{gather}
\end{subequations}
These are equations \eqref{eq:integro_diff_intro} and
\eqref{e:mode_intro} from the introduction.  Truncating
\eqref{e:E_mode_zeromean} to just mode $E^\pm_{\pm 1}$, recovers the
NLCME, subject to the identification of $\mathcal{E}^\pm $ with
$E^\pm_{1}$.

Another time-invariant functional is a consequence of the 
Hamiltonian structure given in the following result, which is straightforward to verify:
\begin{prop}
  The system \eqref{e:E_mode_zeromean} is a Hamiltonian system:
  \begin{equation}
    \dT \Ep_p = - \mathrm{i} p \frac{\delta H}{\delta \overline{E}^+_p},
    \quad \dT \Em_p = - \mathrm{i} p \frac{\delta H}{\delta \overline{E}^-_p},
  \end{equation}
   where with time-invariant Hamiltonian
   \begin{equation}
   H[E^\pm,\overline{E^\pm}] = \int\ \mathcal{H}(\cdot,T)
   \nn\end{equation}
    and Hamiltonian density
  \begin{equation}
    \begin{split}
      \mathcal{H}(Z,T) &= \frac{\mathrm{i}}{2} \sum_{p_1 =1}^\infty
      \frac{1}{p_1}\paren{\Ep_{p_1} \dZ {\bar{E}^+}_{p_1} -
        \Em_{p_1} \dZ{\bar{E}^-}_{p_1}}- \sum_{p_1=1}^\infty N_{2 p_1} \bar{E}^+_{p_1}{\Em_{p_1}}\\
      &\quad- \frac{\Gamma}{3}
      \frac{1}{2}\frac{1}{4}\sum_{p_1+p_2+p_3+p_4=0}
      \Ep_{p_1}\Ep_{p_2}\Ep_{p_3}\Ep_{p_4}+\Em_{p_1}\Em_{p_2}\Em_{p_3}\Em_{p_4}\\
      &\quad - \Gamma
      \frac{1}{2}\frac{1}{2}\paren{\sum_{p_1}\abs{\Ep_{p_1}}^2}\paren{\sum_{p_1}\abs{\Em_{p_1}}^2}+
      \cc\ \ \ .
    \end{split}
  \end{equation}
\end{prop}

\section{Simulations of the Truncated xNLCME}
\label{s:xnlcme_sims}

In this section we simulate truncations of the infinite dimensional
xNLCME system, performed pseudo-spectrally with fourth order
Runge-Kutta time stepping. These simulations suggest that
\begin{itemize}
\item xNLCME has its own localized soliton-like structures which
  better capture the dynamics of the nonlinear periodic Maxwell
  equation for our class of initial conditions than NLCME and
\item xNLCME has singular solutions, $\{E^\pm_p(Z,T)\}$ with a cascade of energy to higher wave numbers, $p$.
The physical electric field 
\begin{align}
E(z,t)\ &\ \approx \epsilon^{1\over2}\left( E^+(z-t,,\epsilon z,\epsilon t) + E^-(z+t,,\epsilon z,\epsilon t)\ \right)\nn\\
&\ =\ \epsilon^{1\over2}\  \sum_{p\in\mathbb{Z}\setminus{0}}\ \left(\ E^+_p(Z,T)e^{ip(Z- T)/\eps} +
  E^-_p(Z,T)e^{-ip(Z+  T)/\eps}\ \right)\nn\\
  &\ \ \ \ \ \ \ \ \  + \cc\ \nn
\end{align}
develops a carrier-shock train structure.
\end{itemize}

As we saw in Section \ref{s:maxwell_solitons}, particularly Figure
\ref{f:standing_soliton}, though the NLCME soliton data appeared
robust, there was some escape of energy.  This can be accounted for in
the xNLCME through the inclusion of additional modes.

Starting with the same initial conditions, we simulate the NLCME
soliton of $E_{\pm 1 }^\pm$ with soliton parameters $v=0$ and $\delta
= \tfrac{\pi}{2}$, and material parameters
\[
\Gamma = 1, \quad N_{\pm 2}=\tfrac{2}{\pi}, \quad N_{j \neq \pm 2} = 0
\]
in \eqref{e:E_mode_zeromean} resolving only a finite number of
harmonics.  The primitive electric field is reconstructed from these
simulations as
\begin{equation}
  E = \sum_{p=-p_{\max}}^{p_{\max}} E^+_p(Z,T)e^{ip(Z- T)/\eps} +
  E^-_p(Z,T)e^{-ip(Z+  T)/\eps} + \cc
\end{equation}
$E$ is plotted in Figures \ref{f:standing_soliton3}, and
\ref{f:standing_soliton15}, which resolve odd modes up to 3 and 15,
respectively.  Comparing with Figure \ref{f:standing_soliton}, we
infer that the two smaller pulses symmetrically expelled from the main
wave were transferred into $E^\pm_{\pm 3}$, since these clearly appear
in Figure \ref{f:standing_soliton3}.  This addresses the {\it
  macroscopic} discrepancy between NLCME and Maxwell.

  \begin{figure}
    \centering
    \subfigure{\includegraphics[width=2.3in]{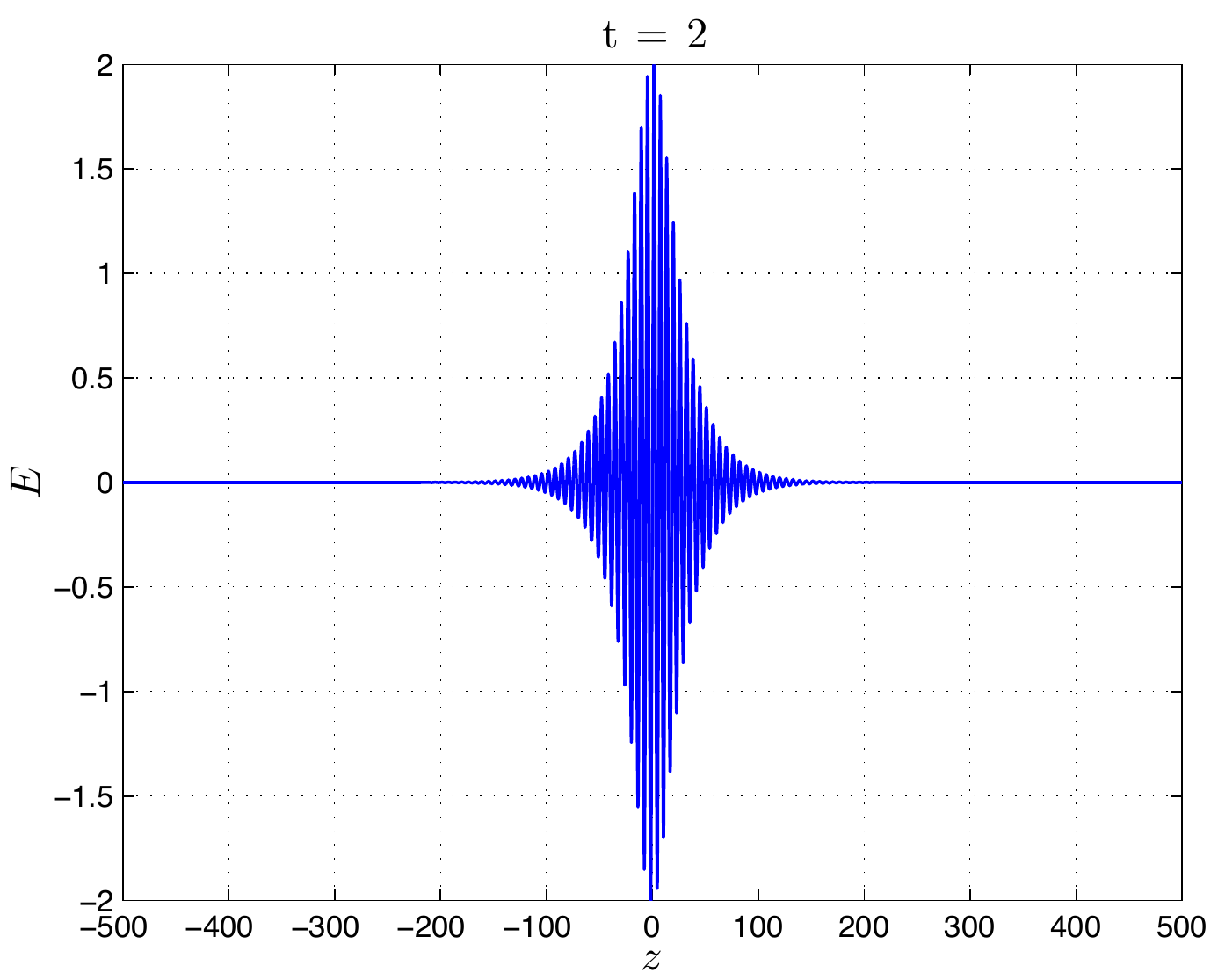}}
    \subfigure{\includegraphics[width=2.3in]{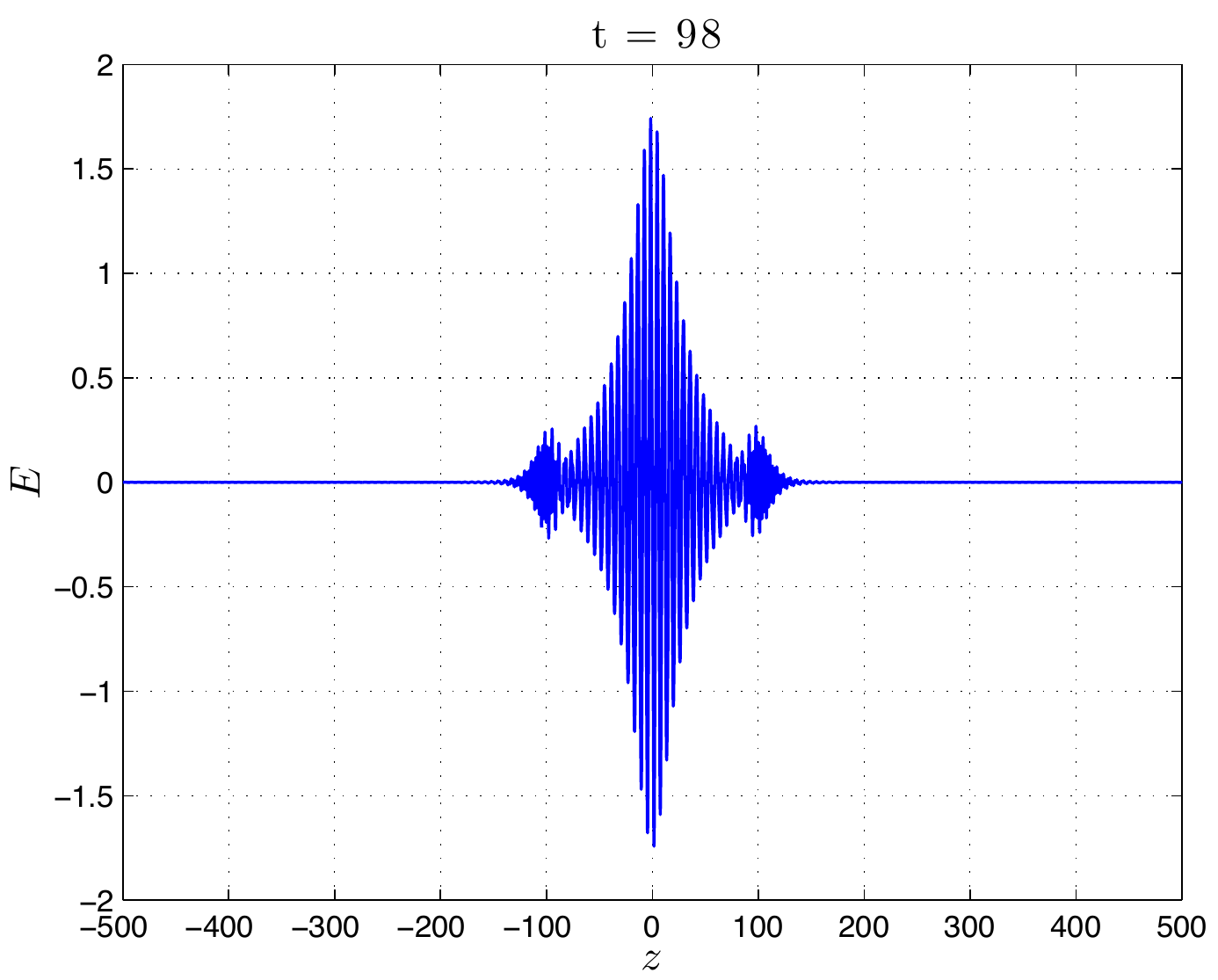}}

    \subfigure{\includegraphics[width=2.3in]{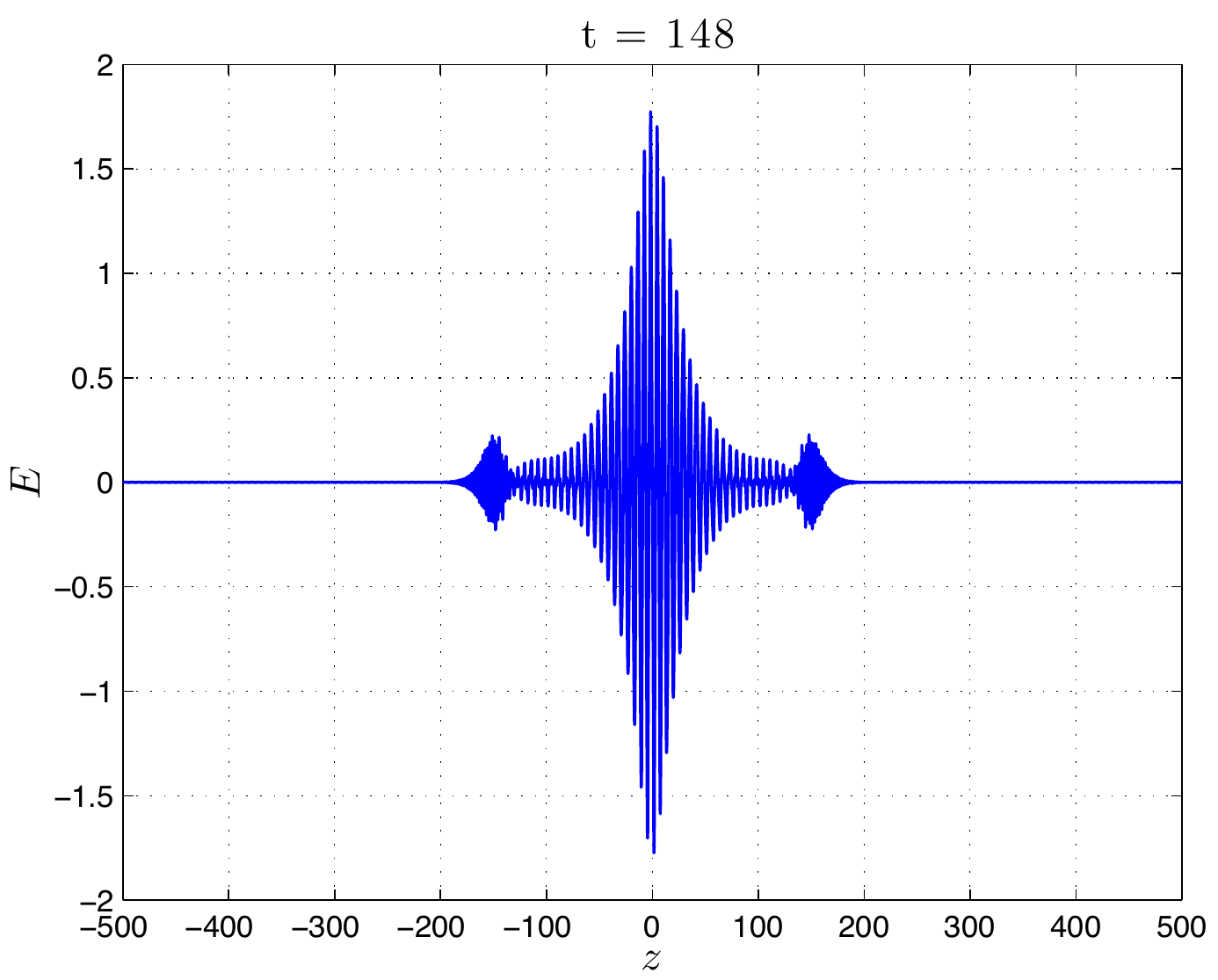}}
    \subfigure{\includegraphics[width=2.3in]{figs/plots_print_standing_wave_M16_N4096_Zmax32/fig100}}

    \caption{Evolution of an NLCME soliton in the xNLCME, resolving
      odd modes $\abs{p}\leq 4$. Computed with $4096$ grid points in
      the $Z$ coordinate. Compare with Figure
      \ref{f:standing_soliton}.}
    \label{f:standing_soliton3}
  \end{figure}

  \begin{figure}
    \centering
    \subfigure{\includegraphics[width=2.3in]{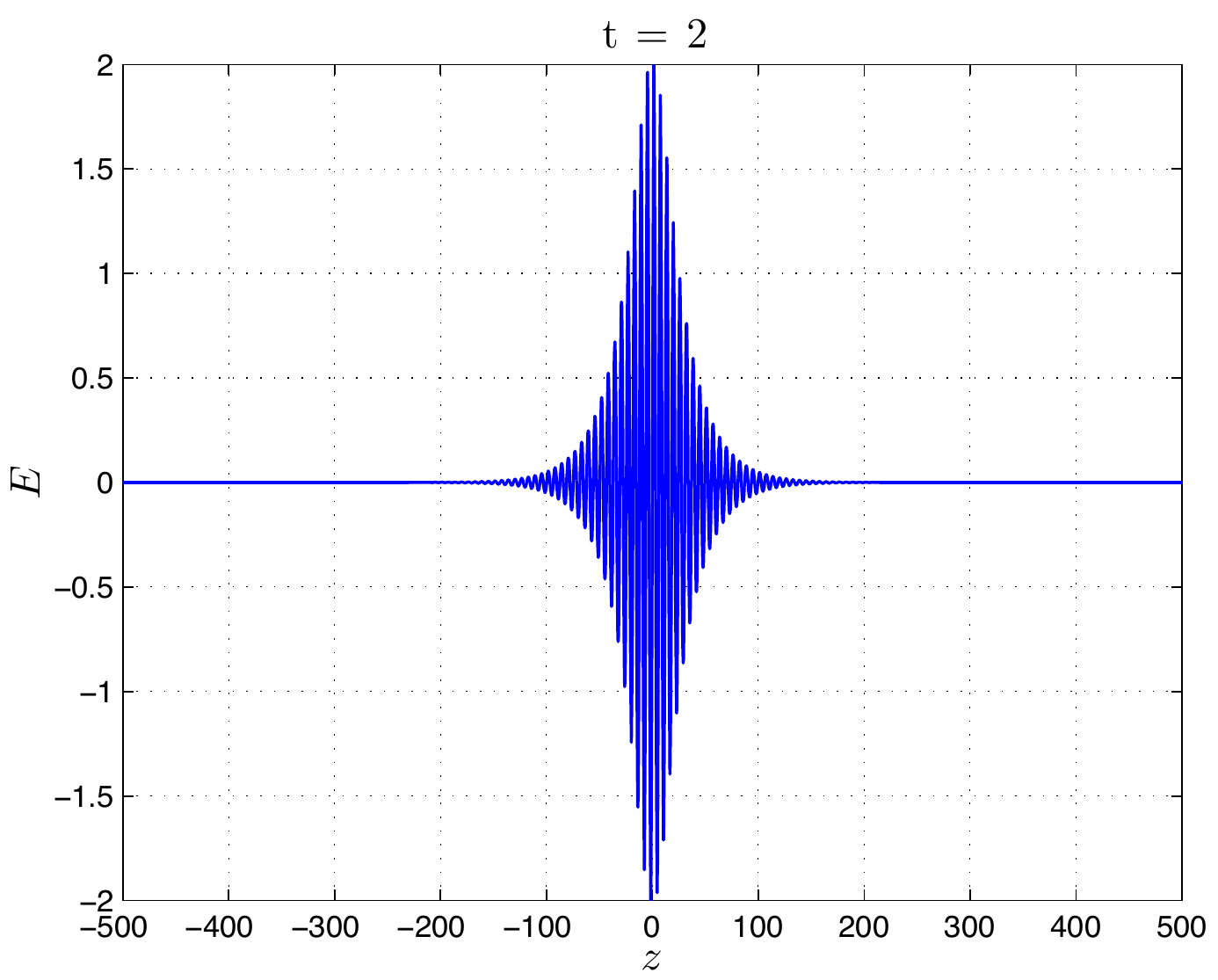}}
    \subfigure{\includegraphics[width=2.3in]{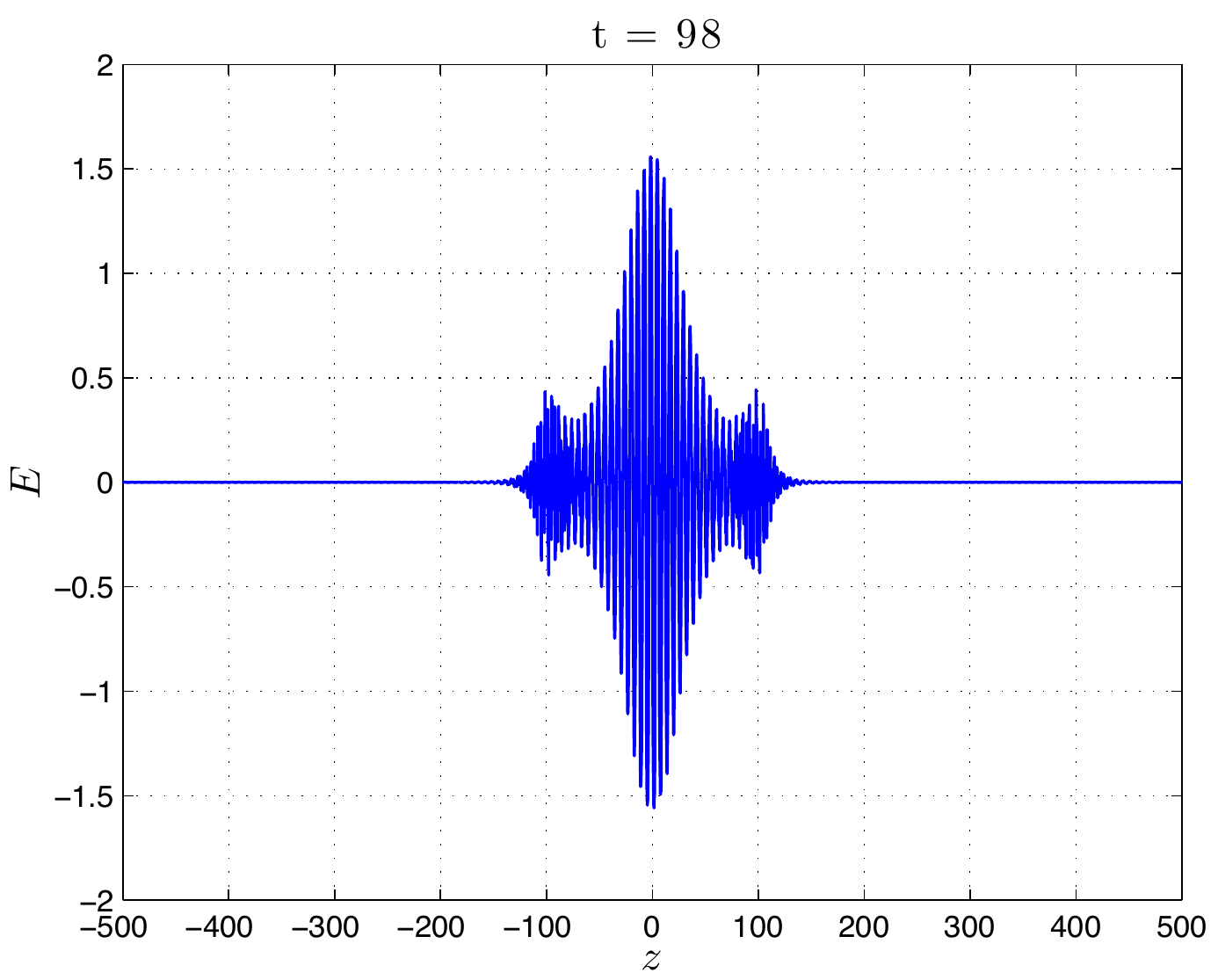}}

    \subfigure{\includegraphics[width=2.3in]{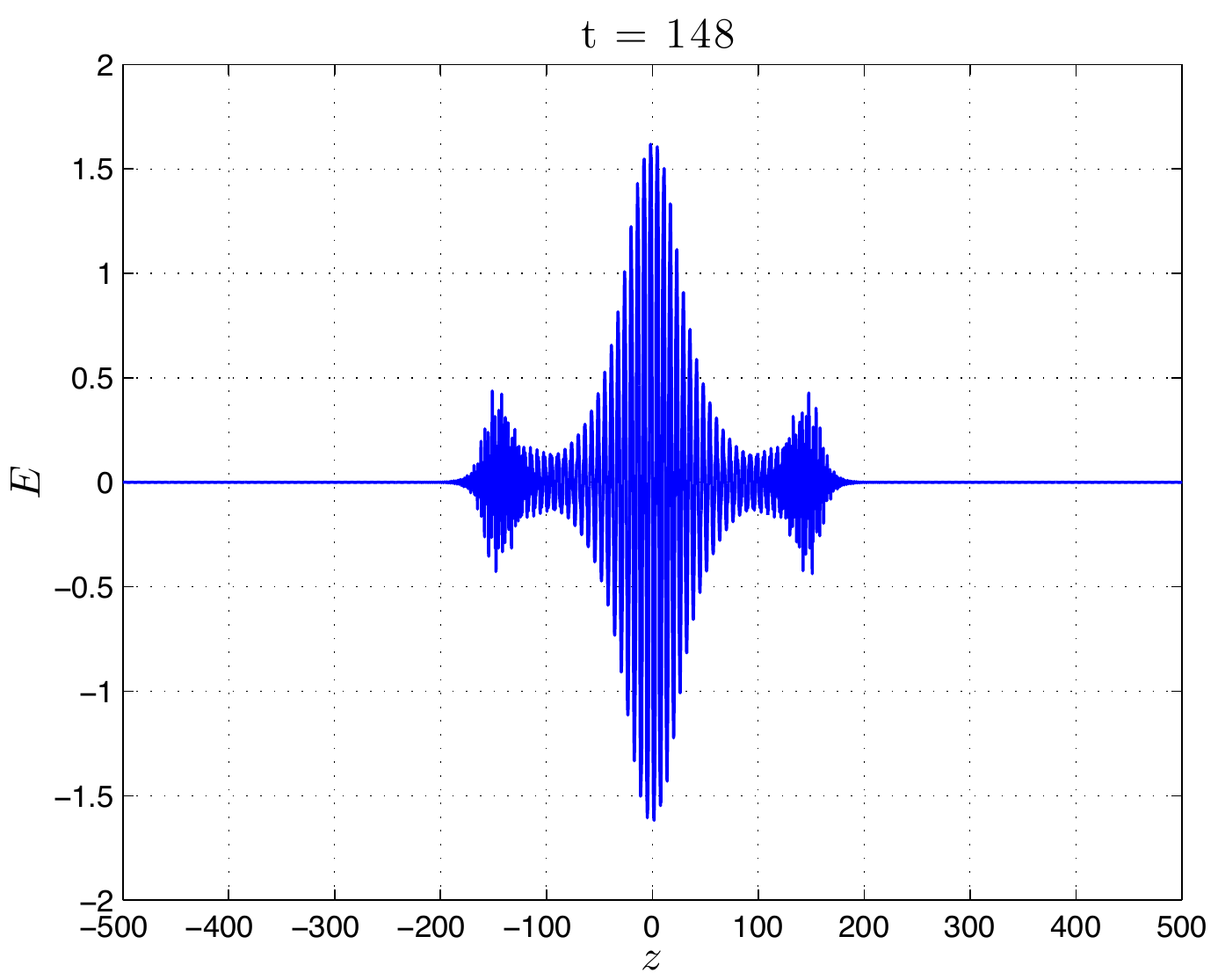}}
    \subfigure{\includegraphics[width=2.3in]{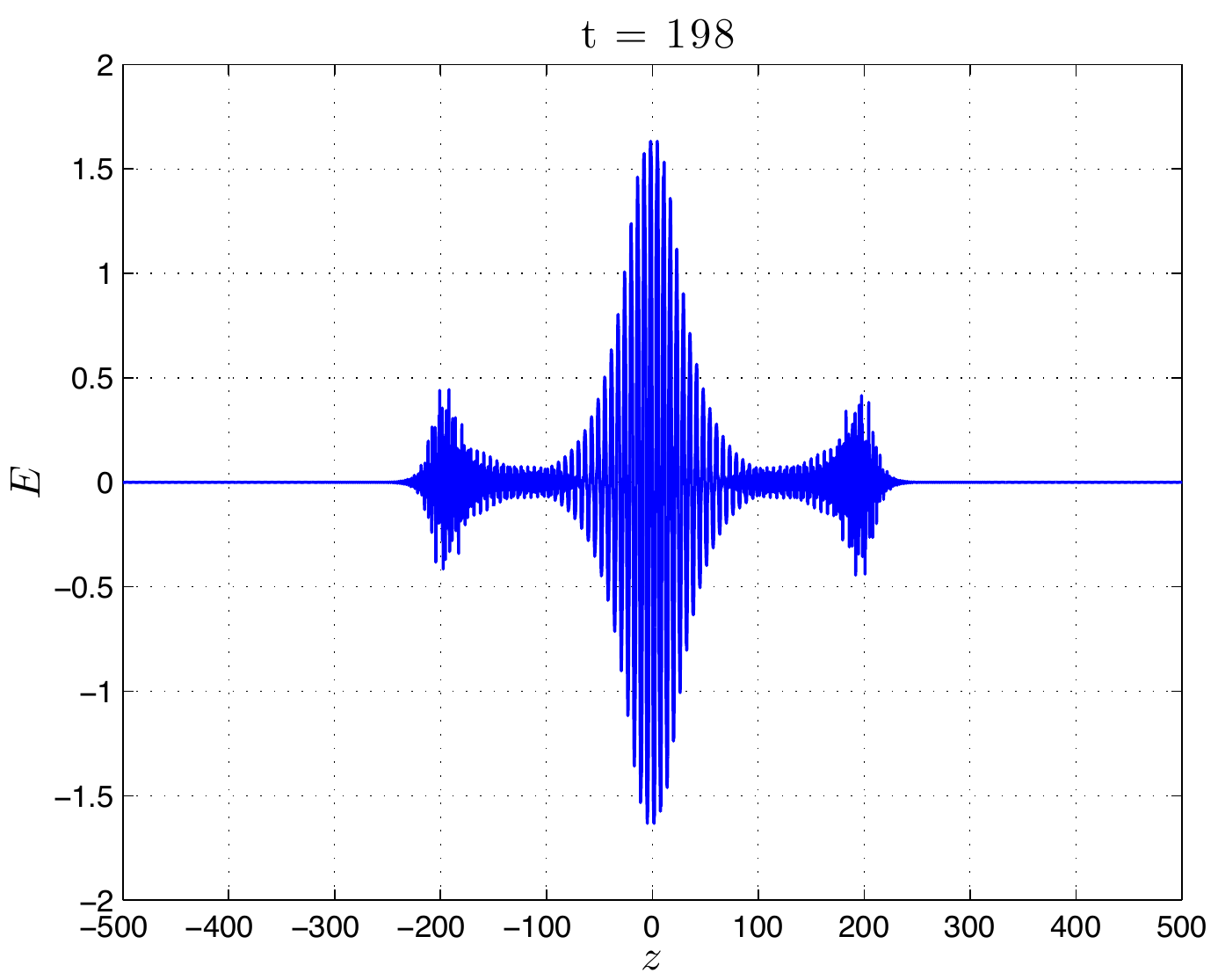}}

    \caption{Evolution of an NLCME solition in the xNLCME, resolving
      odd modes $\abs{p}\leq 16$.  Computed with $16384$ grid points
      in the $Z$ coordinate.  Compare with Figure
      \ref{f:standing_soliton}.}
    \label{f:standing_soliton15}
  \end{figure}

  Including the additional modes also suggests shock formation by
  re-examining Figure \ref{f:shock_comparison}.  The sharper, shock
  like features, can only be resolved by the inclusion of the the
  higher harmonics.  The contrast between different truncations is
  shown in Figure \ref{f:shock_harmonics}.  Indeed, we see the Gibbs
  phenomenon that would be expected from taking a truncated Fourier
  representation of a discontinuous function.

\begin{figure}
  \centering \subfigure[Resolves odd harmonics $\abs{p}\leq
  2$]{\includegraphics[width=2.35in]{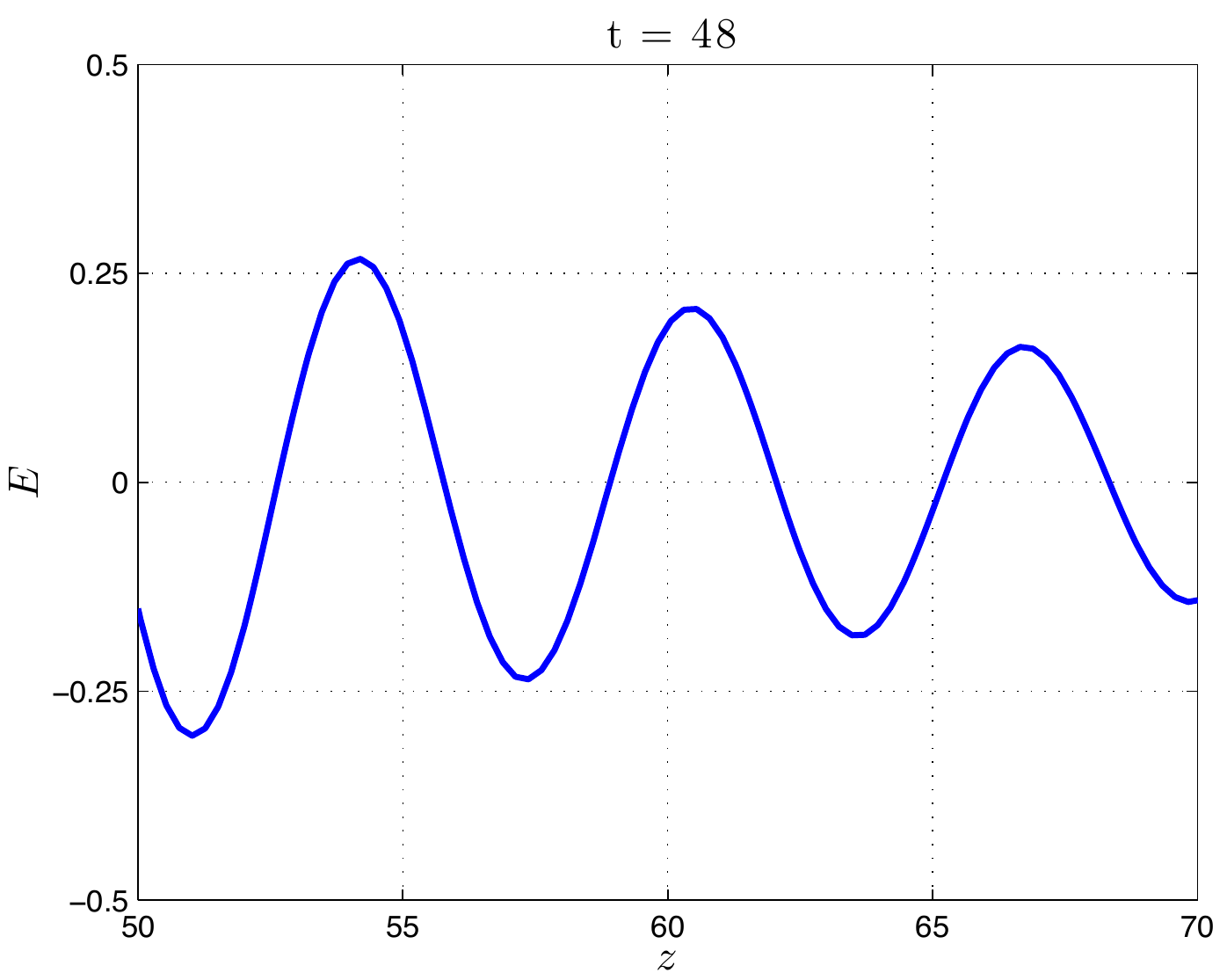}}
  \subfigure[Resolves odd harmonics $\abs{p}\leq
  4$]{\includegraphics[width=2.35in]{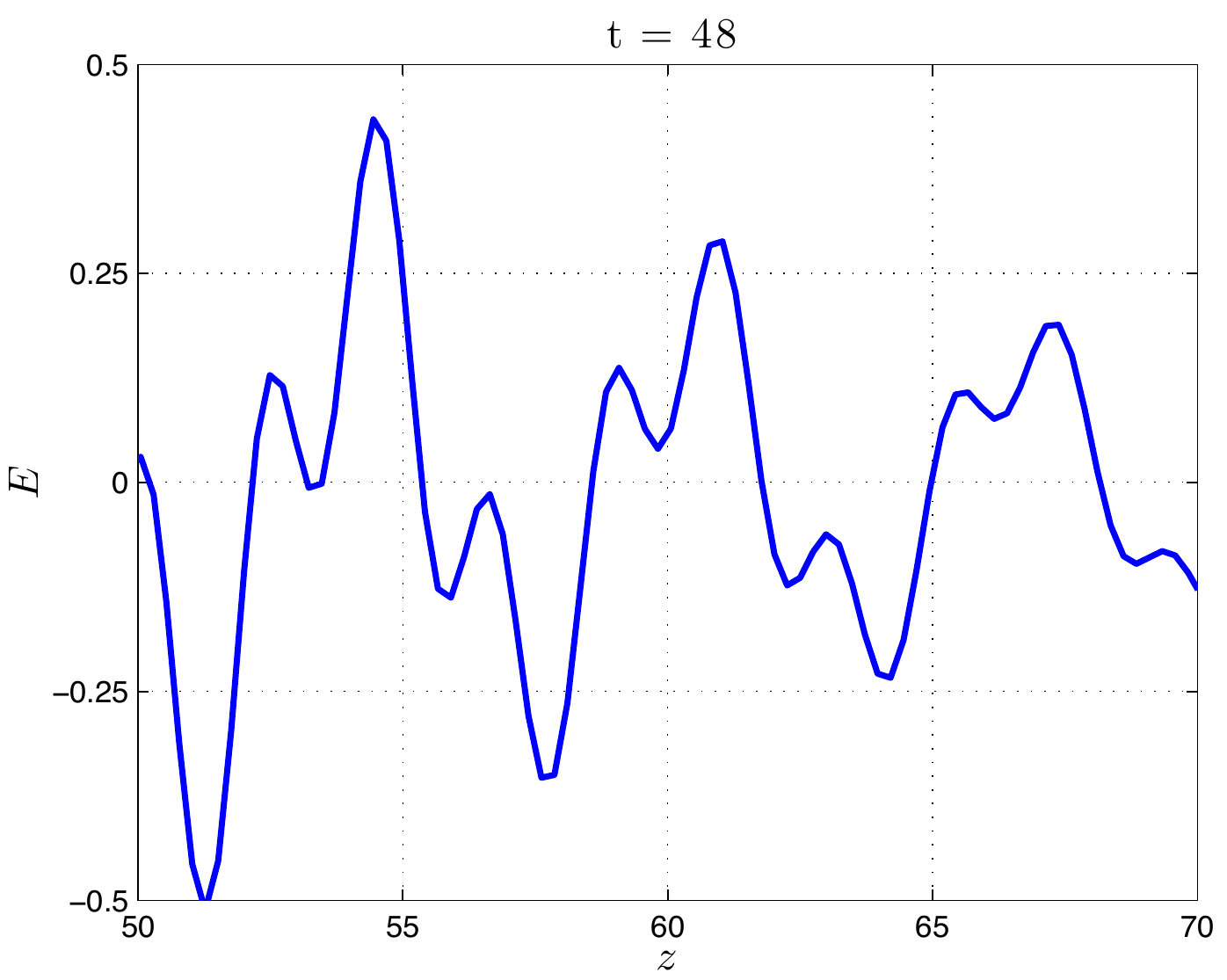}}
  \subfigure[Resolves odd harmonics $\abs{p}\leq
  8$]{\includegraphics[width=2.35in]{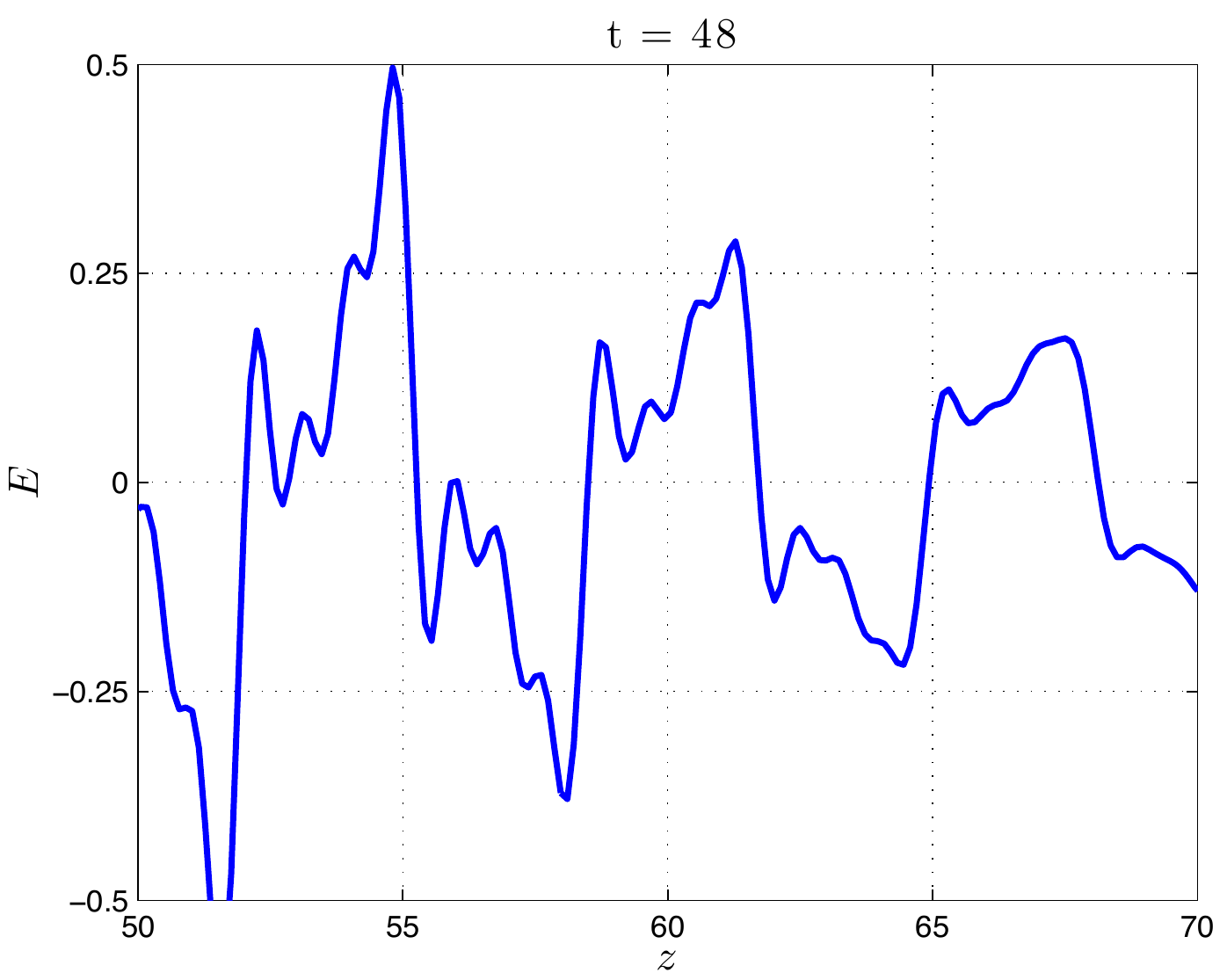}}
  \subfigure[Resolves odd harmonics $\abs{p}\leq
  16$]{\includegraphics[width=2.35in]{figs/plots_print_standing_wave_M64_N16384_Zmax32/shock_fig25}}
  \caption{Comparison of the features that develop on the scale of the
    medium in different truncations of the equations.  Including
    additional harmonics better captures the shocks seen in Figure
    \ref{f:shock_comparison}.  }
  \label{f:shock_harmonics}
\end{figure}

Despite this, NLCME still gets certain leading order effects correct,
such as the main structure in the Maxwell simulations.  The robustness
of NLCME can also be seen by exploring how energy is partitioned
amongst the harmonics.  Let
\begin{equation}
  \label{e:energy_density}
  e_p \equiv \int \paren{\abs{E^+_p}^2 + \abs{E^+_{-p}}^2+\abs{E^-_p}^2
    + \abs{E^-_{-p}}^2}dZ, \quad p = 1,3,\ldots p_{\max}.
\end{equation}
This is the energy  associated with mode $p$.  Their sum is
conserved.  Plotting this for the above simulations in Figure
\ref{f:eng_dist}, we see that most of the energy remains in mode one,
some migrates into mode three, and less in the subsequent modes.


\begin{figure}
  \centering \subfigure[Resolves odd harmonics $\abs{p}\leq
  4$]{\includegraphics[width=2.35in]{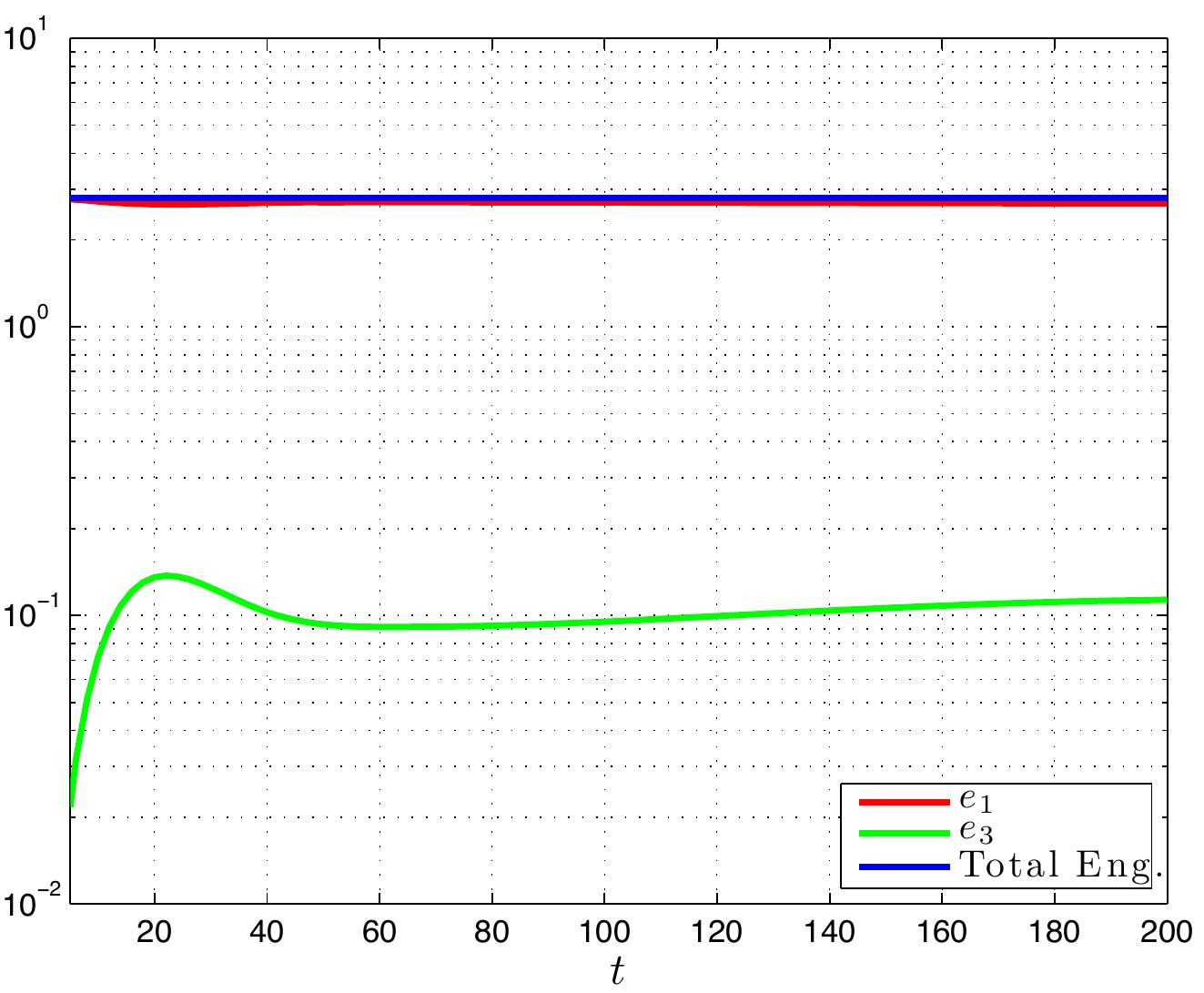}}
  \subfigure[Resolves odd harmonics $\abs{p}\leq
  8$]{\includegraphics[width=2.35in]{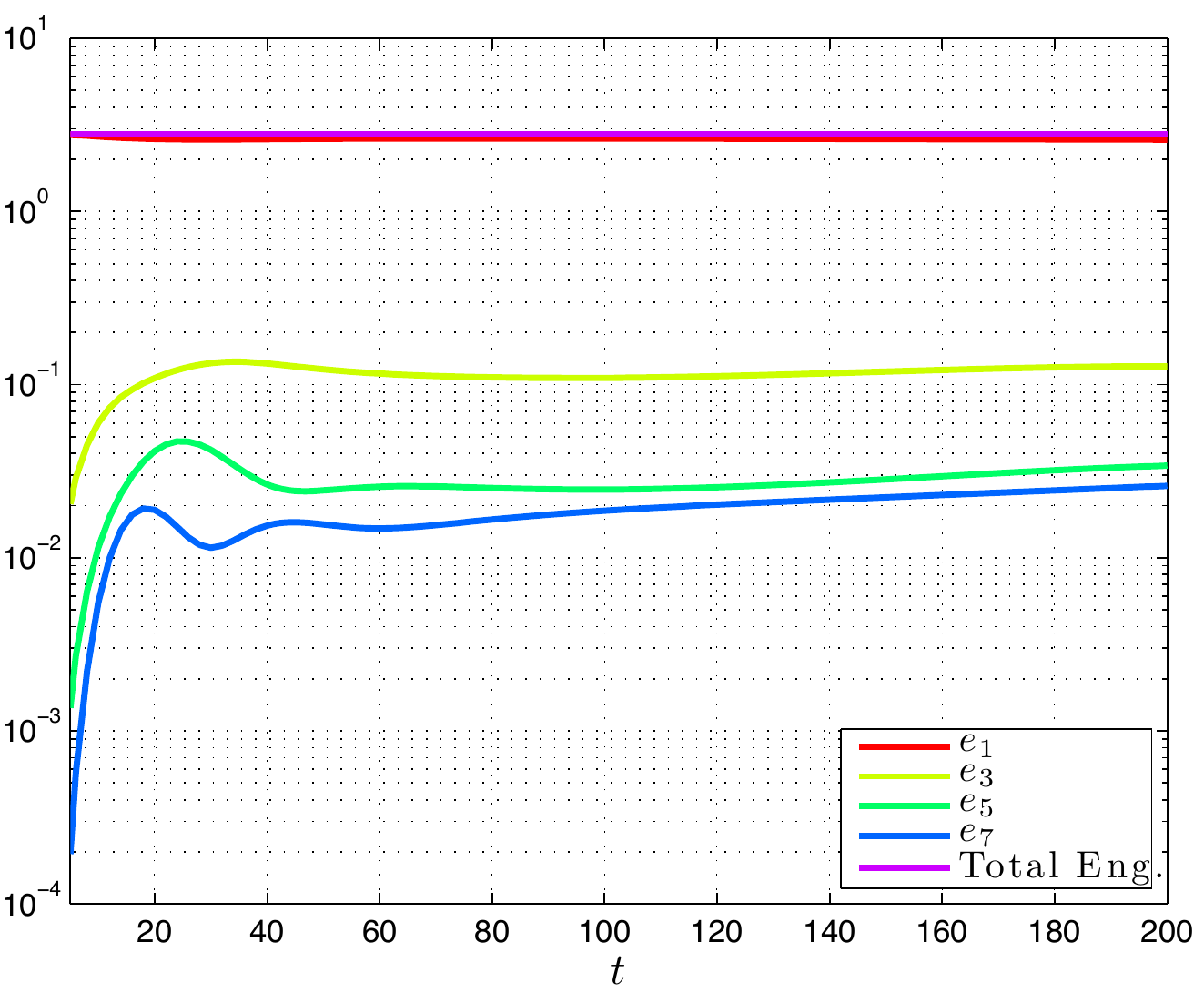}}
  \subfigure[Resolves odd harmonics $\abs{p}\leq
  16$]{\includegraphics[width=2.35in]{figs/plots_print_standing_wave_M64_N16384_Zmax32/eng_dist_8}}
  \caption{Energy distribution, \eqref{e:energy_density}, for
    truncated xNLCME simulations with different numbers of harmonics.
    In all cases, the energy initially residing in mode one tends to
    stay there.}
  \label{f:eng_dist}
\end{figure}

\section{Summary and discussion}
\label{sec:discussion}

We first numerically simulated the one-dimensional nonlinear Maxwell
equations in the regime of weak nonlinearity, low contrast periodic
structure (weak dispersion) with wave-packet data satisfying a {\it
  Bragg resonance condition}, {\it i.e.}  carrier wavelength equal to
twice the medium periodicity.  We observe strong evidence of the
emergence of a coherent structure evolving as slowly varying envelope
with a {\it carrier-shock train}. This violates the nearly-monochromatic
assumption underlying the classical nonlinear coupled mode equations.
We propose our nonlocal integro-differential equations governing
coupled forward and backward waves, derived via a nonlinear
geometrical optics expansion, as the physically correct,
mathematically consistent description of waves governed by nonlinear
Maxwell in a periodic structure with negligible chromatic (nonlocal in
time) dispersion. These equations are equivalent to an infinite
dimensional system of couple first order PDEs, the {\it extended
  coupled mode system} (xNLCME).  The electric field, $E$, obtained
from numerical solution of successively higher truncations of xNLCME
converges toward the envelope carrier-shock trains observed in direct
simulations of the nonlinear Maxwell equations.


Finally we mention that our methods could be applied to study the long time evolution 
of wave-packet type initial conditions  for the problem of quadratically
nonlinear elastic media,
consider in \cite{leveque2002fvm,leveque2003swl,Ketcheson:2009fk}
We obtain nonlocal equations of resonant 
  nonlinear geometrical optics (or equivalently an infinite family of nonlinear coupled mode equations),
  governing   interacting forward and backward propagating waves \cite{Simpson:2010uq}. 
   A difference between the quadratic and cubic case is that the smallest
truncated system that retains nonlinear interactions contains four
modes, $p=\pm 1, \pm 2$.  Nonlinear effects occur through second harmonic generation,
 a process well-known in nonlinear optics.

\noindent {\bf Open problems and conjectures:}  As our simulations
show, there is agreement between finite mode truncations of the
integro-differential equations and the primitive Maxwell system.
Assessing, and proving the time of validity of this approximation is
one upon problem.

Following up on assessing the time of validity, there is also the
question of the time of existence and the well-posedness of the
equations. We expect that solutions of xNLCME for initial data having
a finite number of nonzero mode amplitudes, {\it e.g.} NLCME gap
soliton data, will give rise to solutions of xNLCME that develop
singularities in finite time.  The nature of this blowup is expected
to occur via a cascade to high mode amplitudes (higher index, $p$),
corresponding to modes necessary to resolve the carrier shock
structure in the small scale.  As we mentioned in the discussion,
there is clearly singularity formation when the heterogeneity is
turned off ($N=0$), and either $E^+$ or $E^-$ is initially zero.  It
is an open problem as to whether this particular mechanism for
singularity formation will persist when coupling is restored.

As pointed out in the introduction, the success in modeling
experiments with NLCME suggests that, although there is such a (weakly
turbulent) cascade, it is only a small part of the optical power that
is transferred to high wavenumbers and that this energy contributes
mainly to resolving the small-scale shocks.  To explore this, one needs
to simulate the xNLCME equations with many more harmonics.  Plotting
the Fourier transform (in the $Z$ coordinate) of the simulations in
Section \ref{s:xnlcme_sims} in figure \ref{f:fourier_dist}, we see
that the spectral support grows as we increase the number of resolved
envelopes (the $E_p^\pm$'s).  A related question is whether or not the
primitive Maxwell system, the xNLCME system, or one of its
truncations possess genuine coherent structures.  In
\cite{Tasgal:2005p6335}, the authors found such solutions for a first
and third harmonic system.  This shall be further 
explored in the forthcoming publication, \cite{Simpson:fk}.

\begin{figure}
  \centering \subfigure[Resolves odd harmonics $\abs{p}\leq
  2$]{\includegraphics[width=2.35in]{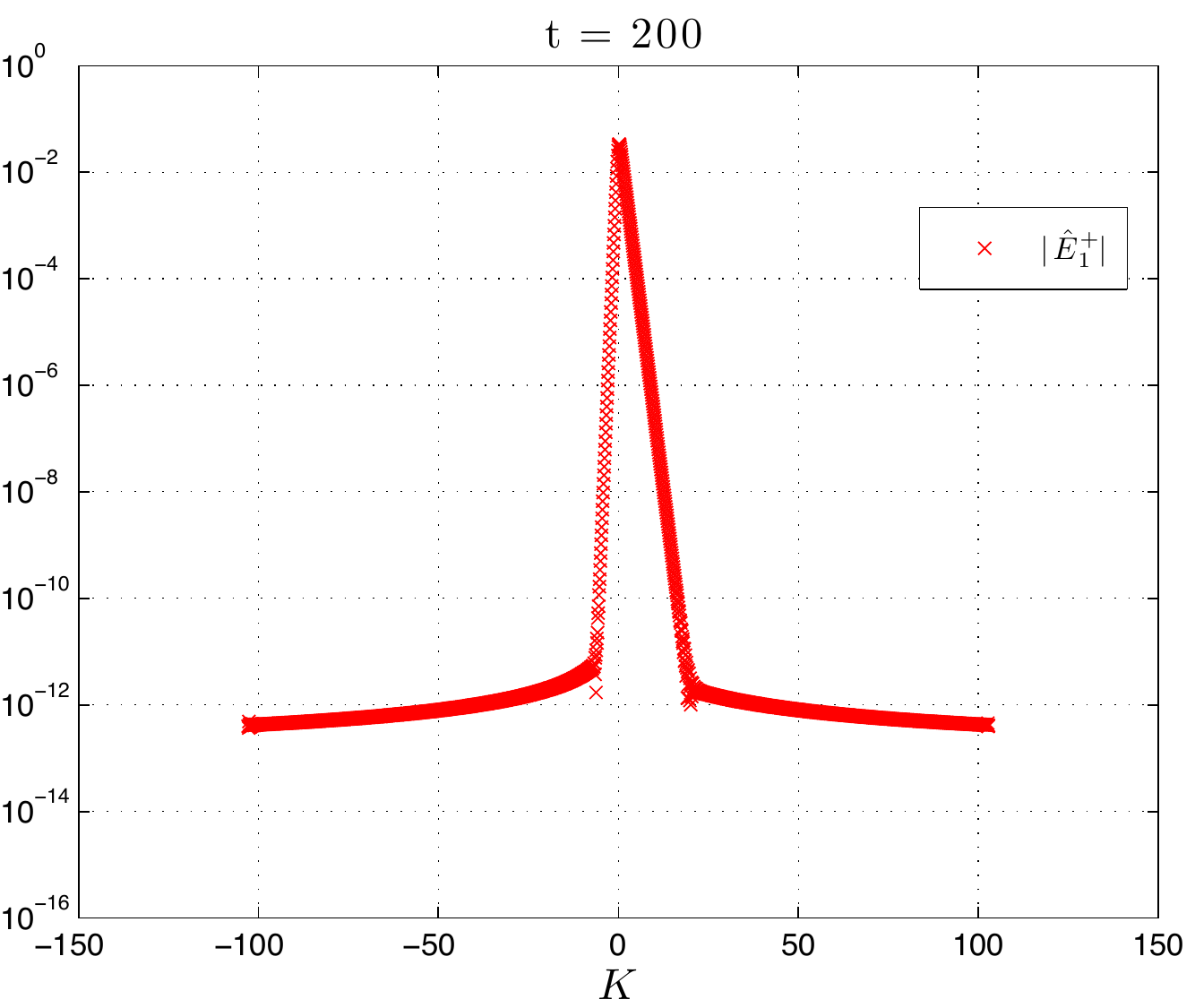}}
  \subfigure[Resolves odd harmonics $\abs{p}\leq
  4$]{\includegraphics[width=2.35in]{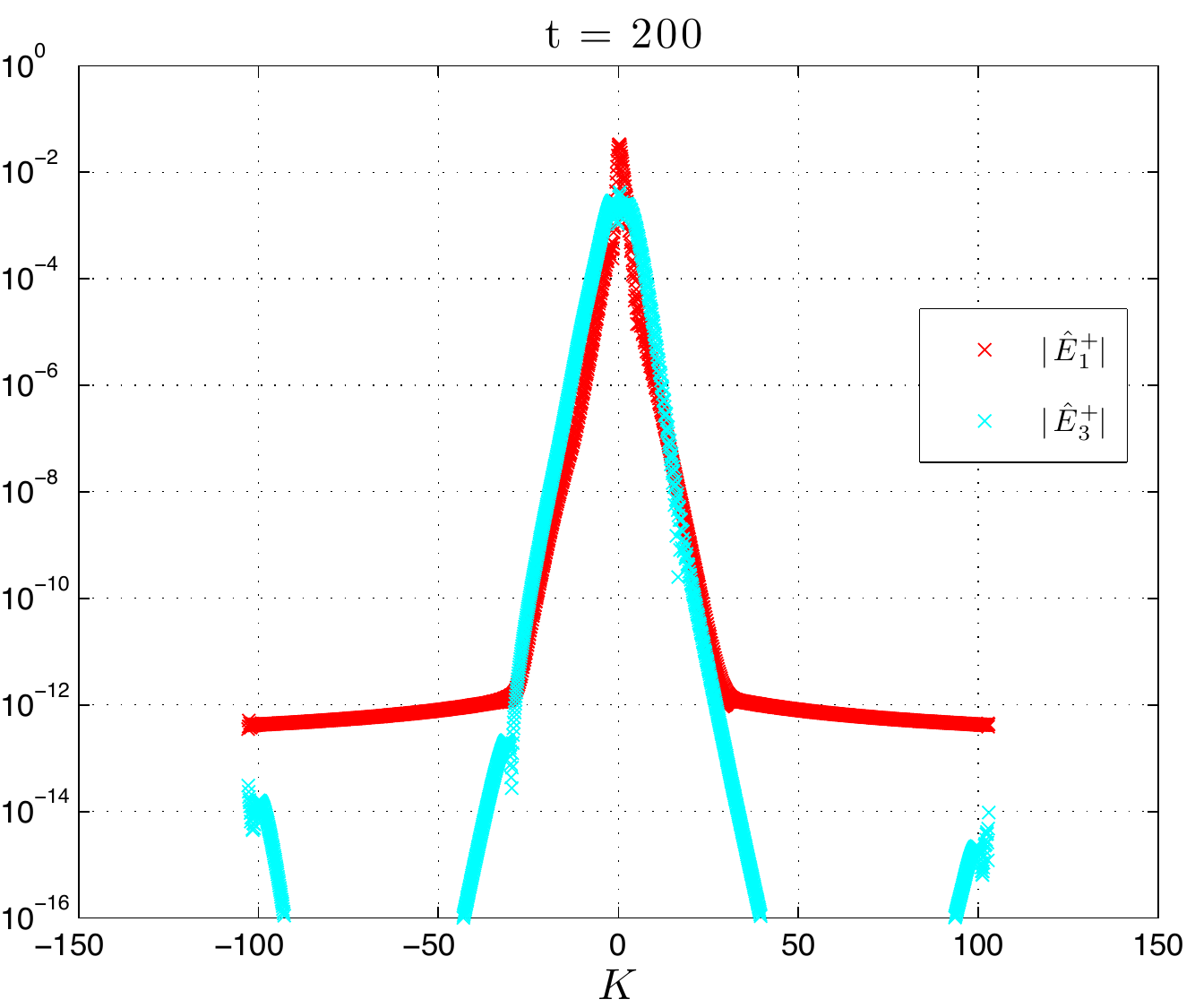}}

  \subfigure[Resolves odd harmonics $\abs{p}\leq
  8$]{\includegraphics[width=2.35in]{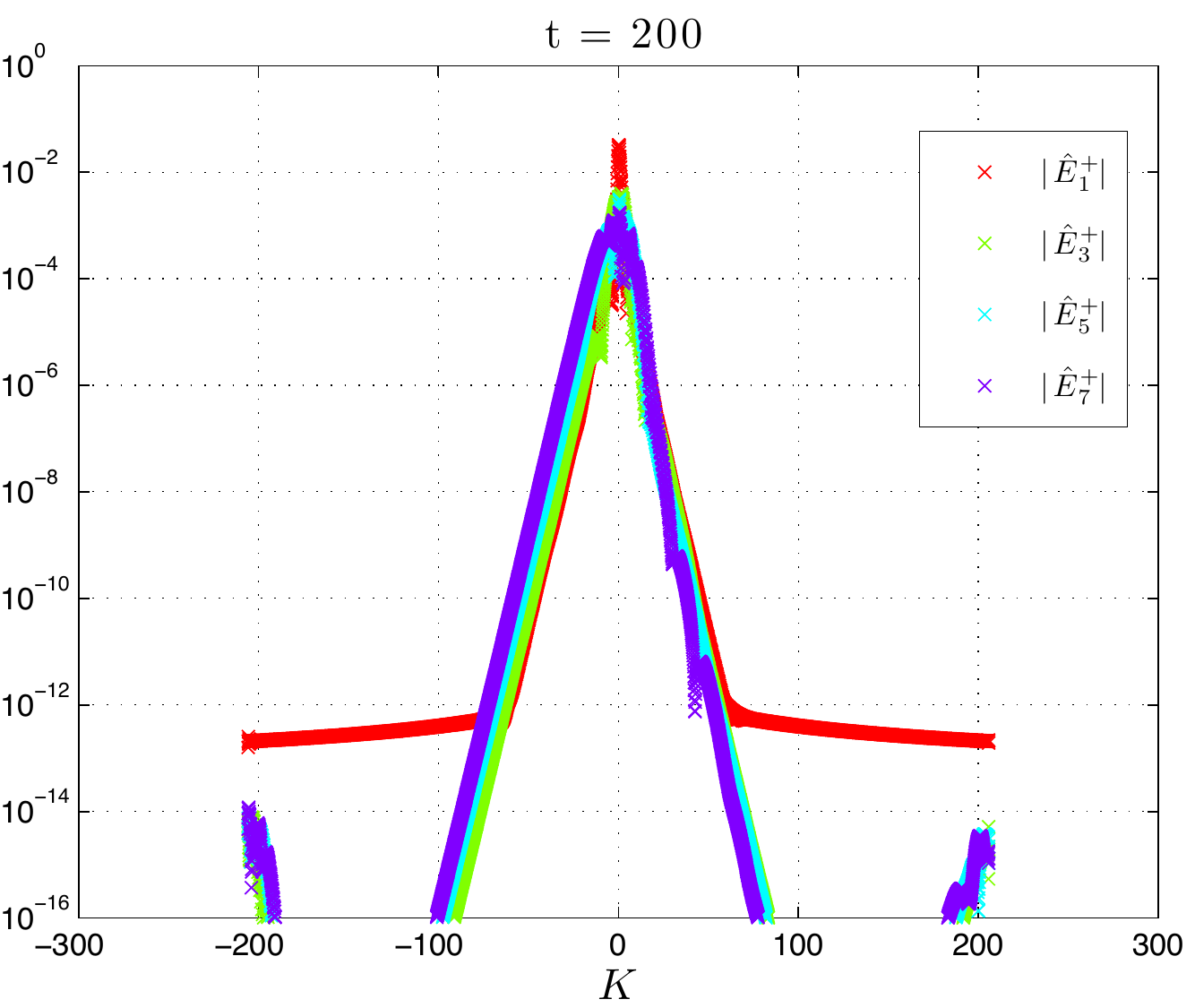}}
  \subfigure[Resolves odd harmonics $\abs{p}\leq
  16$]{\includegraphics[width=2.35in]{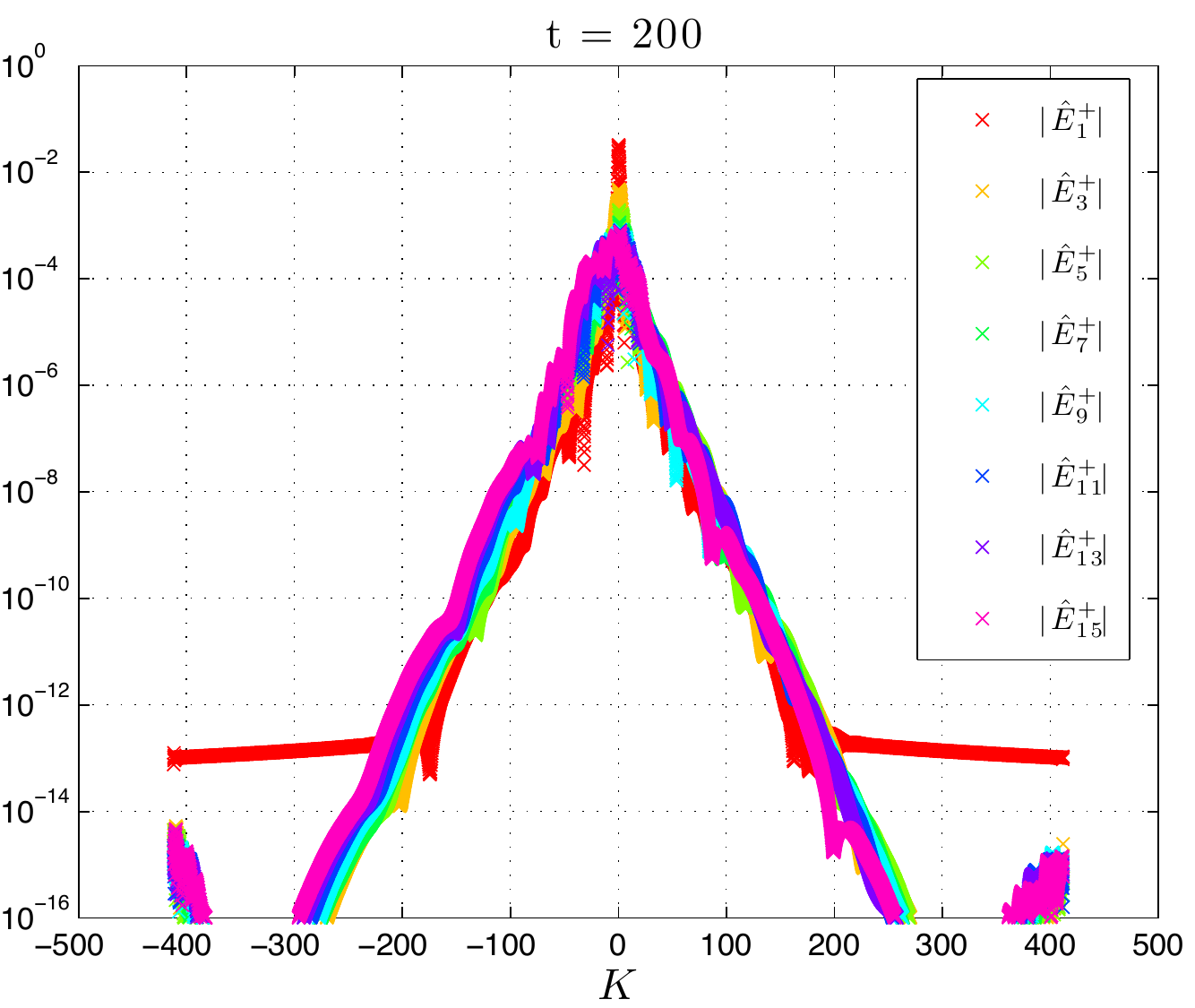}}
  \caption{Fourier transforms of the solutions to truncations of the
    xNLCME equations.  Increasing the number of envelopes expands the
    support in Fourier space.}
  \label{f:fourier_dist}
\end{figure}

Finally, our computations in Section \ref{sec:observations} invoked of
a gas-dynamics entropy condition.  Such a condition is necessary to
use finite volume methods.  Although thermodynamically consistent,
 we do not know whether this is the correct regularizing mechanism of electrodynamics.
\bigskip\bigskip



\bibliographystyle{abbrv}

\bibliography{maxwell_nlgo}


\appendix

\section{The NLCME Soliton}
\label{s:nlcme_soliton}
Using the notation of \cite{goodman01npl}, the NLCME soliton solution
of \eqref{eq:nlcme} is given by:
\begin{subequations}
  \label{eq:NLCME_soliton}
  \begin{gather}
    \mathcal{E}^+(Z,T)  = s \alpha e^{\mathrm{i}\eta} \sqrt{\abs{\frac{N_2}{2\Gamma}}}\frac{1}{\Delta} \sin \delta e^{\mathrm{i} s \sigma} \sech(\theta - \mathrm{i} s \delta/2)\\
    \mathcal{E}^-(Z,T)  = -\alpha e^{\mathrm{i}\eta} \sqrt{\abs{\frac{N_2}{2\Gamma}}}{\Delta} \sin \delta e^{\mathrm{i} s \sigma} \sech(\theta + \mathrm{i} s \delta/2)\\
    \theta = \gamma N_2 \sin\delta (Z - v T),
    \quad \sigma = \gamma N_2 \cos \delta ( v Z - T)\\
    e^{\mathrm{i} \eta}  = \paren{-\frac{e^{2\theta} + e^{-\mathrm{i} s \delta}}{e^{2\theta} + e^{\mathrm{i} s \delta}}}^{2v/(3-v^2)}\\
    \gamma = 1/ \sqrt{1 - v^2}, \quad \Delta  = \paren{\frac{1-v}{1+v}}^{1/4}\\
    s = \sign( N_2 \Gamma), \quad \alpha =
    \sqrt{\frac{2(1-v^2)}{3-v^2}}, \quad \tilde \kappa = \kappa N_2
  \end{gather}
\end{subequations}
We assume that $N_2 \in \R$.  There are two free parameters,
$\abs{v}<1$ and $\delta\in \R$.

\section{Simulating the Nonlinear Maxwell Equations}
\label{sec:methods}
In vector notation, the rescaled Maxwell
system,\eqref{e:rescaled_maxwell}, and constitutive law,
\eqref{e:rescaled_closure}, are expressed as
\begin{equation}
  \label{eq:nonlinhyper}
  \begin{split}
    \dt \begin{pmatrix} D \\ B\end{pmatrix} &+ \dz \begin{pmatrix}
      -B\\-E(D,z)\end{pmatrix}=0 \\
    \dt {\bf v} &+ \dz {\bf f}({\bf v}, z) = 0.
  \end{split}
\end{equation}
To simulate this system of conservation laws, we employ a shock
capturing finite volume scheme with the CLAWPACK software,
\cite{LeVeque:kx,leveque2002fvm}.  Furthermore, we employ the
$\emph{f}$-wave method to accommodate the spatially varying flux
function, \cite{bale2003wpm, leveque2002fvm,leveque2003swl} .

To use finite volume methods we must provide the algorithm with a,
possibly approximate, solution of the Riemman problem.  This
introduces a subtlety as our system has a non-convex flux function.
Non-convex fluxes lead to interesting waves, including rightward (or
leftward) traveling rarefaction and shockwaves that are ``glued''
together.  Such waves, sometimes called {\it compound} or {\it
  composite} waves, were discussed in \cite{wendroff1972arpm,
  wendroff1972brpm, liu1974rpg, liu1975rpg} and more recently in
\cite{wesseling2000umc, wesseling2001uen, muller2006rpe}.  Examples
are also give in the texts \cite{leveque2002fvm, smoller1994swa}.

\subsection{Finite Volume Methods for Maxwell}
\label{sec:nonconvex}

In finite volume numerical methods, at each time step, we must solve a
Riemann problem between adjacent grid cells:
\begin{equation}
  \begin{split}
    \mathbf{v}_t + \mathbf{f}(\mathbf{v};z_j)_z=0\quad\text{for $ z_{j-1/2} < z < z_{j+1/2}$}, \\
    \mathbf{v}_t + \mathbf{f}(\mathbf{v};z_{j+1})_z=0\quad\text{for $ z_{j+1/2} < z < z_{j+3/2}$}, \\
    \mathbf{v}(z, t=t^n) = \begin{cases} \mathbf{v}_j^n & \text{ for $
        z_{j-1/2} < z < z_{j+1/2}$}, \\ \mathbf{v}_{j+1}^n & \text{
        for $z_{j+1/2}< z< z_{j+3/2}$}. \end{cases}
  \end{split}
\end{equation}
$z_{j+1/2}$ is the interface between the cell centered at $z_j$ and
the cell centered at $z_{j+1/2}$.  The fluxes are assumed to be
constant in $z$ within each computational cell. We aim to provide an
exact solution to the Riemann problem, in contrast to an approximate
solutions such as the Roe average.

In the next few sections, we adopt the notation
\begin{equation}
  \begin{split}
    \mathbf{v}_t + \mathbf{f}_l(\mathbf{v})_z=0\quad\text{for $ z <0$} \\
    \mathbf{v}_t + \mathbf{f}_r(\mathbf{v})_z=0\quad\text{for $z>0$} \\
    \mathbf{v}(z,0) = \begin{cases}
      \mathbf{v}_l & \text{ for $ z<0$}\\
      \mathbf{v}_r& \text{ for $z>0$} \end{cases}
  \end{split}
\end{equation}
where $ \mathbf{f}_l(\mathbf{v})=\mathbf{f}(\mathbf{v};z_j)$,
$\mathbf{f}_r(\mathbf{v})=\mathbf{f}(\mathbf{v};z_{j+1})$, and we take
$z_{j+1/2} = 0$.

Given any point $\mathbf{v}$ in $(D,B)$ phase space, we construct two,
entropy condition (specified below) satisfying, manifolds
$\mathcal{W}_1(\mathbf{v})$ and $\mathcal{W}_2(\mathbf{v})$.  These
define the locus of points that can be joined to $\mathbf{v}$ by a
left-going wave in the former case and a right-going wave in the
latter case.  We parameterize them in the $D$ component.  Given a
state $\mathbf{v}_0$, $W_j(D; \mathbf{v}_0)$ is the parametric curve
such that:
\[
\set{\begin{pmatrix}
    D\\
    W_j(D; \mathbf{v}_0)
  \end{pmatrix},\quad \textrm{for $D \in \R$}} =
\mathcal{W}_j(\mathbf{v}_0)\quad \text{for $j=1,2$}
\]

Were the medium homogeneous, solving the Riemann corresponds to
finding the state $\mathbf{v}_\star$ that is the unique point in
$\mathcal{W}_1(\mathbf{v}_l) \bigcap \mathcal{W}_2(\mathbf{v}_r)$.  In
terms of the parametric curves, this point solves the equation:
\begin{equation}
  W_2\paren{D_r;\begin{pmatrix}D_\star \\ W_1(D_\star; \mathbf{v}_l)\end{pmatrix}}= B_r
\end{equation}
As the medium is not homogeneous, we match the flux at the interface.
We seek $\mathbf{v}_\star^l$ and $\mathbf{v}_\star^r$ such that:
\begin{subequations}\label{e:rp}
  \begin{align}
    \label{eq:rp1}
    W_1(D^\star_l; \mathbf{v}_l) & = B^\star_l\\
    \label{eq:rp2}
    \mathbf{f}_l(\mathbf{v}^\star_l) & = \mathbf{f}_r(\mathbf{v}^\star_r)\\
    \label{eq:rp3}
    W_2(D_r; \mathbf{v}^\star_r) & = B_r
  \end{align}
\end{subequations}
$\mathbf{v}^\star_l$ is the the entropy satisfying state immediately
to the left of the interface and $\mathbf{v}^\star_r$ is the entropy
satisfying state immediately to the right of the interface.

For this problem, the flux matching condition is:
\begin{subequations}
  \begin{align}
    E(D^\star_l; z_l) &= E(D^\star_r; z_r)\\
    B^\star_l &= B^\star_r
  \end{align}
\end{subequations}
Defining transfer function, $\mathcal{T}$, that, given $z_l$, $z_r$
and a left state $D^\star_l$, the flux matched displacement is:
\begin{equation}
  \mathcal{T}(D^\star_l; z_l, z_r) = D^\star_r
\end{equation}
With this function, \eqref{e:rp} becomes
\begin{equation}
  W_2\paren{D_r; \begin{pmatrix}
      \mathcal{T}(D^\star_l; z_l; z_r)\\
      W_1(D^\star_l;\mathbf{v}_l)
    \end{pmatrix}}= B_r
\end{equation}
$D^\star_l$ is the unknown.  Once we have this value, we recover
$E^\star$ and $B^\star$ allowing us to compute the fluxes.  Subject to
the specification of the phase space functions $W_j$ are specified,
this is a one-dimensional root finding problem.

\subsection{Non-Convex Fluxes and the Entropy Condition}

It remains to specify the manifolds $\mathcal{W}_j$.  This requires an
additional, non-trivial, assumption on an entropy condition.  While
such a condition is readily apparent in gas dynamics and elasticity,
the appropriate condition for Maxwell is non-obvious.


In this work, we employ a diffusive entropy condition, akin to that
found in gas dynamics.  This was suggested by Sj\"oberg
\cite{sjoberg:uac}, as part of an entropy-flux pair involving the
Poynting vector.  This is also a physically consistent, as many
dielectrics absorb the higher harmonics that would appear as the wave
began to shock.

In constructing the entropy satisfying $W_j$ functions, we closely
follow \cite{liu1974rpg,liu1975rpg,wendroff1972brpm,Muller:2006p1678}
and particularly the $p$-system example in \cite{wendroff1972arpm}.
Graphically, the $W_j$ functions can be constructed by tracing an
appropriate convex hull of $E(D;z)$.  Shock waves occur when points
are joined by chords, rarefaction waves when points are joined along
$E(D;z)$, and composite waves when convex curve is a combination.

{\bf Throughout this section we suppress the $z$ argument, and $E'(D)
  = \partial_D E(D,z)$.}

Since the flux function is no longer uniformly convex, the Lax entropy
condition may not be appropriate.  Instead, Liu entropy condition
\cite{liu1974rpg, Muller:2006p1678} may apply.  Recall, that if
\begin{equation}
  \sigma({\bf v}_0,{\bf v}) \equiv \frac{-E(D)+ E(D_0)}{B-B_0}
\end{equation}
then:
\begin{itemize}
\item A shock joining ${\bf v}_0$ to ${\bf v}_1$ satisfies the Lax
  entropy condition if the system is convex and either
  \begin{subequations}\label{e:lax_condition}
    \begin{align}
      \lambda_1({\bf v}_1) < \sigma < \lambda_1({\bf v}_0), &
      \quad\sigma <
      0\\
      \lambda_2({\bf v}_1) < \sigma < \lambda_2({\bf v}_0), &
      \quad\sigma>0
    \end{align}
  \end{subequations}
  where $\lambda_1 < 0 < \lambda_2$ are the eigenvalues of
  \begin{equation*}
    \begin{pmatrix}
      0 & -1 \\
      - E'(D) & 0
    \end{pmatrix}.
  \end{equation*}
\item A shock joining ${\bf v}_0$ to ${\bf v}_1$ satisfies the Liu
  entropy condition if
  \begin{equation}
    \label{e:liu_condition}
    \sigma({\bf v}_0,{\bf v}_1)\leq \sigma({\bf v}_0,\tilde{\bf {v}})
  \end{equation}
  for all points $\tilde{\bf v}$ between the two points along the
  shock curve in phase space.
\end{itemize}


\subsubsection{Left Traveling Waves}
Given the state $\mathbf{u}_0 = ( D_0, B_0)^T$, we construct $W_1(D;
\mathbf{u}_0)$.  Since we have an inflection point in $E(D)$ at $D=0$,
we dissect all the possible configurations of $D_0$, $D$ and the
inflection point.  Let $D_1$ be the value of $D$ at which the line
tangent to $(D_1, E(D_1))$ intercepts $(D_0, E(D_0))$.

First, suppose that $D< D_0 < 0$, as in Figure
\ref{f:leftward_rarefaction_neg}.  In this region, there is no
difficulty applying the Lax entropy condition \eqref{e:lax_condition};
there is no shock as
\[
\lambda_1(D) = -\sqrt{E'(D)}< \lambda_1(D_0) = -\sqrt{E'(D_0)}< 0
\]
Consequently,
\[
B = B_0 + \int_{D_0}^D \sqrt{E'(s)}ds
\]
Still assuming that $D_0 < 0$, if $ D_0 < D < D_1$ the Lax condition
continues to apply.  $D_0$ and $D$ will satisfy
\eqref{e:lax_condition} and we can join the two with a shock, as in
Figure \ref{f:leftward_shock_neg};
\[
B = B_0 + \sqrt{\abs{E(D) - E_0}\abs{D-D_0}}
\]

Once $D>D_1$, the solution changes.  It is no longer appropriate to
apply the Lax condition as we lose convexity here.  Applying the Liu
condition, \eqref{e:liu_condition}, we see that there is no longer a
shock.  Indeed, we can compute that were there shock solutions,
\begin{align*}
  \sigma({\bf v}_0, {\bf v}) &= \frac{-E + E_0}{B-B_0} = -\sqrt{\abs{\frac{E - E_0}{D-D_0} }}\\
  \sigma({\bf v}_0, {\bf v}_1) &= \frac{-E_1 +
    E_0}{B_1-B_0}=-\sqrt{\abs{\frac{E_1 - E_0}{D_1-D_0} }}
\end{align*}
Examining Figure \ref{f:leftward_compound_neg},
\[
\frac{E_1 - E_0}{D_1-D_0}<\frac{E - E_0}{D-D_0}<0
\]
implying
\[
\sigma({\bf v}_0,{\bf v}_1) < \sigma({\bf v}_0,{\bf v})
\]
which violates \eqref{e:liu_condition}.

As a shock fails to connect the two states, we resort to joining the
states by a compound wave.  The solution is a shock from $D_0$ to
$D_1$ which continues into a rarefaction wave from $D_1$ to $D$.
Thus,
\[
B = B_0 + \sqrt{\abs{E(D_1) - E_0}\abs{D_1-D_0}} + \int_{D_1}^D
\sqrt{E'(s)}ds
\]
\begin{figure}
  \centering \subfigure[$D< D_0$]{
    \includegraphics[width=2in]{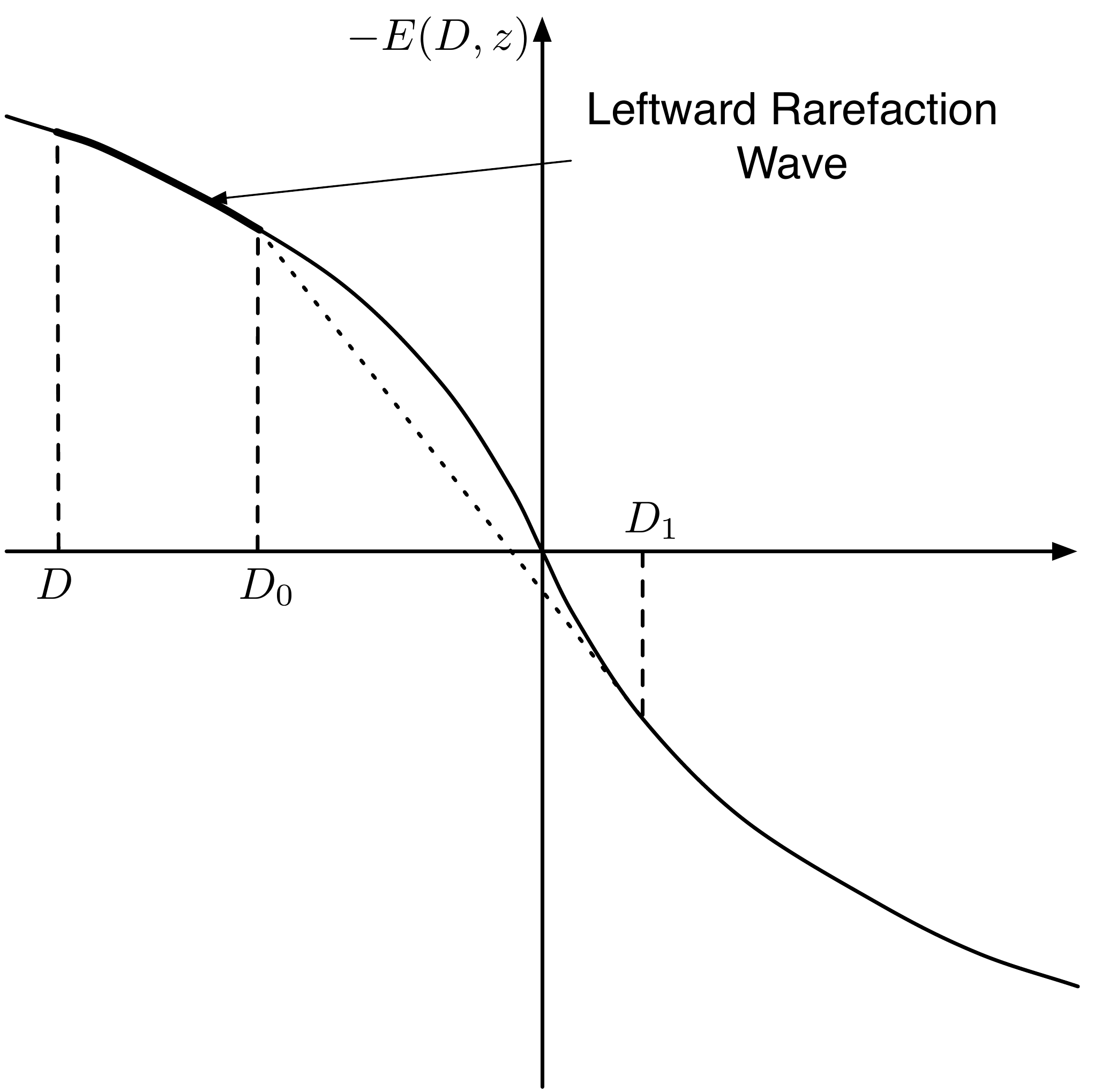}\label{f:leftward_rarefaction_neg}
  } \subfigure[$D_0< D< D_1$]{
    \includegraphics[width=2in]{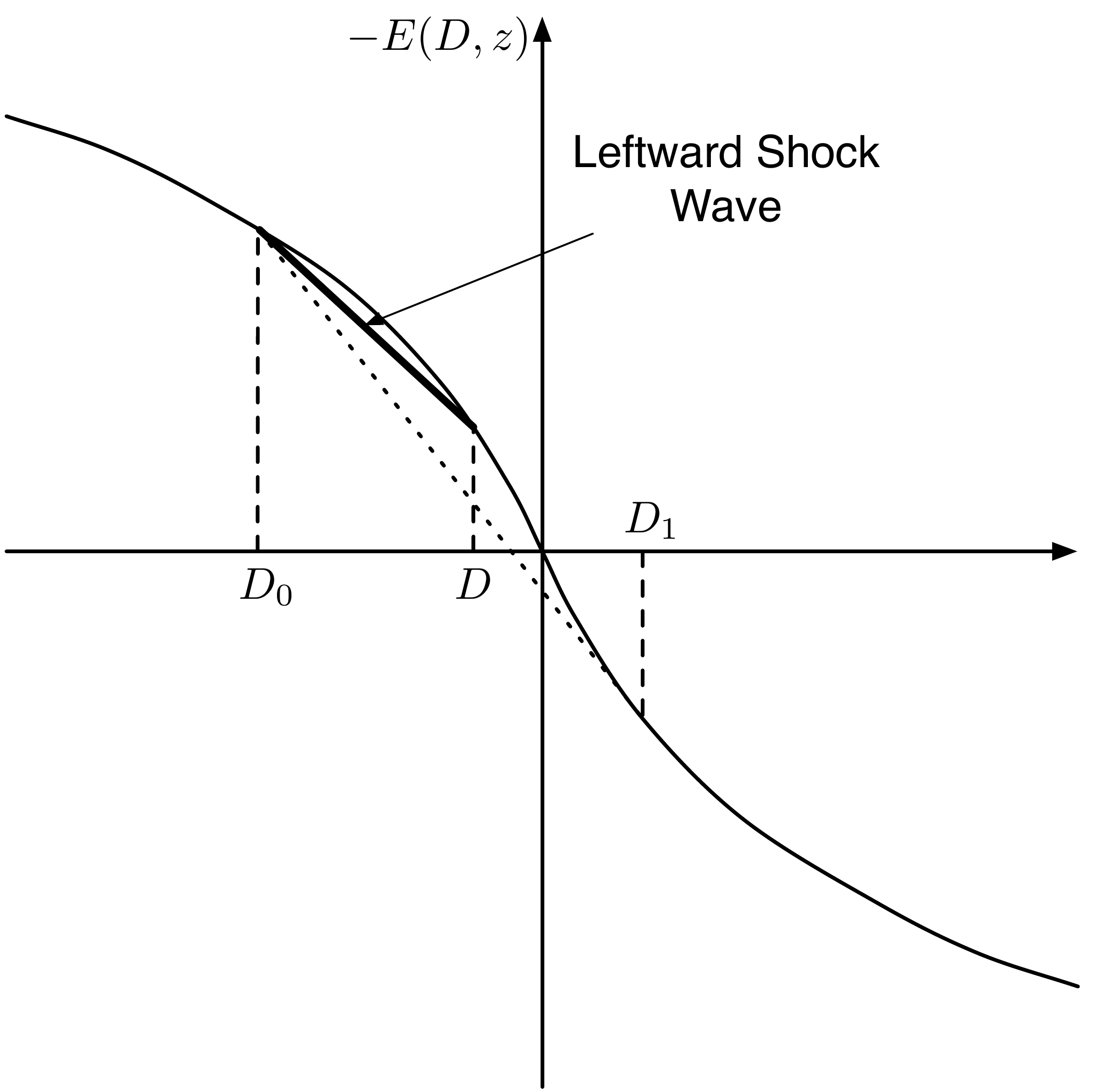}\label{f:leftward_shock_neg}
  }
  \subfigure[$D_0< D_1<D$]{
    \includegraphics[width=2in]{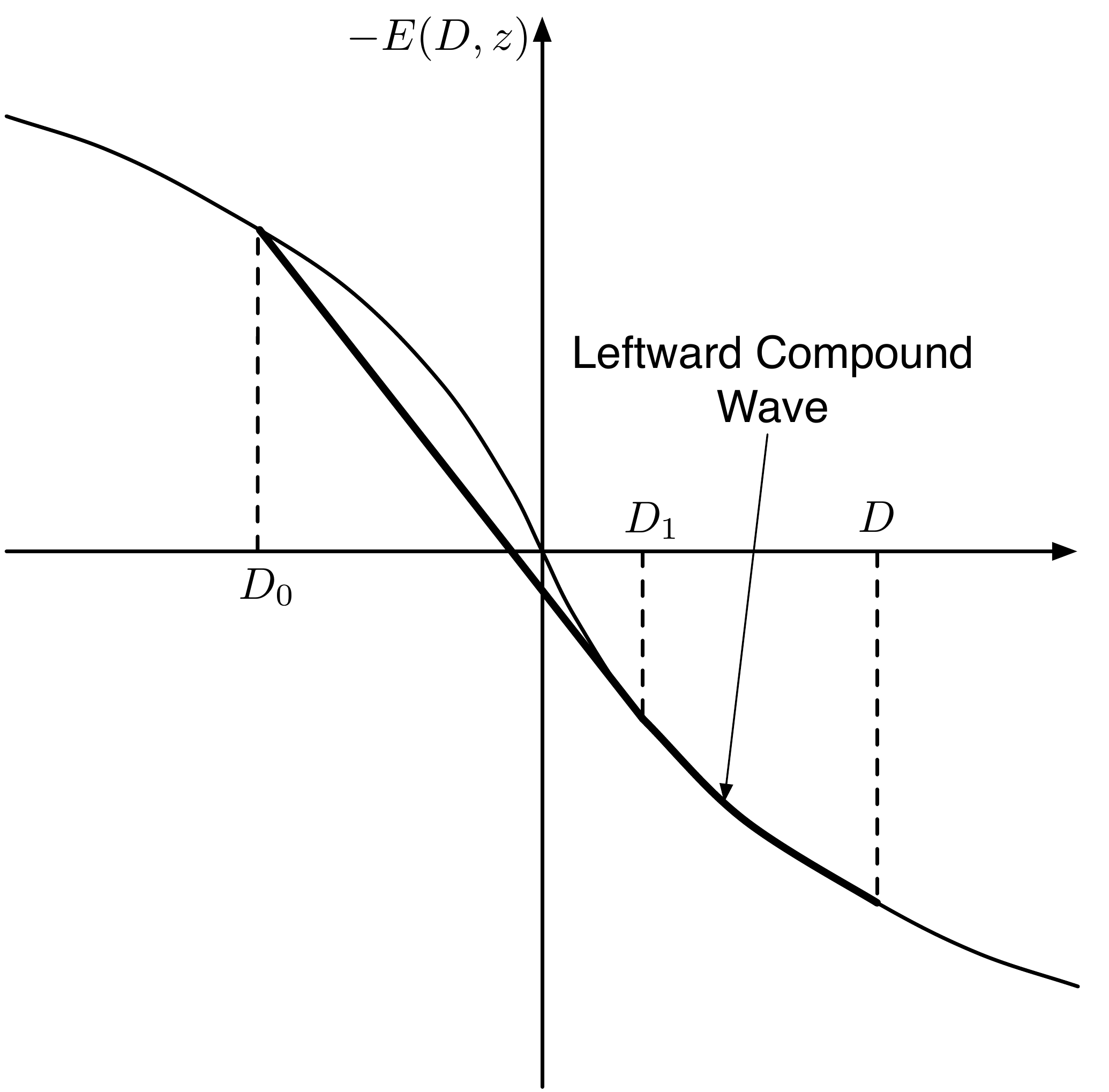}\label{f:leftward_compound_neg}
  }
  \caption{Entropy satisfying leftward traveling waves when $D_0<0$.}
  \label{f:leftward_waves_neg}
\end{figure}

At $D_0=0$ the system is convex yielding leftward traveling
rarefaction waves for all values of $D$:
\[
B = B_0 + \int_{D_0}^D \sqrt{E'(s)}ds
\]

For $D_0 >0$, There are again several cases.  For $D < D_1$, we have
the compound wave again, as in Figure \ref{f:leftward_compound_pos}:
\[
B = B_0 - \sqrt{\abs{E(D_1) - E_0}\abs{D_1-D_0}} + \int_{D_1}^D
\sqrt{E'(s)}ds
\]
As $D$ increases in value and $D_1 < D < D_0$, we have a shock
solution
\[
B = B_0 - \sqrt{\abs{E(D) - E_0}\abs{D-D_0}}
\]
See Figure \ref{f:leftward_shock_pos} Lastly, for $D> D_0$, we get a
leftward traveling rarefaction:
\[
B = B_0 + \int_{D_0}^D \sqrt{E'(s)}ds
\]

\begin{figure}
  \centering \subfigure[$D< D_1< D_0$]{
    \includegraphics[width=2in]{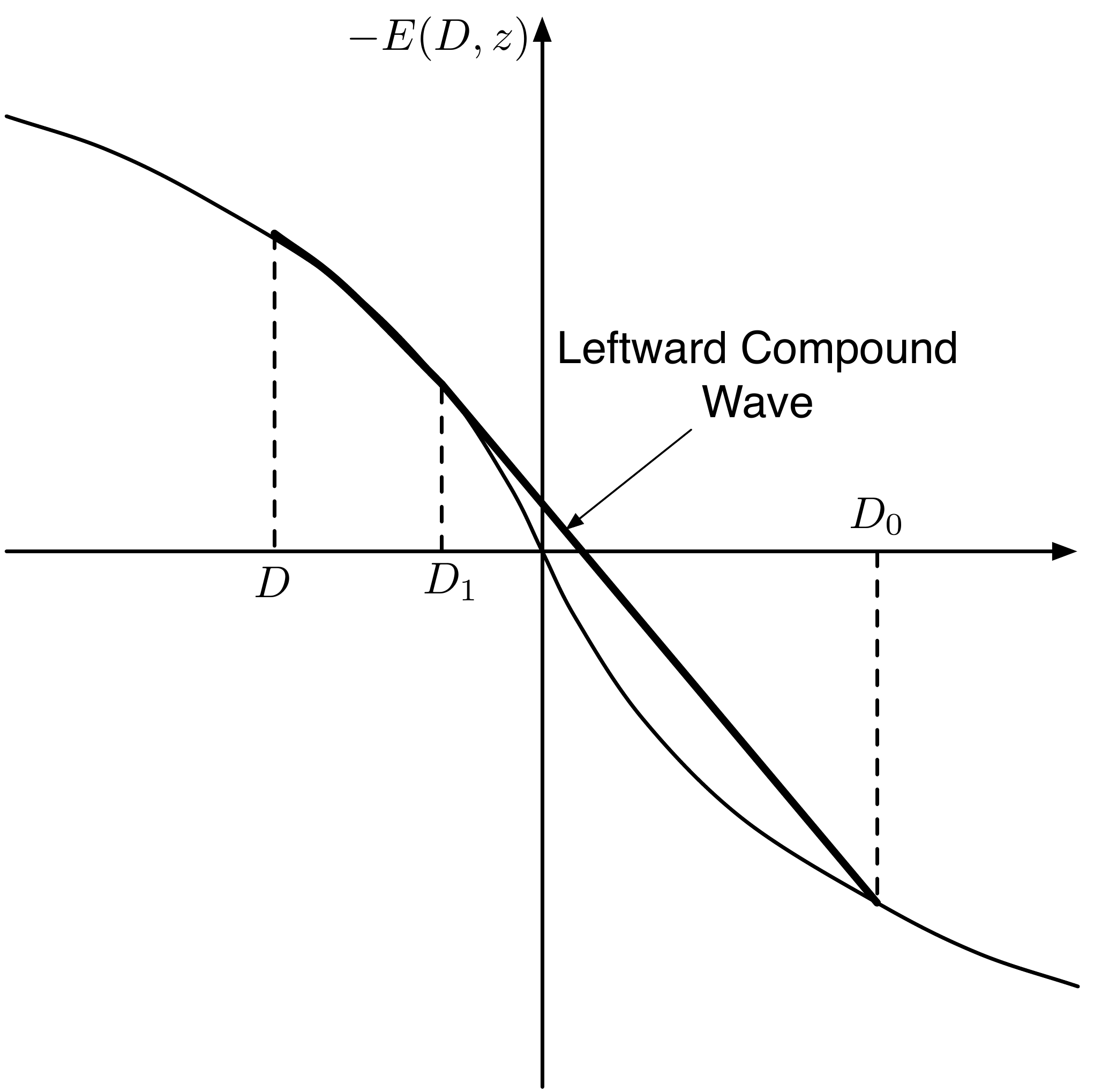}\label{f:leftward_compound_pos}
  } \subfigure[$D_1<D< D_0$]{
    \includegraphics[width=2in]{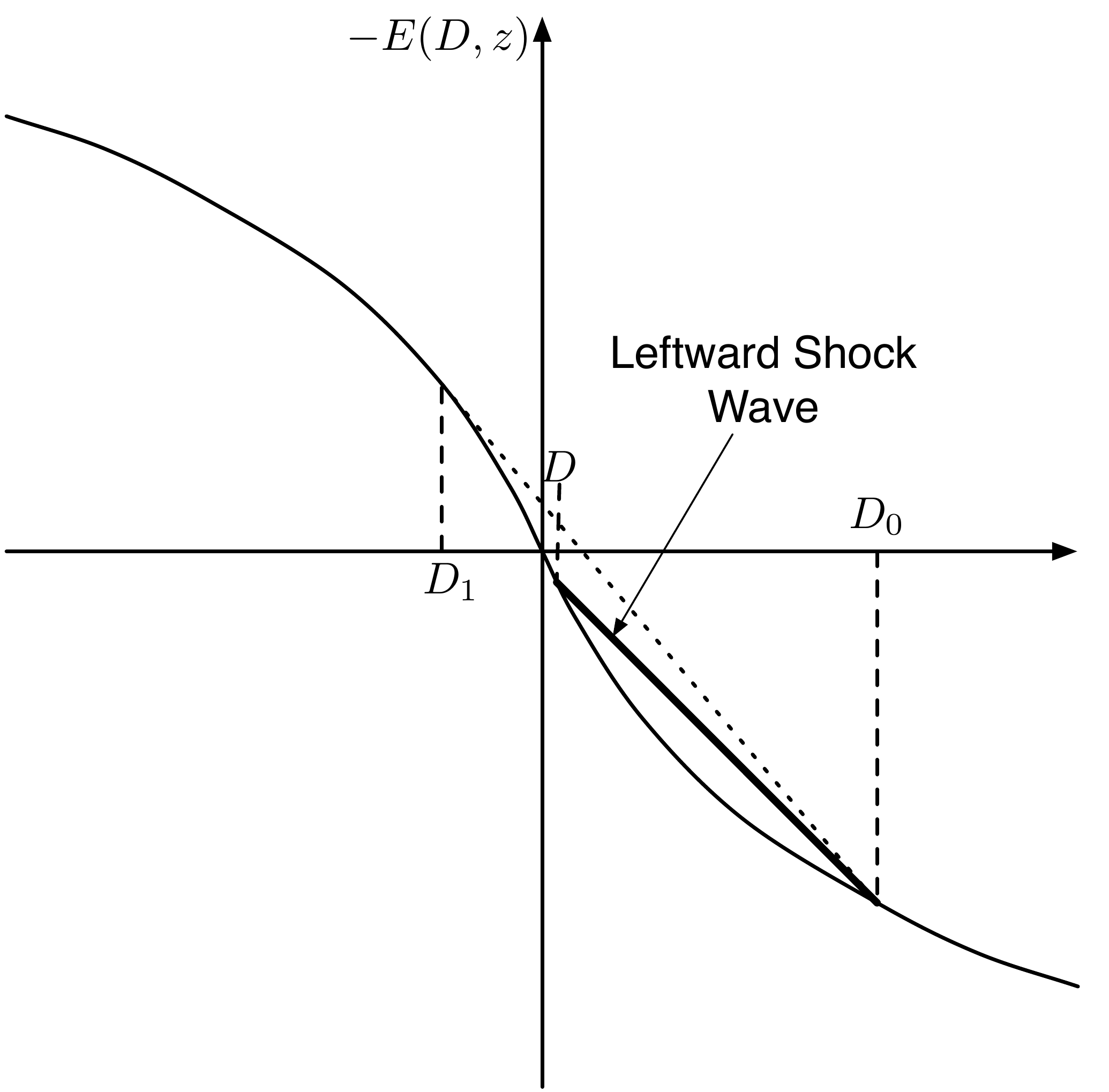}\label{f:leftward_shock_pos}
  } \subfigure[$D_0< D$]{
    \includegraphics[width=2in]{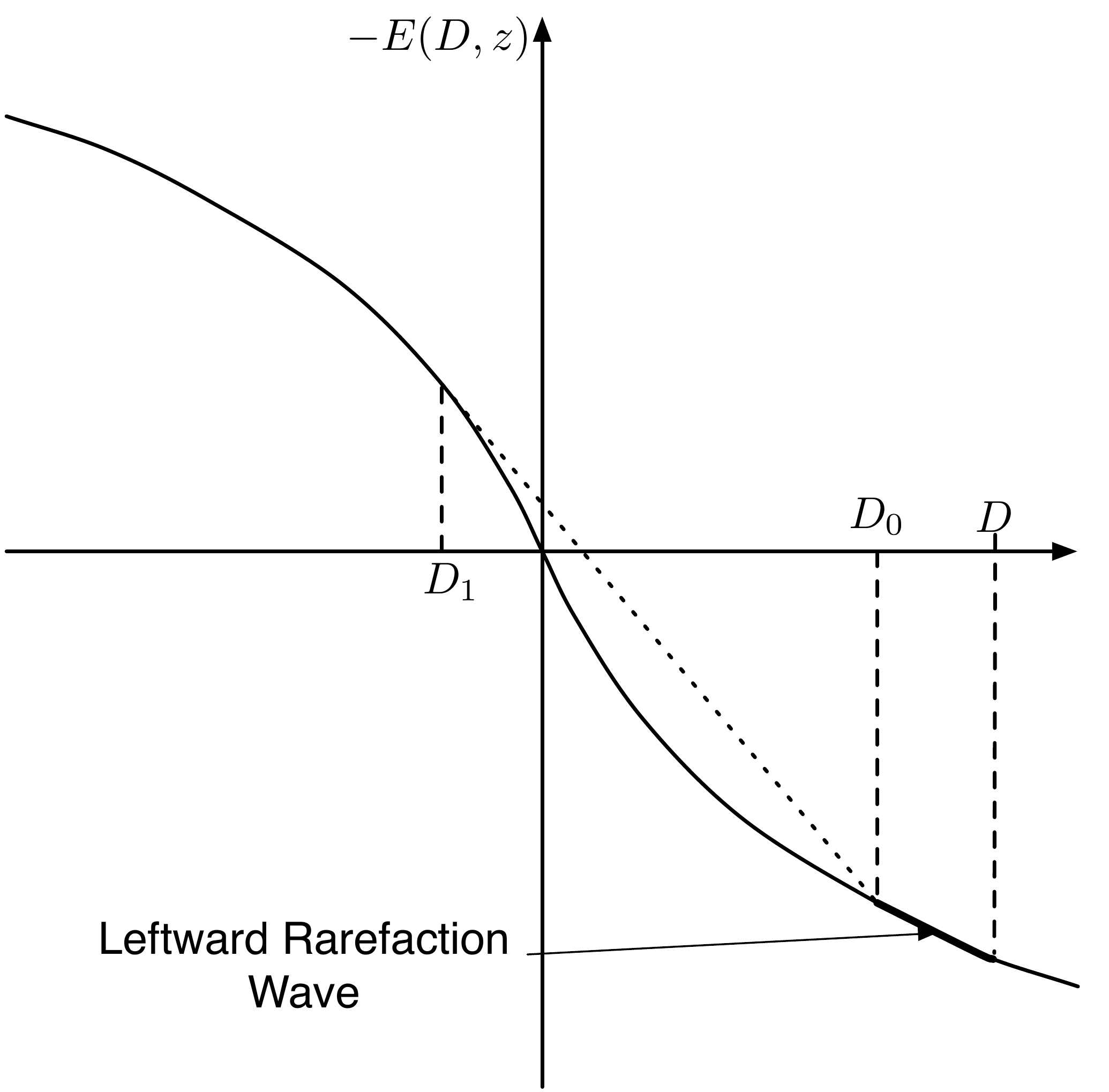}\label{f:leftward_rarefaction_pos}
  }

  \caption{Entropy satisfying leftward traveling waves when $D_0>0$.}
  \label{f:leftward_waves_pos}
\end{figure}

\subsubsection{Right Traveling Waves}
For right traveling waves, the structure is similar.  Given our point
$D_0$, let $(D_2, E_2)$ be the point intercepted by the line tangent
to $(D_0, E(D_0))$.

Again, we first treat $D_0 < 0$.  The different cases are diagrammed
in Figure \ref{f:rightward_waves_neg}.  For $D< D_0<0$, there is a
shock,
\[
B = B_0 + \sqrt{\abs{E(D) - E_0}\abs{D-D_0}}.
\]
For $D_0 < D < 0$, this changes to a rarefaction wave,
\[
B = B_0 - \int_{D_0}^D \sqrt{E'(s)}ds.
\]

Crossing the inflection point, $0 < D< D_2$, it becomes a compound
wave which rarefacts to the point $D_\star$ followed by a shock,
\[
B = B_0 - \int_{D_0}^{D_\star} \sqrt{E'(s)}ds - \sqrt{\abs{E(D) -
    E_\star}\abs{D-D_\star}}.
\]
$D_\star$ is the point $(D_\star, E_\star)$ on the curve whose tangent
intercepts $(D, E)$.  Past $D_2$, the compound wave reduces to a
shock, as the system now satisfies \eqref{e:liu_condition},
\[
B = B_0 - \sqrt{\abs{E(D) - E_0}\abs{D-D_0}}.
\]

\begin{figure}
  \centering \subfigure[$D< D_0$]{
    \includegraphics[width=2in]{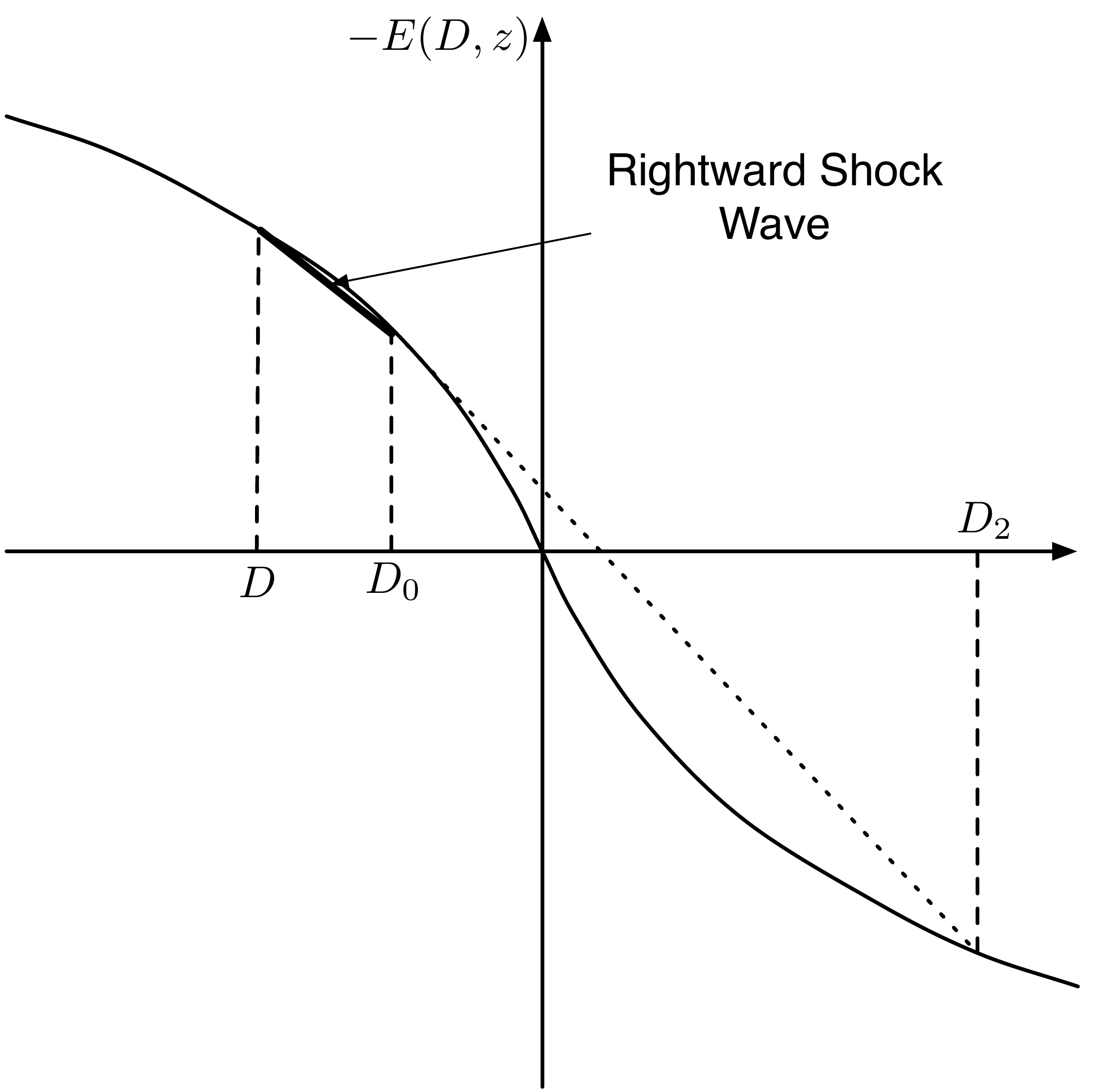}\label{f:rightward_shock_neg}
  } \subfigure[$D_0<D< 0$]{
    \includegraphics[width=2in]{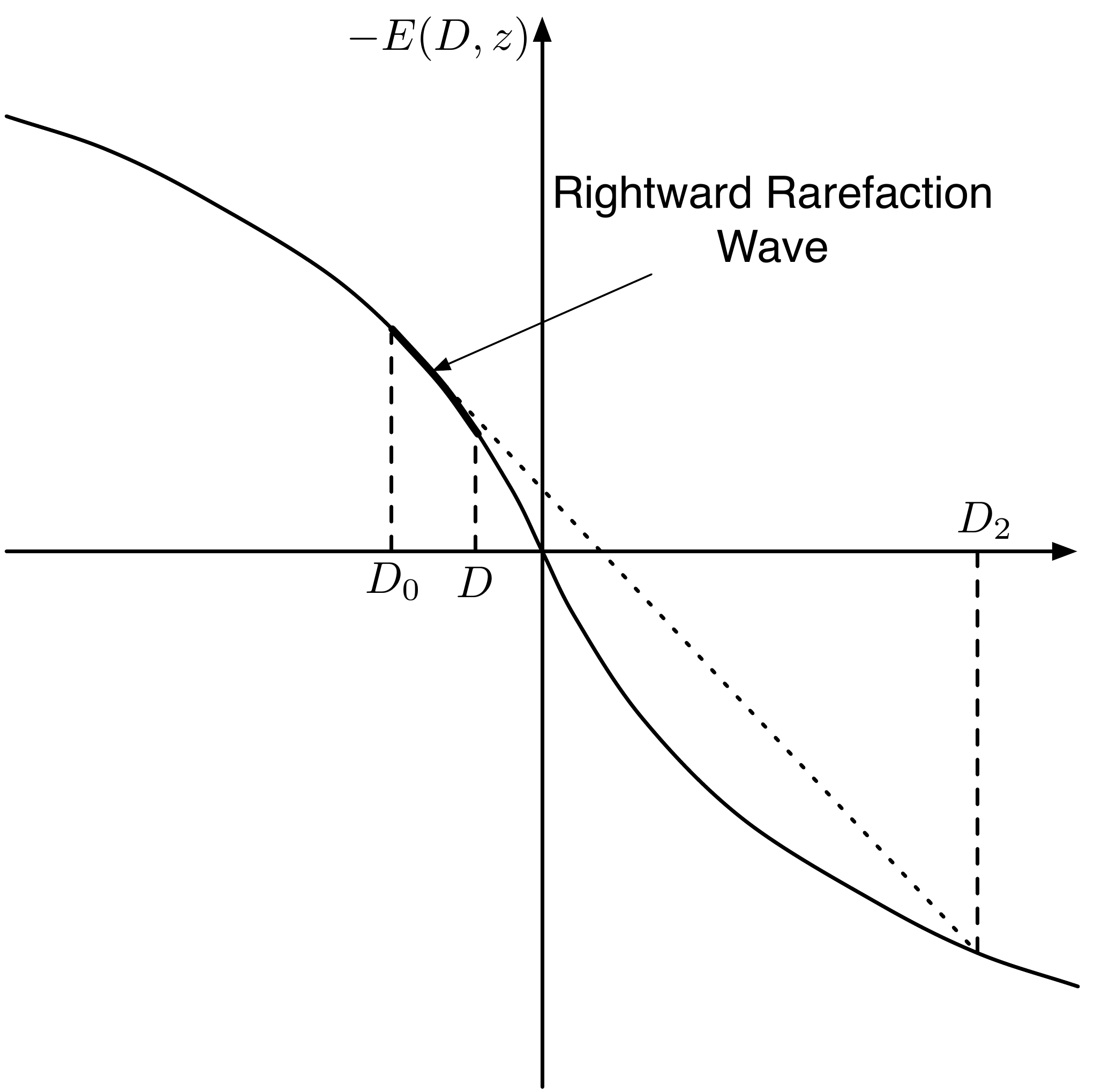}\label{f:rightward_rarefaction_neg}
  } 

\subfigure[$0< D< D_2$]{
    \includegraphics[width=2in]{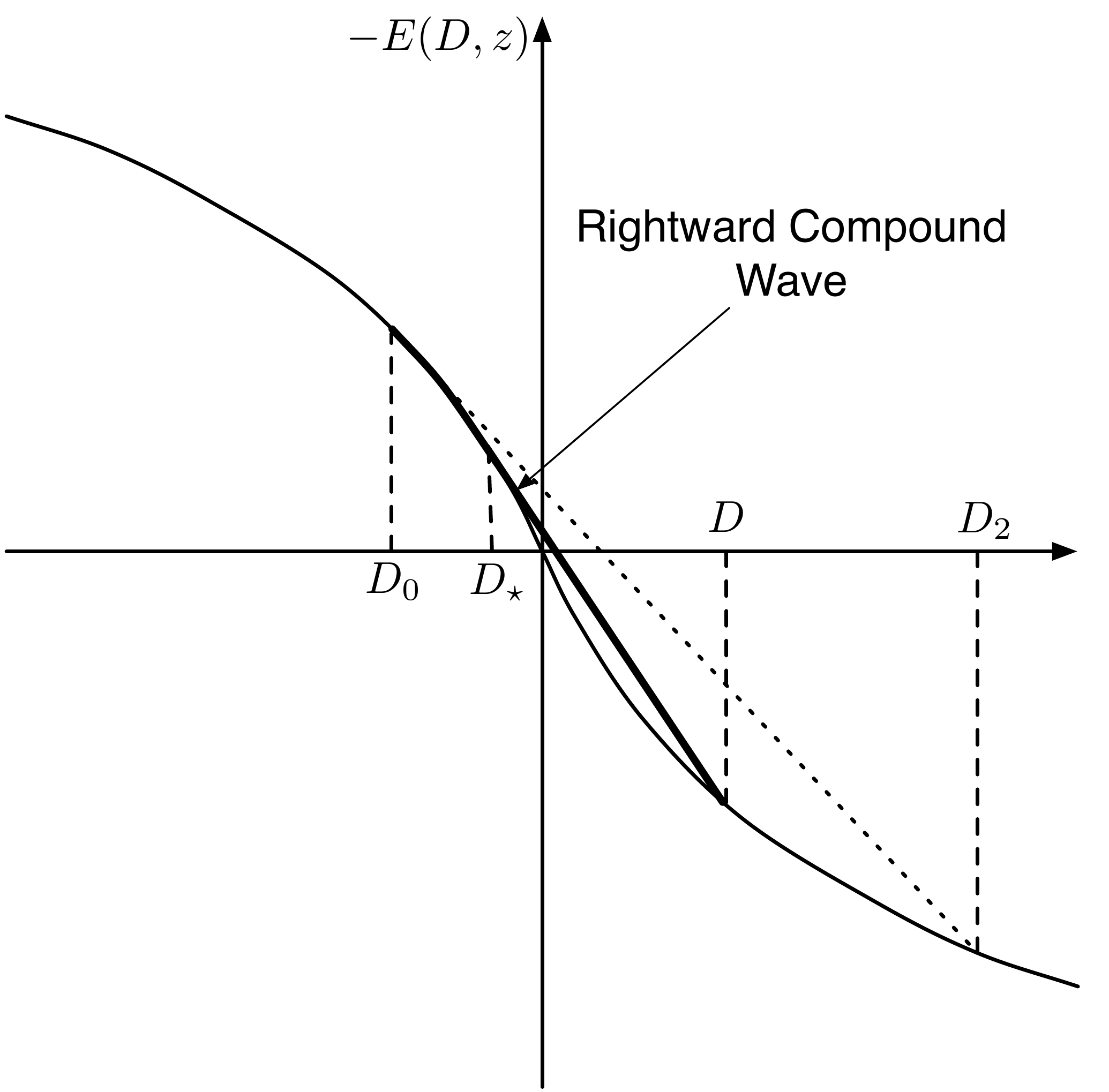}\label{f:rightward_compound_neg}
  } \subfigure[$D_2 < D$]{
    \includegraphics[width=2in]{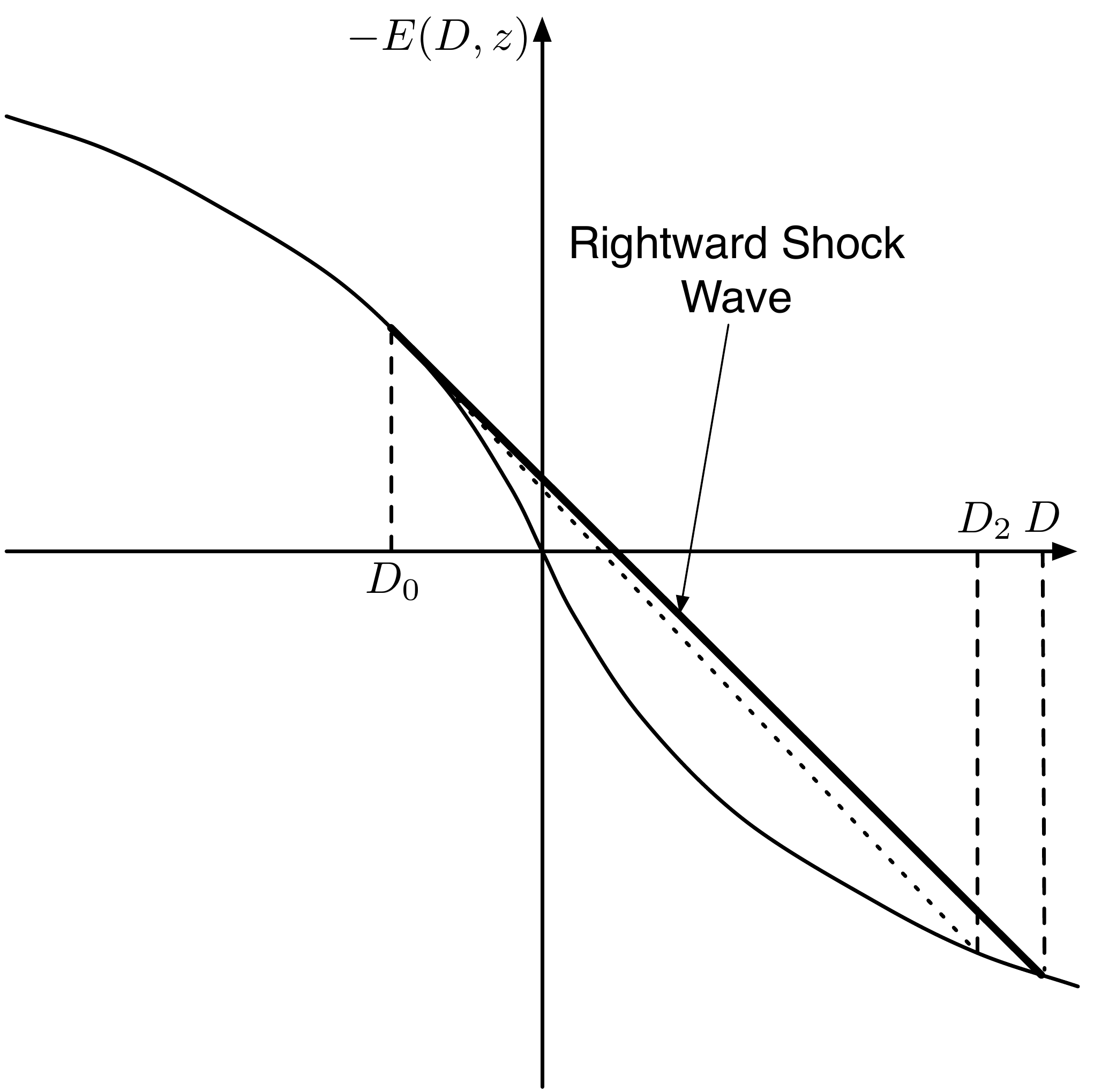}\label{f:rightward_shock2_neg}
  }
  \caption{Entropy satisfying rightward traveling waves when $D_0<0$.}
  \label{f:rightward_waves_neg}
\end{figure}

For $D_0=0$, we have a shock in both directions,
\[
B = B_0 - \sign(D) \sqrt{\abs{E(D) - E_0}\abs{D-D_0}}.
\]

For $D_0>0$, again, we must consider the different positions of $D$
relative to the other points.  These cases appear in Figure
\ref{f:rightward_waves_pos}.  If $D < D_2<0$, there is the shock
solution satisfying \eqref{e:liu_condition},
\[
B = B_0 + \sqrt{\abs{E(D) - E_0}\abs{D-D_0}}.
\]
For $D_2 < D< 0$, this becomes a compound wave,
\[
B = B_0 - \int_{D_0}^{D_\star} \sqrt{E'(s)}ds + \sqrt{\abs{E(D) -
    E_\star}\abs{D-D_\star}}
\]
$D_\star$ is again the point on the curve whose tangent intercepts
$(D,E(D))$.  For $0< D < D_0$, this becomes a purely rarefactory wave,
\[
B = B_0 - \int_{D_0}^{D} \sqrt{E'(s)}ds.
\]
Finally, for $D> D_0$, we again have a shock,
\[
B = B_0 - \sqrt{\abs{E(D) - E_0}\abs{D-D_0}}.
\]

\begin{figure}
  \centering \subfigure[$D< D_2$]{
    \includegraphics[width=2in]{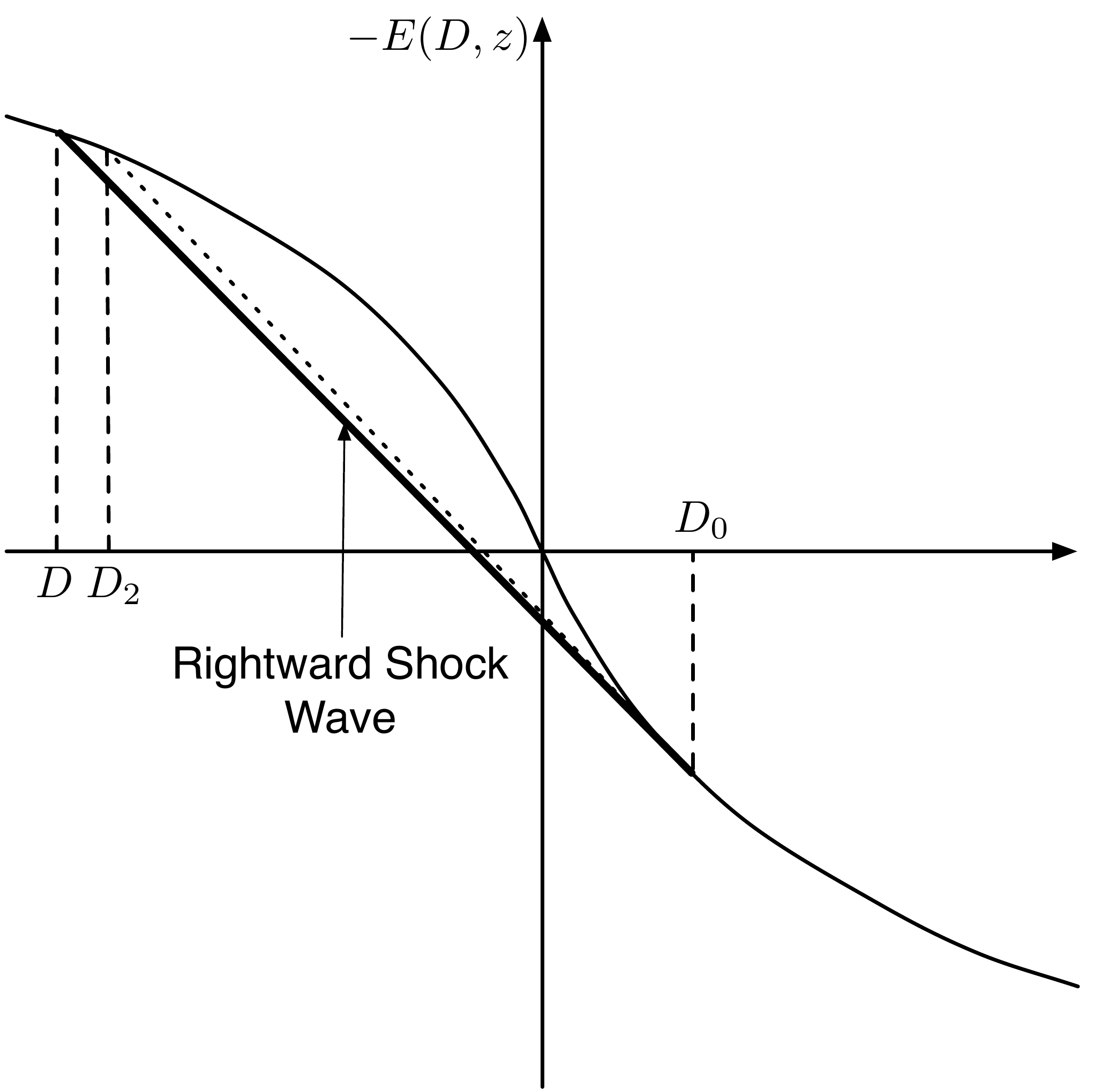}\label{f:rightward_shock_pos}
  } \subfigure[$D_2<D< 0$]{
    \includegraphics[width=2in]{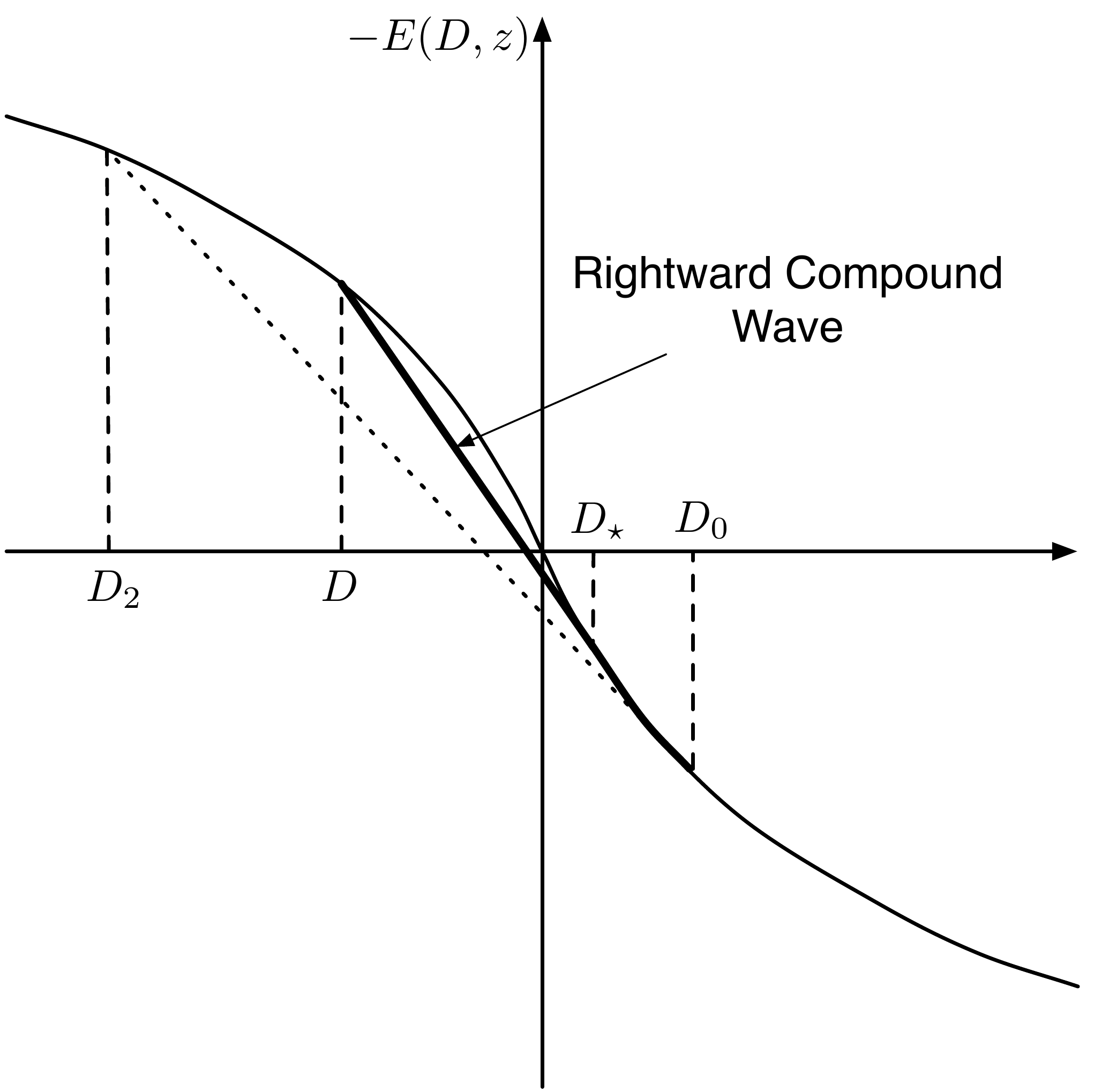}\label{f:rightward_compound_pos}
  } 

\subfigure[$0< D< D_0$]{
    \includegraphics[width=2in]{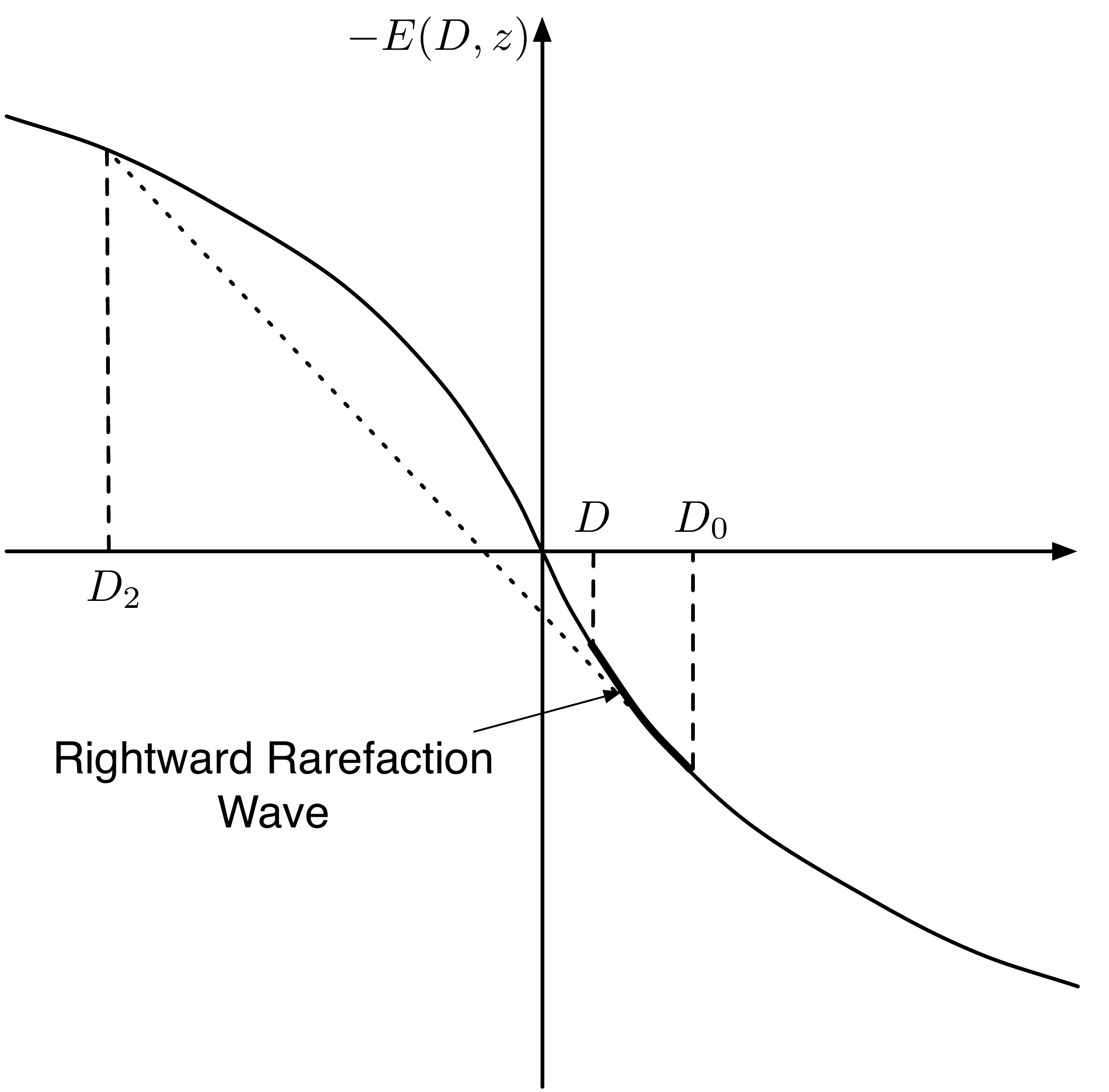}\label{f:rightward_rarefaction_pos}
  } \subfigure[$D_0 < D$]{
    \includegraphics[width=2in]{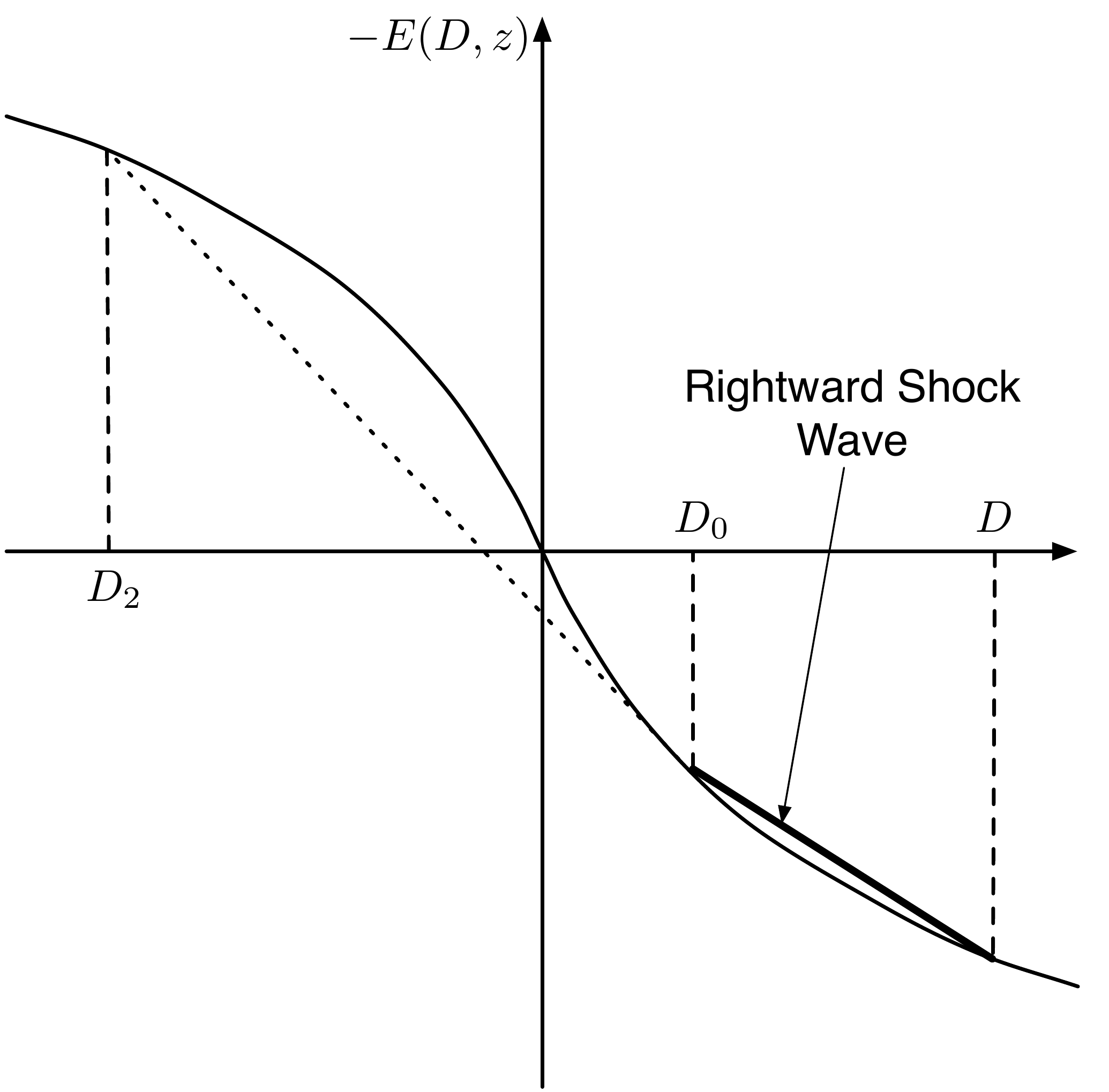}\label{f:rightward_shock2_pos}
  }
  \caption{Entropy satisfying rightward traveling waves when $D_0>0$.}
  \label{f:rightward_waves_pos}
\end{figure}

\end{document}